\documentclass[aos, preprint]{imsart}

\RequirePackage[authoryear]{natbib}

\usepackage{etoc}
\usepackage{graphicx}
\usepackage{amsmath,amsthm,amsfonts,amssymb}
\usepackage{color}
\usepackage{xcolor}
\usepackage[colorlinks=True]{hyperref}
\usepackage{bm}
\usepackage{bigints}
\usepackage{subcaption}
\usepackage{float}
\usepackage[algo2e,boxed,linesnumbered,vlined]{algorithm2e}
\usepackage{comment}
\newcommand{\e}{\mathbb{E}}
\newcommand{\p}{\mathbb{P}}
\newcommand{\lk}{\left[ }
\newcommand{\rk}{\right] }
\newcommand{\lc}{\left(}
\newcommand{\rc}{\right)}
\newcommand{\R}{\mathbb{R}}

\newcommand{\N}{\mathbb{N}}

\newcommand{\id}{\mathbf{1}}

\newcommand{\sumd}{\sum_{i=1}^d}
\newcommand{\prodd}{\prod_{i=1}^d}

\newcommand{\leqd}{1\leq i\leq d}
\newcommand{\rmd}{\mathrm{d}}

\newcommand{\T}{{\intercal}}
\newcommand{\Oo}{\mathcal{O}}
\newcommand{\Vv}{\mathcal{V}}

\newcommand{\oinf}{[0,\infty)}

\newcommand{\Rdo}{\mathbb{R}^d\setminus \{\bm 0\}}
\newcommand{\bmo}{{\mathbf{o}}}
\newcommand{\bmx}{\mathbf{x}}
\newcommand{\bmy}{\mathbf{y}}
\newcommand{\bmX}{\mathbf{X}}
\newcommand{\bmL}{\mathbf{L}}

\newcommand{\bmz}{\mathbf{z}}

\newcommand{\bma}{\mathbf{a}}
\newcommand{\bmb}{\mathbf{b}}

\newcommand{\Pkonv}{\stackrel{P}{\rightarrow}}

\newcommand{\ic}{i} %\mathsf{i}
\newcommand{\balpha}{{{\bm\alpha}}}
\newcommand{\bmc}{{{\mathbf{c}}}}

\newcommand{\bgamma}{{{\bm\gamma}}}

\newcommand{\bac}{\balpha,\bmc}
\newcommand{\bag}{\balpha,\bmc,\bgamma}

\newcommand{\lv}{\left|}
\newcommand{\rv}{\right|}
\newcommand{\HR}{H\"usler--Reiss}
\newcommand{\BA}{Barabasi--Albert}
\newcommand{\HRL}{H\"usler--Reiss L\'evy}
\newcommand{\g}[1]{\mathbf{#1}}

\newcommand{\eps}{\epsilon}%{\varepsilon}

\newcommand{\normop}[1]{{\left\vert\kern-0.25ex\left\vert\kern-0.25ex\left\vert #1 
		\right\vert\kern-0.25ex\right\vert\kern-0.25ex\right\vert}}

\theoremstyle{plain}
\newtheorem{thm}{Theorem}[section]
\newtheorem*{thm*}{Theorem}
\newtheorem{lem}{Lemma}[section]
\newtheorem*{lem*}{Lemma}
\newtheorem{rem}{Remark}[section]
\newtheorem{ex}{Example}[section]
\newtheorem{defn}{Definition}[section]
\newtheorem{cor}{Corollary}[section]
\newtheorem{prop}{Proposition}[section]
\newtheorem{ass}{Assumption}[section]

\DeclareMathOperator{\sgn}{sgn}
\DeclareMathOperator{\Var}{Var}

\DeclareMathOperator{\Cov}{Cov}

\startlocaldefs

\numberwithin{equation}{section}

\endlocaldefs

\begin{document}

\begin{frontmatter}

\title{Graph structure learning for stable processes }

\runtitle{Graphical models for stable processes}

\begin{aug}
%%%%%%%%%%%%%%%%%%%%%%%%%%%%%%%%%%%%%%%%%%%%%%%
%% Only one address is permitted per author. %%
%% Only division, organization and e-mail is %%
%% included in the address.                  %%
%% Additional information (such as           %%
%% indicating the corresponding author) can  %%
%% be included in the Acknowledgments        %%
%% section if necessary.                     %%
\author[A]{\fnms{Florian}~\snm{Brück}\ead[label=e1]{florian.brueck@unige.ch}},
\author[A]{\fnms{Sebastian}~\snm{Engelke}\ead[label=e2]{sebastian.engelke@unige.ch}}
\and
\author[C]{\fnms{Stanislav}~\snm{Volgushev}\ead[label=e3]{stanislav.volgushev@utoronto.ca}}
%%%%%%%%%%%%%%%%%%%%%%%%%%%%%%%%%%%%%%%%%%%%%%
%% Addresses                                %%
%%%%%%%%%%%%%%%%%%%%%%%%%%%%%%%%%%%%%%%%%%%%%%
\address[A]{University of Geneva, Research Institute for Statistics and Information Science, \printead[presep={,\ }]{e1,e2}}

\address[C]{Department of Statistical Sciences, University of Toronto\printead[presep={,\ }]{e3}}
\end{aug}

\begin{abstract}

We introduce Ising--\HR{} processes, a new class of multivariate Lévy processes that allows for
sparse modeling of the {path-wise} conditional independence structure between marginal stable processes with different stability indices. The underlying conditional independence graph is encoded as zeroes in a suitable precision matrix.
An Ising-type parametrization of the weights for each orthant of the L\'evy measure
allows for data-driven modeling of asymmetry of the jumps while retaining an arbitrary sparse graph.
We develop consistent estimators for the graphical structure and asymmetry parameters, relying on a new uniform small-time approximation for Lévy processes. The methodology is illustrated in simulations and a real data application to modeling dependence of stock returns.

\end{abstract}



\end{frontmatter}

\addtocontents{toc}{\protect\setcounter{tocdepth}{-100}}

\section{Introduction}

A L\'evy process $\g L = (\g L(t))_{t\geq 0}$ is an $\mathbb R^d$-valued stochastic process with independent and stationary increments that satisfies $\g L(0) = \g 0$ almost surely. 
Besides being fundamental objects appearing naturally in applied probability theory, their statistical applications range from finance \citep{Tankov2016} over physics \citep{metzler2009} and queuing theory \citep{dkebicki2015} to ruin theory in insurance \citep{asmussen2010}; see \cite{kyprianou2014} for details and more references.
By the famous Lévy--It\^o decomposition, a Lévy process can be written as $\g L = \g W +\g J$, the sum a Brownian motion $\g B$ with drift and a jump process $\g J$. While the Brownian part is very well understood and can be modeled
parametrically through the covariance matrix, the jump part is
characterized by a potentially infinite Borel measure $\Lambda$ on $\mathbb R^d\setminus \{\g 0\}$ called the L\'evy measure.
The latter describes the size distribution of the jumps of the L\'evy process and whether 
these jumps occur simultaneously in the different marginal processes.

The dependence structure of the jumps can be very complex,
especially since the L\'evy measure is allowed to have mass on all
$2^d$ orthants of $\mathbb R^d$. We enumerate the orthants in a binary way where each orthant is associated to one element of $\{-1,1\}^d=:\Oo$, with  the convention that $o_i=1$ means that the $i$th component of vectors in  orthant $\bmo \in \Oo$ is non-negative. Figure \ref{fig_intro} shows the paths and the jumps of a two-dimensional L\'evy process whose Lévy measure has mass on
the four orthants. 
\begin{figure}
    \centering
     \includegraphics[trim = 0 1.2cm 0 1cm, clip=TRUE, width=\linewidth]{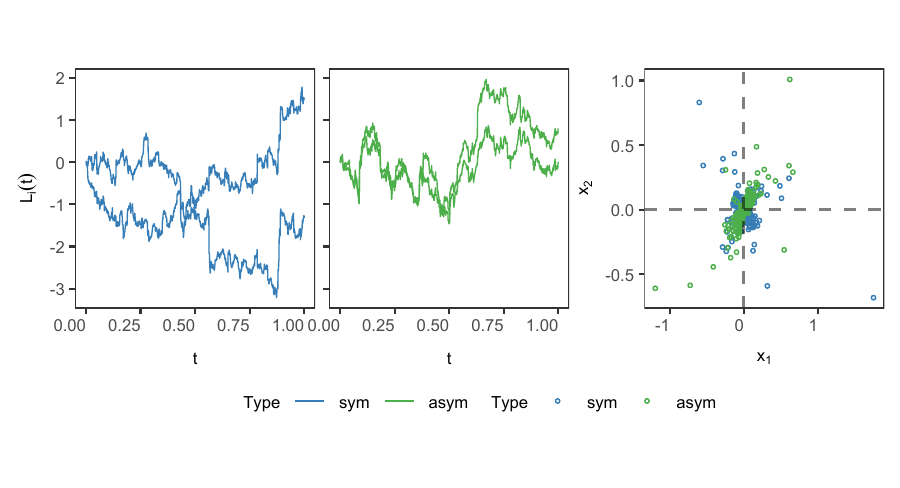}
    \caption{Simulation of a two-dimensional Ising--\HR{} L\'evy process with symmetric (left) and asymmetric (center) orthant weights. The corresponding jumps (right) of the L\'evy measure
    are shown for both the symmetric (blue) and asymmetric (green) process.
    }
     \label{fig_intro}
\end{figure}

Several works address the problem of
characterizing dependence in L\'evy measures, including the 
notion of L\'evy copulas \citep{kallsentankov2006} and Pareto L\'evy measures \citep{ederklueppelberg2012}.
Statistical models for jump processes resulting from these and other approaches in the literature are however either
simplistic with very few parameters, spectrally discrete 
\citep{misrakuruoglu2016}, or they are tractable only in very low dimensions $d$.

In classical statistics, moving to higher dimensions requires the concept of 
sparsity to build parsimonious models \citep{wainwright2008}. Conditional independence and graphical models are the most commonly used tools for this purpose \citep{lauritzen1996book}.
For sparse modeling of data from higher-dimensional stochastic processes, \cite{engelke2024levygraphicalmodels} introduce L\'evy graphical models as a L\'evy process $\g L$ that satisfies the pairwise Markov property on a graph $G =(V,E)$ with nodes $V=\{1,\dots, d\}$ and edge set $E\subseteq V\times V$, i.e.,
\begin{align}\label{pairwise_markov}
    (i,j) \notin E \quad \Rightarrow \quad L_i\perp L_j \mid \g L_{V\setminus \{i,j\}},
\end{align}
where the conditional independence is to be understood on the level
of the sample paths.
They show that
the two components $\g W$ and $\g J$ in the Lévy--It\^o decomposition can be studied separately with regard
to conditional independence properties, and we focus 
on the less-understood jump part. In the sequel, we therefore assume that there is no Brownian component and that $\g L$
is a pure jump process.
In this case, the conditional independence in~\eqref{pairwise_markov} is characterized by the factorization 
of the L\'evy measure density into lower-dimensional terms---a property that is crucial for parsimonious models.

In this work, we introduce a L\'evy process analog of Gaussian graphical models, the most widely used class of probabilistic graphical models \citep{friedmanEtAl2007, maathuis2019handbook}. Our model allows for flexible modeling of both, dependence within each orthant and asymmetry across different orthants.
Importantly, it turns out to be well-adapted to graphical structures 
in the sense of~\eqref{pairwise_markov}.
The model is inspired by the popular \HR{}
distribution in extreme value theory \citep{husler1989}, which can be parameterized by a variogram matrix $\Gamma$ or, equivalently, by its positive semi-definite precision matrix $\Theta$ \citep{hentschelengelkesegers2023}. 
While the classical \HR{} exponent measure $\Lambda_{\text{HR}}$ is a measure on the positive orthant only, we extend it 
to a flexible L\'evy measure on any Borel set $A\subset \mathbb R^d\setminus \{\g 0\}$ by
\begin{align}\label{model_def}
    \Lambda(A) = \sum_{\bmo\in \Oo} \gamma_\bmo  \Lambda_{\text{HR}} \lc \{ \bmx\in(0,\infty)^d : (o_ix_i)_{i\in V} \in A\}\rc, 
\end{align}
which can possibly have mass on all orthants; here $\gamma_\bmo\geq 0$ is the weight for orthant $\bmo$.

In the symmetric case where all weights $\gamma_\bmo$ are equal,
we can introduce sparsity on a graph $G$ in this model class through the \HR{} precision matrix: the conditional independence~\eqref{pairwise_markov} holds if and only if $\Theta_{ij} = 0$ \citep{engelkehitz2020}.  
For more general, asymmetric models, the orthant weights must additionally satisfy certain constraints 
in order to be compatible with the graphical structure~$G$.
We propose a new approach to obtain such weights based on
the Ising model from statistical physics \citep{ising1925}. 
In fact, any parameter vector $\Psi = (\psi_{i,j}: (i,j)\in E)$ with $\psi_{i,j} = \psi_{j,i} \in \R, (i,j) \in E$,
defines a valid set of orthant weights through
\begin{align*}
    \gamma_\bmo(\Psi):=\frac{2}{C(\Psi)} \exp\Bigg\{ \sum_{(i,j)\in E} \psi_{i,j}o_io_j \Bigg\},  %\label{defisingweight0}
\end{align*}
where $C(\Psi)>0$ is a normalization constant. 
The corresponding L\'evy process with L\'evy measure~\eqref{model_def} is then indeed a graphical model on $G$.
This Ising model approach allows to model asymmetries in
the dependence structure of the jumps of the L\'evy process in a
simple yet flexible way. Interestingly, it turns out that for 
a tree, i.e., a connected graph without cycles, this 
model class coincides with all possible orthant weights. For more general graphs, it is a large but strict sub-class of
all weight combinations.
The graphical model approach thus serves two purposes in this  Ising--\HR{} model class. On the one hand, it allows flexible and parsimonious modeling of the jump dependence encoded in $\Lambda_{\text{HR}}$.  On the other hand, it provides a sparse and data-driven way to define valid orthant weights; indeed, note that we only have to specify the $|E|\leq {d(d-1)/2}$ different parameters in $\Psi$ instead of all $2^d$ possible weights. While we apply the Ising--\HR{} model in the context of L\'evy processes, it can directly be used as a model for multivariate Pareto distributions in extreme value theory \citep{roo2006, Kiriliouk02012019}, which allows for dependence between extreme positive and negative values as encountered in financial applications \citep{mikosch2024extreme}.

Statistical inference for our model consists of two steps: estimation of the \HR{} parameters including the graphical 
structure; and estimation of the Ising parameters $\Psi$ on the 
estimated graph.
We first derive an estimator $\widehat \Gamma$ of the \HR{} parameter matrix $\Gamma$ and show its consistency.
A crucial ingredient for this proof, and a key theoretical contribution of our work is a uniform
small-time approximation of the L\'evy measure 
$$ \Big\vert P\{  \g L(t) \in R \} -t\Lambda(R) \Big\vert \leq Ct^2,$$
for small enough $t>0$, a constant $C>0$ and suitable rectangular sets $R\subset \mathbb R^d \setminus \{\g 0\}$.
This approximation holds for general L\'evy processes under certain regularity assumptions, and is therefore of independent interest for studying estimators for jump processes.
It is a strengthening and correction of a result in~\cite{buechervetter2013},
where our assumptions allow for L\'evy processes with mass on all orthants (as opposed to only the positive orthant), more general sets $A$ and arbitrary dimensions $d$.
We show that this approximation holds
for the Ising--\HR{} model with L\'evy measure~\eqref{model_def}.
The consistent estimator $\widehat{\Gamma}$ can directly 
be used with known structure learning methods~\citep{engelkelalancettevolgushev2022} to obtain
 a consistent
estimator $\widehat{G} = (V, \widehat{E})$ of the underlying graph structure.

In the second step, we study the weights $\Psi$ of the Ising 
parameters on the estimated graph.
While we do not have access to observations from the associated Ising model directly, we show that they can be substituted
using a moment matching method to obtain a consistent estimator $\widehat{\Psi}$. Putting both steps together, our methodology enables  learning complex, sparse dependence structures for L\'evy processes.
We illustrate the effectiveness of this approach in simulations and an application to modeling the dependence of stock returns of 16 American companies.

\section{Background}
\label{secbackground}
We first introduce some notation. 
Throughout, we let $V = \{1,\dots, d\}$ denote
the index set of the $d$ marginal processes. For integers $m \in \mathbb{N}$, we use the abbreviated notation $[m] := \{1,\dots,m\}$. Vectors are always denoted in bold letters $\g x$ and random objects are capitalized
as $\g X$. For a subset $I\subset V$ we write for the $I$-th margin of a vector $\bmx_I:=\lc x_I\rc_{i\in I}$, and use the abbreviation $\bmx_{\setminus I} = \bmx_{V \setminus I}$ with similar notation for matrices. For an orthant $\bmo \in \Oo$ as defined in the introduction and a vector $\bmx \in \mathbb R^d$, in a slight abuse of notation, we
write $\bmx \in\bmo$ if $\bmx\in\{\bmx\in\R^d\mid \bmo(\bmx)=\bmo\}$, where the orthant function is 
defined as $\bmo(\bmx)=\lc 1-2\id_{\{x_i< 0\}}\rc_{i\in V}$, i.e., the vector of sign functions of the components.
The norm $\|\bmx \|$ denotes the Euclidean $L_2$-norm if not otherwise stated. We define $B_\infty(\delta) = \{\bmx \in \R^m: \|\bmx\|_\infty < \delta\}$. For a set $A$ of finite cardinality we use $|A|$ to denote the cardinality of $A$. For two real-valued sequences $(a_n)_{n \geq 1}, (b_n)_{n \geq 1}$ we write $a_n \gg b_n$ if $b_n/a_n \to 0$ and $a_n \ll b_n$ if $a_n/b_n \to 0$.

\subsection{Lévy processes}\label{sec:levy_processes}
A $d$-dimensional L\'evy process $\g L:=(\g L(t))_{t\geq 0}$ is a càdlàg stochastic process with values in $\mathbb R^d$ that possesses stationary and independent increments; see \cite{sato} for details.
The famous L\'evy--Khintchine formula states that the characteristic functions of this process is
\begin{align*}
    \e\lk e^{ \ic\bmz^\intercal \g L(t)} \rk=\exp\left\{
    \ic t\bmz^\intercal\bm \tau +t\bmz^\intercal \Sigma \bmz/2 +
    t \int_{\Rdo} e^{\ic\bmz^\intercal \bmx} -1-\ic \bmz^\intercal \bmx \id_{\{\|\bmx\|_2\leq 1\}}  \Lambda(\rmd\bmx)  \right\},
\end{align*}
where $\bm\tau\in\mathbb R^d$, $\Sigma\in\R^{d\times d}$ is a positive semidefinite $d\times d$ matrix, and the L\'evy measure $\Lambda$ on $\Rdo$ possibly explodes at the origin but always satisfies
$$\int_{\Rdo}\min\{1,\Vert \bmx\Vert_2^2 \} \Lambda(\rmd \bmx)<\infty.$$
The L\'evy process is then said to have characteristic triplet $(\bm \tau, \Sigma, \Lambda)$.
For a subset $I\subset V$, the marginal process $\g L_I = \{(\g L(t))_I\}_{t\geq 0}$
is again a L\'evy process with characteristic triplet $(\bm \tau'_I, \Sigma_I, \Lambda_I)$,
where $\bm \tau'_I$ is not simply the sub-vector of $\bm \tau$ but can depend on other quantities. The marginal L\'evy measure $\Lambda_I$ is obtained in the usual way
by integrating out the components in $V\setminus I$ and removing $\{0\}$ from its support. Whenever the Lévy measure has a Lebesgue density $\lambda$, i.e., $\Lambda(\rmd \bmx) = \lambda(\bmx)\rmd \bmx$, we call $\lambda$ a Lévy density. 

The Lévy--It\^o decomposition shows that we can decompose the Lévy process as
\begin{align}\label{levy_ito}
    \g L = \g W +\g J,
\end{align} 
where $\g W = (\g W(t))_{t\geq 0}$ is a Brownian motion with covariance matrix $\Sigma$ and drift $\bm \tau$, and $\g J = (\g J(t))_{t\geq 0}$ is a pure jump Lévy process with characteristic triplet $(0,0,\Lambda)$, which is independent of $\g W$.
This decomposition often allows to analyze properties of $\g L$ separately for $\g J$ and $\g W$. We say that $\g L$ is a non-Gaussian Lévy process whenever $\Sigma=0$, i.e., there is no Brownian part.

The behavior of a non-Gaussian Lévy process $\g L$ is fully encoded in $\Lambda$, which contains information about both the behavior of the marginal processes  $L_i$, $i\in V$, and their dependence. Similarly to a copula approach in multivariate distributions, it is possible to separate the marginal and dependence structures via a so-called Pareto Lévy measure (PLM) \citep{ederklueppelberg2012}; a closely related concept are L\'evy copulas introduced in \cite{kallsentankov2006}.
Formally, a PLM is any Lévy measure $\Lambda^\star$ that satisfies the marginal constraints 
\begin{align}
    \label{marginal_constraint} 
\Lambda^\star(x_i>z)=\Lambda^\star(x_i<-z)=z^{-1}, \quad z>0,~i\in V.
\end{align}
It can be seen as a standardized version of a Lévy measure whose marginal processes correspond to $1$-stable processes. 

Every PLM $\Lambda^\star$ together with a set of univariate tail integrals $(U_{i})_{\leqd}$ uniquely defines a Lévy measure, see \citet[][Theorem 1(ii)]{ederklueppelberg2012}. Conversely, if the closure of the range of each univariate tail integral $U_{i}$ of a Lévy measure $\Lambda$ is $\R$, there is a unique PLM $\Lambda^\star$ that corresponds to $\Lambda$ \cite[Theorem 1(i)]{ederklueppelberg2012}. This is an analog of Sklar's copula theorem in the context of possibly infinite L\'evy measures. In the sequel, we therefore often focus on the normalized measure $\Lambda^\star$ to discuss dependence properties.

For the models in this paper, we exclude compound Poisson processes and assume that each marginal process $L_i$ has infinite activity, i.e., $\Lambda(\{\g x \mid x_i\neq 0\}) = \infty$ for all $i\in V$. This implies that the marginal distribution function $F_{ti}$ of $i$th process is continuous for all times $t>0$ \cite[Theorem 27.4]{sato}.
Consequently, the copula $C_t$ of the random vector $\g L(t)$ is unique.
\citet[][Theorem 5.1]{kallsentankov2006} show that the PLM can be approximated by small-time behaviour of the survival copulas by
\begin{align}\label{copula_approx}
    \Lambda^\star(\bmy: y_i \geq x_i : i\in V) = \lim_{t\to 0} t^{-1} P \{ F_{ti}(L_i(t)) > 1 - t / x_i: i\in  V\}, \quad \bmx \geq 0.
\end{align}
Similar approximation can be obtained for sets in other orthants.
This result enables statistical inference since it connects the dependence properties of the L\'evy measure to observations of the process itself. 

\begin{ex}\label{ex:stable}
     A PLM $\Lambda^{\star}$ is called stable if it satisfies
     \[\Lambda^\star(tA)=t^{-1}\Lambda^{\star}(A), \quad t > 0, A\subseteq \mathbb R^d.\]
     It is a Lévy measure of a multivariate $1$-stable process, i.e., a non-Gaussian Lévy process for which $\g L(t)$ has the same distribution as $ t\g L(1)+\g b(t)$ for some deterministic function $\bmb: (0,\infty) \to \R^d$ and any $t>0$.    
\end{ex}

There are several parametric models for PLMs 
that allow to fit flexible models to data. 
The most relevant example for us is the \HR{}
model, which was introduced in \cite{husler1989}
in the context of extreme value modeling.

\begin{ex}\label{ex:HR}
    The symmetric \HR{} PLM $\Lambda^\star_{(\text{sym})}$
    is defined by a symmetric, strictly
conditionally negative definite $d\times d$-matrix $\Gamma$. The model can alternatively be parameterized by its positive semi-definite precision
matrix $\Theta = (P (-\Gamma/2) P)^+$,
where $P = \text{Id} - \g 1 \g 1^\top /d$ is the projection matrix onto the orthogonal complement of $\g 1$, and $A^+$ denotes the Moore--Penrose pseudoinverse of a matrix $A$ \citep{hentschelengelkesegers2023}; note that the one-vector $\g 1 = (1,\dots, 1)$ is in the kernel of $\Theta$.
The density of the \HR{} Pareto L\'evy measure on the standardized scale is 
\begin{align}
    \lambda_{(\text{sym})}^\star (\bmx)=  c_\Theta \exp\left\{ -\frac{1}{2}\log(|\bmx|)^\T \Theta\log(|\bmx|)+\log(|\bmx|)^\T r_\Theta \right\} \prodd \frac{1}{|x|_i^{1+1/d}} , \quad \g x \in (\mathbb R\setminus\{\g 0\})^d, \label{defhüslerreissexponentmeasuredensity}
\end{align}
        with $r_\Theta=-\Theta \Gamma \id/(2d)$ and a normalizing constant $c_\Theta$ such that $\int_{\{\vert x_i\vert >1\}} \lambda_{(\text{sym})}^\star(\bmx) \rmd\bmx= 2$ for all $i\in V$. The subscript indicates that this is a symmetric L\'evy measure with the same density on all orthants. 
    The corresponding PLM $\Lambda^\star_{(\text{sym})}$ is stable 
    as defined in Example~\ref{ex:stable} and satisfies the marginal constraint~\eqref{marginal_constraint}.    
 
\end{ex}

\subsection{Lévy graphical models}

A graph $G=(V,E)$ is a collection of vertices $V = \{1,\ldots,d\}$ and edges $E\subseteq V\times V$. We work with undirected models in which $(i,j)\in E$ whenever $(j,i)\in E$. Moreover, we say that there is a path between $i$ and $j$ if there is a sequence of edges $(i_k,j_k)$ such that $i_1=i$, $j_n=j$ and $i_k=j_{k-1}$ for all $2\leq k\leq n$. A graph is connected if there is a path between every two nodes. 
A classical probabilistic graphical model for a random vector $\bmX$ links a graph $G$ to conditional independence statements through the pairwise Markov property, i.e., whenever $(i,j)\not\in E$ then $X_i\perp X_j\mid \bmX_{V\setminus \{i,j\}}$ \citep{lauritzen1996book}. 

The traditional notion of conditional independence of random vectors 
can be extended to stochastic processes in several ways \citep{didelez2008,hansen2022,matencasoloferuccimogensensalvikilbertus2025}. 
For a L\'evy process $\g L$ with characteristic triplet $(\bm\tau, \Sigma, \Lambda)$, here we work with a notion on the level of sample-paths and write
for sets $A,B,C \subset V$
\begin{align}
    \g L_A \perp \g L_B \mid \g L_C  \label{defcondind2}
\end{align}
if the (path-valued) random variables $\g L_A$ and $\g L_B$ are  independent given $\g L_C$.
\cite{engelke2024levygraphicalmodels} show that the Lévy--Ito decomposition in~\eqref{levy_ito} naturally extends to this notion since
\begin{align*}
    \g L_A \perp \g L_B \mid \g L_C \quad  \Leftrightarrow \quad \g W_A \perp \g W_B \mid \g W_C   \text{ and }   \g J_A \perp \g J_B \mid \g J_C,
\end{align*}
which allows to separately analyze conditional interdependence statements for the Gaussian part and the jump part. 
While the conditional independence statement for the Gaussian part is well-understood in terms of zeros on the precision matrix $\Sigma^{-1}$ \citep{engelke2024levygraphicalmodels}, the corresponding property of the jump part has only recently been studied.

Since the L\'evy measure $\Lambda$ characterizes the jump part $\g J$,
it must also contain information about the conditional independence
structure. Since this measure is generally not finite and supported on the non-product space $\mathbb R^d \setminus \{\g 0\}$, classical
notions of independence do not apply.  \cite{engelkeivanosvsstrokorb2022} therefore define an abstract notion of conditional independence for the L\'evy measure through classical conditional independence of random vectors 
with distribution induced by $\Lambda$ restricted to sets
with finite mass.
More precisely, for every product set $R=\times_{i=1}^d R_i$ bounded away from the origin $\g 0$ with Borel sets $R_i\subseteq \mathbb R$, associate a random vector $\g Z^R\overset{d}{=} \Lambda(\cdot \cap R)/\Lambda(R)$. 
Conditional independence for the L\'evy measure $\Lambda$ is then defined for any subsets $A,B,C \subset V$ as
\begin{align}
    A\perp_\Lambda B \mid C ~~ :\Leftrightarrow ~~ \g Z^R_A \perp \g Z^R_B \mid \g Z^R_C \text{ for all product sets }R\text{ bounded away from }\bm 0 \text{ s.t.\ } \Lambda(R)>0.  \label{defcondind3}
\end{align}
This conditional independence notion enjoys many natural properties and
is equivalent to density factorizations, whenever densities exist.
Most importantly for our work, it turns out that under mild regularity conditions, this $\Lambda$-conditional independence is equivalent to the corresponding statement
$\g J_A\perp \g J_B\mid \g J_C$ on the sample-path level of the jump process.
This equivalence allows us to rely on simple properties of the Lévy measure $\Lambda$ to study graphical models for Lévy processes.
Importantly, the conditional independence properties~\eqref{defcondind3}
are independent of the marginal processes and only depend on the corresponding
PLM $\Lambda^\star$ \citep[][Prop.~4.4]{engelke2024levygraphicalmodels}.

Based on the conditional independence notion~\eqref{defcondind3} 
for the L\'evy measure, \cite{engelke2024levygraphicalmodels} 
introduce L\'evy graphical models as a tool to model stochastic processes with sparse dependence structures.
For an undirected graph $G = (V,E)$, a L\'evy process $\g L$ with characteristic triplet $(\bm\tau, \Sigma,\Lambda)$ is a L\'evy graphical model if it satisfies the pairwise Markov property on $G$ with respect
respect to the notion~\eqref{defcondind3}, i.e., 
\[(i,j)\notin E \quad \Rightarrow \quad \{i\} \perp_\Lambda \{j\} \mid V\setminus \{i,j\},\]
or equivalently, $L_i\perp L_j \mid \g L_{V\setminus \{i,j\}}$.
\cite{engelke2024levygraphicalmodels} study the special case of tree
graphs $G$, which only have $d-1$ edges and are the sparsest type of connected
graphs. While this allows very efficient inference and structure estimation, trees are often too simplistic for more complex dependence
structures. In the next section, we introduce a sparse L\'evy graphical model for general graphs, which uses the following fundamental property of the \HR{} model from Example~\ref{ex:HR}. 

\begin{ex}\label{HR_CI}
    Let $\g L$ be a \HR{} PLM with precision matrix $\Theta$ whose L\'evy measure $\Lambda^\star_{(\text{sym})}$ is defined through its density in~\eqref{defhüslerreissexponentmeasuredensity}.
    We can then identify conditional independence statements 
    from zero entries in the precision matrix \citep{engelkehitz2020}, i.e., 
    \[\{i\} \perp_{\Lambda^\star_{(\text{sym})}} \{j\} \mid V\setminus \{i,j\} \quad \Leftrightarrow \quad \Theta_{ij} = 0.\] 
\end{ex}
Informally, this holds because factorization of the density $\lambda^\star_{(\text{sym})}$ is equivalent to its factorization into identical factors on each orthant. A formal proof is in the proof of Theorem~\ref{thmglobalmarkovprop}.

%%%%%%%%%%%%%%%%%%%
\section{\HR{} Lévy process and Ising asymmetry}
\label{secHRPRMnadprocess}

\subsection{\HR{} L\'evy processes}

The \HR{} PLM introduced in Example~\ref{ex:HR} with precision matrix $\Theta$ is symmetric around the origin and the jumps of the associated Lévy process inherit a symmetric dependence structure; see the left panel in Figure~\ref{fig_intro}. This property is too restrictive for most applications and we therefore introduce 
an asymmetric version of this model that may have different weights $\gamma_\bmo\in [0,1]$ for each orthant $\bmo\in \mathcal O$. A higher weight implies that the jumps of the Lévy process lie more frequently in this orthant.
This yields an asymmetric \HR{} PLM with density 
\begin{align}
    \lambda^\star(\bmx):= \lambda^\star_{(\text{sym})}(\bmx)\sum_{\bmo\in\mathcal O} \gamma_\bmo \id_{\{\bmx\in \bmo\}},  \quad \bmx \in\mathbb R^d\setminus \{\mathbf 0\},
 \label{defHRparetolevycop}
\end{align}
where $\lambda^\star_{(\text{sym})}$ is as in~\eqref{defhüslerreissexponentmeasuredensity} and the weights must satisfy the constraint 
\begin{align}\label{weight_norm}
    \sum_{\bmo\in\Oo, o_i=a}\gamma_\bmo=1, \quad   \text{for all  } i\in V, a\in\{-1,1\},
\end{align} 
so that the marginal constraint~\eqref{marginal_constraint} of a PLM is satisfied. The symmetric case correspond to the choice $\gamma_\bmo = 2^{-d+1}$ for all $\bmo\in\Oo$.

\begin{defn}
    \label{defnHRprocess}
    Let $\Lambda^\star$ be an asymmetric \HR{} PLM~\eqref{defHRparetolevycop} with parameter matrix $\Gamma$, precision matrix $\Theta$ and asymmetry parameters $\bgamma = (\gamma_\bmo)_{\bmo\in\Oo}$. 
    For parameters $c_i^\pm >0$ and $\alpha_i \in (0,2)$, we set the marginal integrals to 
    \begin{align*}
        U_i(x) = c_i^+ x^{-\alpha_i}\id_{\{x>0\}}-c_i^- |x|^{-\alpha_i}\id_{\{x<0\}} ,\ x\neq 0.
    \end{align*}
    
    This uniquely defines a L\'evy measure $\Lambda = \Lambda^{(\bag,\Theta)}$ with marginal tail integrals $U_i$ and 
    PLM $\Lambda^\star$. The corresponding L\'evy process $\g L$
    with L\'evy triplet $(\bm \tau,0,\Lambda)$ is called the \HR{} L\'evy process with drift $\bm\tau$, dependence parameters $\Theta$ and $\bgamma$ and marginal parameters $\bmc=(c_i^-,c_i^+)_{\leqd}$ and $\balpha = (\alpha_i)_{\leqd}$.
\end{defn}

The next result shows that the class of \HR{} L\'evy processes is closed under scaling, translation and marginalization, showing that its  definition is natural.

\begin{lem}
\label{lemHrprocessclass}
Let $\g L$ denote an \HRL{} process as in Definition~\ref{defnHRprocess} with drift $\bm\tau\in\R^d$, parameter matrix $\Gamma$ and precision $\Theta$, and asymmetry parameters $\bgamma$. Let $\bma,\bmb\in \R^d$ with $\bma \neq \bm 0$ and denote $I:=\{i\in V\mid a_i=0\}$. Then $\big( ( a_iL_i(t)+tb_i)_{i\not\in I}\big)_{t\geq 0}$ is a $(d-\vert  I\vert)$-dimensional \HRL{} process with parameter matrix $\Gamma_{\setminus I, \setminus I}$. 
\end{lem}

In general, a \HRL{} process $\g L$ is not self-similar in the sense $(\g L(at))_{t\geq 0}\overset{d}{=} (a^{H}\g L(t))_{t\geq 0}$. In fact, the latter only holds
in this model class if $\balpha=(\alpha,\ldots,\alpha)$ for some $\alpha\in (0,2)$.
We can however show that all \HR{} L\'evy processes
satisfy a generalized version of self-similarity, which will be crucial for estimation of their parameters. Here, we present a simplified version of this statement, a more general result can be found in Theorem~\ref{thmselfsimHRprocess} in the Appendix. There, we also provide an explicit expression for $\bmb(\cdot)$ which depends on the parameters $\Theta, \bm \alpha, \bm c$ of the process $\g L$. 

\begin{cor}
\label{cortimescalingHRprocess}
Let $\g L$ denote a \HRL{} process as in Definition~\ref{defnHRprocess} with drift $\bm \tau = \bm 0$. Then, for any $t > 0$
$$  
\g L(t) \overset{d}{=} \lc t^{1/\alpha_i} L_i(1) +b_i(t)  \rc_{ i \in  V},
$$
for some deterministic $\bmb(t)\in\R^d$.
\end{cor}

An important implication of this result is that the copula of $\g L(t)$ does not depend on $t$, which we will later exploit in our estimation procedures.

\subsection{Ising asymmetry}

The \HR{} L\'evy process $\g L$ in Definition~\ref{defnHRprocess} is an ideal candidate to use the concepts of sparsity and graphical modeling in the framework of L\'evy processes. Indeed, since conditional independence properties are fully
determined by the PLM regardless of the marginal processes, it follows from Example~\ref{HR_CI} that in the symmetric case
$\gamma_\bmo = 2^{-d+1}$, $\bmo\in\Oo$, the graph
structure of a \HR{} L\'evy process is encoded
in the zero pattern of its precision matrix. More precisely, for an undirected graph $G:=(V,E)$, we have that 
\[ L_i \perp L_j \mid \g L_{V\setminus\{i,j\}}, \quad \Leftrightarrow \quad \Theta_{i,j} = 0.\]
This property does no longer hold for
general, asymmetric \HR{} L\'evy processes.
In this case, a zero in the precision matrix
is no longer sufficient to imply a factorization of the L\'evy density according to the corresponding graph as the following example shows.
\begin{ex}
\label{exdensfac}

{Let $d=3$ and consider a \HR{} PLM as in~\eqref{defHRparetolevycop} with precision matrix $\Theta$;}
assume that the only $0$-entry of $\Theta$ is given by $\Theta_{1,3}=$0. {By Example~\ref{ex:HR}}, a possible graph $T=(V,E)$ for {the symmetric version $\lambda^\star_{(\text{sym})}$ as in~\eqref{defhüslerreissexponentmeasuredensity} of} this model is a tree with edge set $E=\{(1,2),(2,3)\}$. For $m_{1,3}, m_{2,3}\in (0,1)$, define the weights 
 $$\gamma_\bmo:= [ m_{1,3}\id_{\{ o_1 o_3>0\}}+(1-m_{1,3})\id_{\{o_1o_3<0\}}][ m_{2,3}\id_{\{ o_2 o_3>0\}}+(1-m_{2,3})\id_{\{o_2o_3<0\}} ],~ \bmo \in \Oo,$$ 
which can be checked to satisfy condition~\eqref{weight_norm}.
However, it is shown in Appendix~\ref{sec:proofexdensfac} that the corresponding \HR{} L\'evy process in Definition~\ref{defnHRprocess} {with those weights} does not
satisfy the conditional independence $L_1\perp L_3\mid L_2$
for any $m_{1,3}\neq 1/2$.
\end{ex}
The issue in the above example (and more generally) arises from the fact that even though the L\'evy density
$\lambda^\star(\bmx)=\lambda^\star_{(\text{sym})}\bgamma_{\bmo(\bmx)}$ factorizes according to $G$ on every separate orthant \citep{engelkehitz2020}, we cannot extend this factorization to all of $\Rdo$. The reason is that the weights $\gamma_{\bmo(\bmx)}$, seen as a function of $\bmx$, do not factorize according on $G$.
 Thus, all we know in the case $\Theta_{i,j}=0$ is that 
\begin{align*}
    \lambda^\star(\bmx)&=f (\bmx_{V\setminus j} ) g(\bmx_{V\setminus i}) \gamma_{\bmo(\bmx)} ,
\end{align*} 
for some functions $f$ and $g$.
In order to have conditional independence $L_i \perp L_j\mid \g L_{V\setminus \{i,j\}}$ it is necessary
that also $\gamma_{\bmo(\bmx)}$ factorizes according to $G$, i.e., that $\gamma_{\bmo(\bmx)}=k( \bmx_{V\setminus j} ) l(\bmx_{V\setminus i})$ for some functions $k$ and $l$ whenever $\Theta_{i,j}=0$.

The key idea to construct such graphical weight functions is to build a bridge to Rademacher distributions and Ising models.
A multivariate Rademacher vector $\g B = (B_1,\dots, B_d)$ takes values in $\{-1,1\}^d = \Oo$ and the marginals satisfy $P(B_i=1)=P(B_i=-1)=1/2$. There is a one-to-one correspondence
between distributions for $\g B$ and valid orthant weights 
that satisfy~\eqref{weight_norm}, which is given by
\begin{align}\label{ising_equiv}
    \gamma_\bmo =2P\lc\g B=\bmo\rc.
\end{align} 
Moreover, every distribution on $\g B$ that factorizes w.r.t.\ some $\Theta$ results in a $\bgamma$ which is compatible with the graph $G$.

Describing all possible distributions for $\g B$ that are compatible with $\Theta$ is difficult. We therefore restrict our attention to a subclass of distributions that is tractable and yet  flexible enough to model complex dependence between the univariate discrete margins. The Ising model originated in statistical mechanics \citep{ising1925} and is the most popular 
probabilistic model for observations in $\{-1,1\}^d$. 
It is a pairwise interaction model with probability mass function
$$ P(\g B=\bmo):= \frac{1}{C(\Psi)} \exp\lc \sum_{ i \in V} \psi_io_i +\sum_{i,j \in V} \psi_{i,j}o_io_j \rc, \quad o\in\Oo,$$
where $\Psi$ contains all marginal and dependence parameters $\psi_i\in\mathbb R$, $i\in V$, and $\psi_{i,j} \in\R$, $i,j\in V$, respectively. The constant $C(\Psi) >0$ normalizes the mass of the distribution to one. 
In order to obtain a standardized distribution for $\g B$ that satisfies $P(B_i=1)=P(B_i=-1)=1/2, i \in V$ it is well known that we 
need to set $\psi_i = 0$ for all $i\in V$. Furthermore, a direct computation shows that a factorization $P(\g B=\bmo) = p(\bmo_{V\setminus j})q(\bmo_{V\setminus j})$
for suitable functions $p$ and $q$ is equivalent to $\psi_{i,j} = 0$. 

In the sequel, we assume that all marginal parameters satisfy $\psi_i =0$, $i\in V$, and then drop them in $\Psi$. We then always consider the parameters of an Ising model as symmetric matrix $\Psi \in \R^{d\times d}$ with entries $\psi_{i,j}$ that satisfy $\psi_{i,i} = 0$.
According to~\eqref{ising_equiv}, we can thus define a vector of weights $\bgamma(\Psi)$ that is compatible with a given $\Theta$ and the corresponding graph $G = (V,E)$, by setting $\psi_{i,j} = 0$ for all $(i,j)\notin E$ and 
\begin{align}
\gamma_\bmo(\Psi):=\frac{2}{C(\Psi)} \exp\lc \sum_{(i,j)\in E} \psi_{i,j}o_io_j \rc, \quad \psi_{i,j}\in\R, (i,j)\in E \label{defisingweight}
\end{align}

We refer to such weights as \emph{Ising weights} and to the corresponding \HR{} L\'evy process $\g L$ constructed via Ising weights as \emph{Ising--H\"usler--Reiss (IHR) process}. It can be checked that the class of IHR processes is closed under scaling and translation. However, their multivariate marginal distributions are generally only \HR{} L\'evy processes but not IHR processes, since the Ising models are not closed under marginalization.
 
We have the desirable property that conditional independence statements of the \HR{} L\'evy processes are entirely encoded in $\Theta$ and $\Psi$.

\begin{thm}
\label{thmglobalmarkovprop}
For an IHR process $\mathbf L$ with \HR{} precision matrix $\Theta$ and Ising weights $\Psi$ we have 
\[ 
L_i\perp L_j\mid \g L_{V\setminus \{i,j\}} \quad \Leftrightarrow \quad \big(\Theta_{i,j}=0~~~\mbox{and}~~~ \psi_{i,j} = 0 \big).
\]
\end{thm}

\begin{rem}
\label{remtrees}
Suppose that the graph $T=(V,E)$ is 
a tree, i.e., a connected graph without
cycles.
Further let $\g L$ be an IHR process
with \HR{} parameter matrix $\Gamma$,  precision $\Theta$ and Ising weights $\Psi$ satisfying 
the global Markov property on the tree $T$.
As shown in \cite{engelkevolgushev2022}, the tree structure induces a particularly simple form of the 
parameter matrix, namely it is a tree metric
$$ \Gamma_{i,j}=\sum_{(k,l)\in \text{path}(i,j)} \Gamma_{k,l},$$
where $\text{path}(i,j)$ denotes all  edges on the undirected path from $i$ to $j$.

The Ising weight of orthant $\bmo\in \Oo$ in~\eqref{defisingweight} also takes a special form for trees, namely
$$\gamma_\bmo(\Psi)=\prod_{(i,j)\in E}{m}_{i,j}^{\id\{o_io_j>0\}}(1-{m}_{i,j})^{\id\{o_io_j<0\}}, \quad    {m}_{i,j}=\frac{\exp(2\psi_{i,j})}{1+\exp(2\psi_{i,j})}\in(0,1).$$   
Interestingly, this form of weights represents \emph{all}
valid orthant weights in the sense of~\eqref{weight_norm} in the case of tree graphical models; see~\citet[][equation~(28)]{engelke2024levygraphicalmodels}.
Therefore, Ising weights can be seen as a natural generalization of tree orthant weights for Lévy processes on general graphs; note that in the latter case, the Ising weights~\eqref{defisingweight} are just a sub-class of all valid orthant weights.
\end{rem}

\section{Estimation of Ising--Hüsler--Reiss processes}
\label{secestimation}

In this section we develop estimation methods for the parameters of an IHR process, which can naturally be divided into the following three subproblems
\begin{enumerate}
    \item[(i)] estimation of the graph $G$;
    \item[(ii)] estimation of the parameters of $\Theta$ that are compatible with the estimate of $G$;
    \item[(iii)] estimation of Ising parameters $\Psi$ that are compatible with the estimate of $G$.
\end{enumerate}
Steps (i) and (ii) will be treated jointly, whereas we devote a separate subsection to step (iii). Before discussing our suggested estimation methods we need an additional technical result, which is of independent interest and is essential to obtain the consistency of our estimators.

\subsection{A concentration inequality for the small time behavior of Lévy processes}\label{sec:levy_approx}
The basis for proving the consistency our estimation procedures is a fundamental result about the small time behavior of a Lévy process, which is an extension and correction of \citet[][Lemma 4.3]{buechervetter2013}; see Section \ref{app_counterexample} for additional details. Its proof can be found in Appendix~\ref{app_proof_univ_conv_levy_measure} and is one of the main theoretical contributions of this paper. Define 
\begin{align}\label{S_eps}
    S_\epsilon:=\{ c_\epsilon: \R^d\to[0,1]\mid c_\epsilon\text{ smooth and }\id_{\{\bmx \in B_\infty(\epsilon/2)\}} \leq c_{\epsilon}(\bmx)\leq \id_{\{ \bmx\in B_\infty(\epsilon)\}} \},
\end{align}
where we say that a function is smooth if it is infinitely differentiable. We need a regularity condition on the Lévy density  to derive our bound on the small time behavior of the corresponding Lévy process.

\begin{ass}
\label{assmpsmootheneslevydens}
Let $\lambda$ denote the density of the Lévy measure $\Lambda$ of a Lévy process $\g L$. Assume that for every $\epsilon>0$ there exists a $c_\epsilon\in S_\epsilon$ such that $h^{(\epsilon)}:=(1-c_\epsilon) \lambda$ satisfies the following conditions. 
      
When $d=1$ assume that $(1-c_\epsilon) h^{(\epsilon)}$ is bounded and continuously differentiable with bounded derivative. When $d\geq 2$ assume that
    \begin{enumerate}
    \item[(i)] $h^{(\epsilon)}$ has uniformly bounded univariate and bivariate margins $( h^{(\epsilon)}_i)_{\leqd}$ $( h^{(\epsilon)}_{i,j})_{1\leq i\not=j\leq d}$ where $h^{(\epsilon)}_i(x_i) := \int_{\R^{d-1}} h^{(\epsilon)}(\bmx) d\bmx_{\backslash i}$ and $h^{(\epsilon)}_{i,j}(x_i,x_j) := \int_{\R^{d-2}} h^{(\epsilon)}(\bmx) d\bmx_{\backslash \{i,j\}}$.
    \item[(ii)] $h^{(\epsilon)}$ is continuous  and $x_i\mapsto \partial_i h^{(\epsilon)}(\bmx)$ is continuous for all $\leqd$.
    \item[(iii)] For every $a \in \R$ and $j=1,\dots,d$ there exists $\sigma(a,j)>0$ and a function $M_{j,a,\sigma}: \R^{d-1} \to [0,\infty)$  such that $\int_{\R^{d-1}} M_{j,a,\sigma}(\bmx_{\setminus j}) \rmd\bmx_{\setminus j} < \infty$ and for all $1\leq i\not=j\leq d$  $y_i \mapsto \int_{\R^{d-2}} M_{j,a,\sigma}(\bmx_{\setminus j}) \rmd \bmx_{\setminus  i,j}$ is bounded on bounded sets and
    \[
    \sup_{|x_j-a|\leq \sigma} \Big|\partial_j h^{(\epsilon)}(x_1,\ldots,x_d) \Big| \leq M_{j,a,\sigma}(\bmx_{\backslash j}).
    \]
    \item[(iv)] For all $1\leq j\leq d$ $x_j \mapsto \int_{\R^{d-1}} \lv\partial_j h^{(\epsilon)}(\bmx)\rv \rmd\bmx_{\setminus j}$ is uniformly bounded.
    \end{enumerate}
    
\end{ass}

Essentially, we need Assumption~\ref{assmpsmootheneslevydens} to ensure that for all $I\subset  V$ the functions of the type $\R^{\vert I\vert}\ni\bmx_I \mapsto\int_{\bmx_I\geq \bmy_I} \Lambda (\rmd \bmy)$ are twice continuously differentiable with bounded derivatives outside a neighborhood of $\bm 0\in\R^{\vert I\vert }$, which will be an essential ingredient to the proof of Theorem \ref{thmunifconvlevymeasure} below. While Assumption~\ref{assmpsmootheneslevydens} may seem quite technical, it is surprisingly difficult to formulate simple regularity conditions directly in terms of $f$ which ensure differentiability and boundedness of terms of the form $x_i\mapsto U_i(x_i)=\sgn(x_i) \int_{T(x_i)} \int_{\R^{d-1}} \lambda(\bmx)\rmd\bmx_{\setminus i}$ for $x_i\in \R\setminus \{0\}$. When smoothing out the origin via $(1-c_\eps)\lambda$, the resulting regularity conditions are still non-trivial, but can be summarized relatively concisely via Assumption~\ref{assmpsmootheneslevydens}. To further illustrate the issue, Appendix \ref{app_proof_univ_conv_levy_measure} provides an example of a continuous bivariate Lévy density $\lambda$ whose marginal Lévy densities are not continuous on $\R\setminus \{0\}$. 

Lévy densities which have a singularity along the axes, i.e., $\lim_{x_i\to 0}\lambda(\bmx)=\infty$ for $\bmx_{ V\setminus i} \neq \bm 0$, cannot satisfy Assumption~\ref{assmpsmootheneslevydens}, since $x_i\mapsto \partial_i h^{(\epsilon)}(\bmx)$ cannot exist for small enough $\eps$. Thus, if $\lambda$ satisfies Assumption~\ref{assmpsmootheneslevydens}, it can only explode towards the origin. An example of a Lévy density that satisfies Assumption~\ref{assmpsmootheneslevydens} is the Lévy density of a \HRL{} process, as is verified in the following Proposition. The proof of this result relies on subtle analytic bounds which are necessary due to the complicated form of the \HR{} density and its partial derivatives.

\begin{prop}
\label{propunifconvHRprocess}
The \HRL{} process from Definition~\ref{defnHRprocess} satisfies Assumption~\ref{assmpsmootheneslevydens} for every $d\in\N$, $\balpha\in(0,2)^d$, $\bm \tau \in \R^d$ $\bmc\in(0,\infty)^{2d}$, $\bgamma\in [0,\infty)^{2^d}$ and any valid \HR{} precision matrix $\Theta$.   \end{prop}

We next show that that Assumption~\ref{assmpsmootheneslevydens} is sufficient to uniformly control the small time behavior of a multivariate Lévy process on rectangles, i.e., sets of the form $R=\times_{i=1}^m J_m$ where each $J_m\subseteq \R$ is an interval, which then allows us to uniformly control the small time behavior of the \HRL{} process by Proposition \ref{propunifconvHRprocess}.

\begin{thm}
\label{thmunifconvlevymeasure}
    Let $\g L$ denote a $d$-dimensional Lévy process with Lévy measure $\Lambda(\rmd \bmx)=\lambda(\bmx)\rmd\bmx$, drift $\bmb$ and covariance matrix $\Sigma$. Assume that $\lambda$ satisfies Assumption~\ref{assmpsmootheneslevydens}. 
    Then, for every $\delta>0$, there exists a constant $C$ which only depends on $\delta$ but does not depend on $t$ such that for all rectangles $R \subset B_\infty^\complement(\delta)$ and $t>0$ we have
    $$ 
    \Big\vert P\lc  \g L(t) \in R \rc -t\Lambda\lc R \rc \Big\vert \leq Ct^2.
    $$
\end{thm}
  
Note that Theorem~\ref{thmunifconvlevymeasure} implies the same uniform bound for the small time behavior of all $I$-margins $\g L_I$ by setting $J_i=\R$ for all $i\not\in I$. 

\begin{rem}
It is well known, see e.g.\ \cite[Corollary 8.9]{sato}, that the small time-behavior of a Lévy process $\bm L$ ``reveals'' the Lévy measure $\nu$ via
\[
t^{-1} \e\lk g\big(\g L(t)\big)\rk \overset{t\to 0}{\longrightarrow} \int_{\mathbb{R}^d} g(x)\,\Lambda(\rmd x) 
\]
for every bounded continuous function $g:\R^d\to \R$ which vanishes in a neighborhood of~$\bm 0$. Theorem \ref{thmunifconvlevymeasure} above yields an analog of this result when $g$ is the indicator function of a rectangular set $R$ which is bounded away from $\bm 0$, additionally quantifying the rate of the convergence to be of order $t$ uniformly over all $R\subset B_\infty^\complement(\delta)$.
\end{rem}

\subsection{Estimation of the graph structure}

Throughout this section, we will work with the following observation scheme.
\begin{ass}\label{ass:obsscheme}
Assume that we observe a \HRL{} process $\g L$ with drift $\bm\tau$ and parameters $(\Theta,\bag)$ on $n$ regularly spaced grid points $\Delta, 2\Delta, \dots, n\Delta$ with $\Delta > 0$. Denote the corresponding variogram matrix by $\Gamma$ and assume $\Gamma_{i,j} > 0~ \forall i \neq j$. Let 
\begin{align}\label{levy_samples}
\g D_s := \g L(s\Delta)-\g L( \{s-1\}\Delta), \quad s \in [n].
\end{align} 
\end{ass}
The stationary and independent increments property of Lévy processes implies that $\g D_s$ are i.i.d. across $s \in [n]$. The asymptotic regime that we will study in what follows is $n\to\infty$ for fixed $\Delta > 0$, that is, we observe the process on a growing time grid.

We will first discuss the estimation of the variogram matrix $\Gamma$ based on the sample $\g D_s, s \in [n]$. To this end, we will rely on ideas from extreme value theory. The dependence structure of a random vector $\bmX = (X_j: j\in  V)$ with continuous marginal distribution functions $F_i$ is \emph{multivariate regularly varying with exponent measure $\nu$} if {$\nu$ is a Borel measure on $[0,\infty)^d$ such that} 
\begin{align}\label{MRV}
{\lim_{t\to 0} t^{-1} P\big(  F_i(X_i) > 1 - t/x_i ~\forall i\in  V\big) = \nu(\{\bmy \mid y_i > x_i ~\forall i\in  V\}) }
\end{align}
holds for all {$\bmx \in [0,\infty)^d\setminus\{\bm 0\}$ that are continuity points of $\g x \mapsto \nu((\bmx,\bm\infty))$} \citep[e.g.,][]{segers2020}; note the similarity to the small-time approximation of L\'evy processes in~\eqref{copula_approx}. Our first step is to establish multivariate regular variation of certain random vectors constructed from $\g L(\Delta)$ and provide a quantification of the corresponding convergence rate for certain values of $\g x$. 

Fix $I \subseteq  V$ with $|I| \geq 2$. For every orthant
$\bmo\in\{-1,1\}^{\vert I\vert}$, we define the random vector $\g L_I^\bmo(\Delta)$ with law
$$
\g L_I^\bmo(\Delta) \overset{d}{=} \g L_I(\Delta)\mid \{ \g L_I(\Delta)\in \bmo\};
$$
note that the probability of the conditioning event is positive by Lemma \ref{lem:posorthantprob}. 

\begin{thm}\label{thm:secondorder}
Let Assumption~\ref{ass:obsscheme} hold for some $d \geq 3$. Define the marginal distribution functions of $|\g L^\bmo(\Delta)|$ (with absolute value interpreted component-wise) as $F_i^{+,\bmo}(z)=P\lc |L^\bmo_i|\leq z\rc$. Then, for any $0 < \zeta < 1 \wedge \min_{i\in  V}\alpha_i^{-1}$,
    there exist $C>0$ such that for all $\bmx\in[0,1]^d$
	\begin{multline*} 
		\mathcal{D}_I(\bm x,q) := \Big|  P\lc F_i^{+,\bmo} \lc |L^\bmo_i(\Delta)|\rc >1-qx_i\ \forall\ i\in I\rc
		\\
		- q \Lambda_{\rm{HR}}^{(I)}\big(\{\bmy\in[0,\infty)^{|I|}\mid  \bmy\geq 1/\bmx_I \}\big) \Big|  \leq C q^{1+\zeta}
	\end{multline*} 
	for all $I\subseteq V$ with $\vert I\vert \in\{2,3\}$, where $\Lambda_{\rm{HR}}^{(I)}$
	denotes the exponent measure of the H\"usler--Reiss distribution with variogram matrix $\Gamma_{I,I}:=(\Gamma_{i,j})_{i,j\in I}$. 
\end{thm}

The bound above corresponds to a \emph{second order condition} in extreme value theory \citep{DF2006}; this result is key for constructing an estimator for $\Gamma$ and prove its consistency. Further, the result is of independent interest beyond the analysis of extremal variograms, since it allows one to study general M-estimators for parameters of a \HRL{} process that are based on ideas from extreme value theory. Note that conditioning on $\g L_I(\Delta)\in \bmo$ is crucial to obtain the correct exponent measure. The random vector $\g L_I$ itself is not multivariate regularly varying with the correct exponent measure, and taking absolute values does not solve this problem either. 

Next, we utilize ideas from the estimation of 
\HR{} Pareto distributions \citep{engelkelalancettevolgushev2022,hentschelengelkesegers2023} in order to construct estimators for $\Gamma$ and subsequently $\Theta$. Consider a sequence $q_n \to 0$ with $q_n \in (0,1)$ for all $n$. For any fixed $i\neq j,m \in  V$ let $J = (i,j,m)$. For $\ell \in \{i,j,m\}$ and an orthant $\bmo \in \{-1,1\}^{3}$, let $\widehat F_\ell^{\bmo,J}$ denote the empirical distribution function of the sample $\{D_{s,\ell} \mid s \mbox{ such that } \g D_{s,J} \in \bmo\}$, where we explicitly allow for $i=m$ or $j=m$ which results in $D_{s,J}$ having repeated entries. Define $n^{J,\bmo} := |\{s \in [n]: \g D_{s,J} \in \bmo\}|$ and let
\[
\mathcal{S}^{J,\bmo}_m := \Big\{ \log\big[ \tfrac{n^{J,\bmo}+1}{n^{J,\bmo}} - \widehat{F}_i^{\bmo,J}(D_{s,i}) \big] - \log\big[ \tfrac{n^{J,\bmo}+1}{n^{J,\bmo}} - \widehat{F}_j^{\bmo,J}(D_{s,j}) \big] \ \Big\vert\ \widehat{F}_m^{\bmo,J}(D_{s,m}) > 1 - q_n \Big\}
\]
denote the samples of $\log\big[ \tfrac{n^{J,\bmo}+1}{n^{J,\bmo}} - \widehat{F}_i^{\bmo,J}(D_{s,i}) \big] - \log\big[ \tfrac{n^{J,\bmo}+1}{n^{J,\bmo}} - \widehat{F}_j^{\bmo,J}(D_{s,j}) \big]$ for which $\widehat{F}_m^{\bmo,J}(D_{s,m}) > 1 - q_n$. 
An estimator for $\Gamma_{ij}$ is now constructed as follows 
\begin{align}\label{defvarioest_incr}
\widehat{\Gamma}_{ij}^{(m,\bmo)} := \widehat{\mathrm{Var}}(\mathcal{S}^{J,\bmo}_m),
\end{align}   
where and $\widehat{\Var}(\mathcal{S}^{J,\bmo}_m)$ is the empirical variance of the sample $\mathcal{S}^{J,\bmo}_m$ with scaling $(|\mathcal{S}^{J,\bmo}_m|+1)^{-1}$ where $|\mathcal{S}^{J,\bmo}_m|$ denotes the number of observations in the set $\mathcal{S}^{J,\bmo}_m$. This estimator is motivated by the empirical variogram of~\cite{engelkevolgushev2022}. Note, however, that we have a possibly different sample for each orthant, so the matrix $(\widehat{\Gamma}_{ij}^{(m,\bmo)})_{i,j \in  V}$ does \emph{not} correspond to the empirical variogram matrix from~\cite{engelkevolgushev2022}. As mentioned above, conditioning on orthants first is crucial to obtain valid estimators of $\Gamma_{i,j}$, and a naive application of the ideas from~\cite{engelkevolgushev2022} does not lead to consistent estimators.  

Since for any fixed $i,j$ the quantity $\widehat{\Gamma}_{ij}^{(m,\bmo)}$ is an estimator for the same population parameter $\Gamma_{i,j}$, we construct a combined estimator by considering all possible values of $m, \bmo$ and define 
\begin{align}
    \widehat{\Gamma}_{i,j}:=\frac{1}{d} \sum_{m=1}^d \sum_{\bmo\in \{-1,1\}^{3}} \frac{n^{J,\bmo}}{n} \widehat{\Gamma}^{(m,\bmo)}_{i,j} \label{defestimatorvariogrammatrix},
\end{align}
implicitly setting $\widehat{\Gamma}^{(m, \bmo)}_{i,j}=0$ whenever there is no observation available to compute it, or when $i=j$. Note that $\widehat{\Gamma}^{(m,\bmo)}_{i,j}$ depends on $q_n$, and thus so does $\hat{\Gamma}$. Repeating this for all pairs $i,j\in V$ yields and
estimator $\hat{\Gamma}$ of the \HR{} parameter matrix $\Gamma$ of the IHR process $\g L$. Our first fundamental result shows consistency of this estimator
\begin{thm}
\label{thmconsitencygammaestimator}
Let Assumption~\ref{ass:obsscheme} hold and consider the estimator $\widehat \Gamma$ from~\eqref{defestimatorvariogrammatrix}. Assume that $q_n\to 0$ with $1/2\geq q_n \gg n^{\xi-1}$ for some $\xi \in(0,1)$. Then $\widehat{\Gamma}\Pkonv \Gamma$.
\end{thm}

Consistent estimation of the variogram $\Gamma$ is the
fundamental building block for a data-driven 
recovery of the graph structure. 
Indeed, the \texttt{EGlearn} algorithm in \citet[][Algorithm 1]{engelkelalancettevolgushev2022} for estimation of extremal graphical models {can be viewed as generic algorithm that takes as input an estimator of a variogram matrix and outputs an estimate of the sparsity pattern of the corresponding \HR{} precision matrix $\Theta$. This directly yields an estimator $\widehat G = (V, \widehat E)$ of the graphical structure of the IHR process by Theorem~\ref{thmglobalmarkovprop}, provided that the Ising parameters $\Psi$ satisfy $\psi_{i,j} = 0$ for all $i,j$ such that $\Theta_{i,j} = 0$.}

Given the estimator of the graph structure, \cite{hentschelengelkesegers2023} construct an estimator $\widehat \Gamma^{\widehat G}$
whose precision matrix $\widehat \Theta^{\widehat G}$
is graph structured according to $\widehat G$.
To this end, they solve the matrix completion problem
\begin{align*}
\begin{alignedat}{2}
    \widehat \Gamma^{\widehat G}_{ij} &= \widehat \Gamma_{ij} &\quad& \forall (i,j)\in \widehat E,\\
    \widehat \Theta^{\widehat G}_{ij} &= 0                   &\quad& \forall (i,j)\notin \widehat E.
\end{alignedat}
\end{align*}
We will use the same approach in our framework of L\'evy graphical models.

{\cite{engelkelalancettevolgushev2022} describe two base learners for the \texttt{EGlearn} algorithm: neighborhood selection \cite[Section 4.3.1]{engelkelalancettevolgushev2022} and graphical lasso \cite[Section 4.3.2]{engelkelalancettevolgushev2022}. In their simulations, neighborhood selection consistently outperforms graphical lasso in terms of graph recovery, and we observe the same in our simulations. For this reason, we only provide theory on consistency of neighborhood selection for graph estimation below and only consider graphical lasso in the simulations.}
To state the next result, we need to introduce some additional notation. For $m \in  V$, let $\Sigma^{(m)} := [\Theta_{\setminus m, \setminus m }]^{-1}$ denote the inverse of the \HR{} precision matrix $\Theta$ from which the $m$-th row and column have been removed. Let $S_{m,\ell} := \{j \in  V, j \neq m,\ell: \Theta_{\ell j} \neq 0 \}$ and define the neighborhood selection incoherence parameter through 
\[
\eta := 1 - \max_{m,\ell \in  V, m \neq \ell} \normop{\Sigma^{(m)}_{ S_{m,\ell}^\complement,S_{m,\ell} } [\Sigma^{(m)}_{S_{m,\ell},S_{m,\ell}}]^{-1} }_\infty
\]
where the complement of $S_{m,\ell}$ is taken in $ V\setminus \{m\}$ and $\normop{A}_\infty := \max_{j \in  V}\sum_{i=1}^d |A_{ij}|$.

\begin{cor}
\label{corsparsistentesttheta}
Let Assumption~\ref{ass:obsscheme} hold and denote by $G$ the graph corresponding to $\Gamma$. Denote by $\widehat G$ the graph estimator from the \emph{EGlearn} algorithm with neighbourhood selection with penalty parameter $\rho_n$ based on the variogram estimator $\widehat{\Gamma}$ from~\eqref{defestimatorvariogrammatrix}. Assume that $\eta > 0$ {and that, for some $1 > \xi >0$ and $ 0 < \zeta < 1 \wedge \min_{i\in  V}$$ \alpha_i^{-1}$  
\[
1 \gg q_n \gg n^{\xi-1},\quad 1 \gg \rho_n \gg q_n^\zeta (\log q_n)^2 + (nq_n)^{-1/2}.
\]
} 
Then we have $(\widehat{G},\widehat{\Theta}^{\widehat{G}}) \Pkonv ( G,\Theta)$. 
\end{cor}

\subsection{Estimation of Ising weights}

Given a \HR{} precision matrix $\Theta$ and the corresponding graph $G$, it remains to estimate compatible Ising weights. We first assume that we know $G$ and later replace it with an estimated graph. Let $\g B$ denote the Rademacher random vector linked to the Ising weights through~\eqref{ising_equiv}.
Recall that for a IHR L\'evy process on graph $G$, the Ising parameters need to satisfy $\psi_{i,j}=0$ if $(i,j)\not \in E$. 
The Ising model is an an exponential family as its likelihood in~\eqref{defisingweight} can be represented as 
$$ 
P(\g B=\bmo)={\gamma_\bmo(\Psi)}/{2}=\exp\{ \langle\Psi, S(\bmo)\rangle  -A(\Psi)\}, 
$$
where $S(\bmo)=\lc o_i o_j\rc_{(i,j)\in E}$ is the sufficient statistic and $A(\Psi)=\log\{C(\Psi)\}$ is a normalizing constant.
Since this only depends on the bivariate interactions $o_io_i$, intuitively, one should be able to estimate $\psi_{i,j}$ via the observations of $(L_i,L_j)$. The maximum likelihood estimator (MLE) of the Ising model with given edge set $E$ is unique when it exists and in the latter case can be found by setting the gradient of its log-likelihood to $0$. The MLE for a collection of i.i.d. samples $\g B^{(s)}, s \in [n]$ exists iff the following system of equations
\begin{align}\label{est_eq}
    \e_{\Psi}\lk B_iB_j\rk - n^{-1} \sum_{s=1}^n B_i^{(s)} B_j^{(s)} = 0 \quad (i,j)\in E ,
\end{align} 
has a solution where $\e_{\Psi}$ denotes expectation with respect to an Ising model with parameter matrix $\Psi$; see \citet[][Section 17.4.1]{hastie2009} for a detailed derivation. Note that $ n^{-1} \sum_{s=1}^n B_i^{(s)} B_j^{(s)} - \e_{\Psi}\lk B_iB_j\rk$ is one half times the partial derivative of the Ising log-likelihood functions based on the observations $\g B^{(s)}, s \in [n]$, evaluated at the parameter $\Psi$ with respect to $\Psi_{i,j}$. 

This estimator only depends on the bivariate marginal distributions and one can show that the log-likelihood of the Ising model is concave in $\Psi$, which also implies that the MLE is unique when it exists. In our setting of IHR processes, we do not have access 
to observations from the Ising model. 
Instead, we note that the right-hand side of~\eqref{est_eq} is a consistent estimator of $\Cov_\Psi(B_i,B_j) = \mathbb E_{\Psi}\lk B_iB_j\rk$ since $\e_\Psi\lk B_i\rk=0$ in our model class.
We derive a different estimator for
the covariance that can be computed from
observations as in~\eqref{levy_samples} of the L\'evy process.
While a bivariate margin of an Ising model is generally 
not an Ising model itself, we can express the covariance
in terms of the L\'evy measure.

\begin{lem}
\label{lemIsingcorrviaPLM}
    Let $\Lambda^\star$ be the PLM of a IHR process with Ising weights $\Psi$. Then
    \begin{align}\label{cov_rep}
        \Cov_\Psi(B_i,B_j) =\frac{{\chi}_{i,j}^{(1,1)}+{\chi}_{i,j}^{(-1,-1)}- {\chi}_{i,j}^{(1,-1)}-{\chi}_{i,j}^{(-1,1)}}{{\chi}_{i,j}},
    \end{align}
where $\chi^{(o_1,o_2)}_{i,j}:=\Lambda^\star \lc   o_1 x_i >1, o_2 x_j >1 \rc$,
which are related to the Lévy correlation $\chi_{i,j}$ in 
\cite{engelke2024levygraphicalmodels} by
\[ \chi_{i,j} :=2^{-1}\sum_{(o_i,o_j)\in \{-1,1\}^2} \chi^{(o_i,o_j)}_{i,j}.\]   
\end{lem}

A natural estimator for the four components of the
L\'evy correlation is
\begin{align}
    \widehat{\chi}^{(o_1,o_2)}_{ij} = \frac{1}{k} \sum_{t=1}^{n} \mathbf{1} \left\{ 
        o_1 \widehat{F}_i(\Delta_i) -\id_{\{o_1=1\}}  + \frac{k}{2n}>0,\ 
        o_2 \widehat{F}_j(\Delta_i) -\id_{\{o_2=1\}}  + \frac{k}{2n}>0
        \right\},\label{chi_est}    
\end{align}
where $\widehat F_i$ is the empirical distribution function of the increments  $\Delta_i(s) = L_i(s\Delta) - L_i( \{s-1\}\Delta)$ of the L\'evy process, for any $i\in V$.
\citet[][Theorem 6.4]{engelke2024levygraphicalmodels} show consistency of this estimator for a suitable choice of intermediate sequence $k=k(n)$ such that $k\to \infty$ and $k/n\to 0$, as $n\to\infty$. In order to obtain an estimator $\widehat \Psi{(E)}$ of the Ising weights $\Psi$, we consider the system of equations
\begin{align}\label{est_eq2}
    \e_{\Psi{(E)}}\lk B_iB_j\rk - \widehat{\Cov}_{\Psi{(E)}}(B_i,B_j) = 0,  \quad (i,j)\in E ,
\end{align} 
where $\widehat{\Cov}_{\Psi{(E)}}(B_i,B_j)$ denotes the plug-in estimator of~\eqref{cov_rep} using~\eqref{chi_est} for $(o_1,o_2)\in \{-1,1\}^2$. This system of equations may not always have a solution in finite samples, but we will show in the proof of Theorem~\ref{thmchiestfixedgraph} that it has a solution with probability going to one. We further show that this provides a consistent estimator of the Ising weights, even when the graph is unknown and needs to be inferred from the data. For the purpose of the following theorem, let $\widehat{\Psi}{(\widehat{E})}$ denote the solution of~\eqref{est_eq2} when this solution exists and set $\widehat{\Psi}{(\widehat{E})} \equiv 0$ otherwise.

\begin{thm}
\label{thmchiestfixedgraph}
{Let Assumption~\ref{ass:obsscheme} hold and denote by $G = (V,E)$ the graph corresponding to $\Gamma$. Assume additionally that $\g L$ from  Assumption~\ref{ass:obsscheme} is an IHR process with Ising parameters $\Psi$ and that $\psi_{i,j} = 0$ for all $(i,j) \not\in E$. Denote by $\widehat G = (\widehat E,V)$ the graph that is estimated by the \texttt{EGlearn} algorithm and assume that the conditions of Corollary~\ref{corsparsistentesttheta} hold.} {Consider the estimator $\widehat{\Psi}{(\widehat{E})}$ described before this theorem and assume $k \to \infty, k/n \to 0$.} The Ising weights are consistent, i.e., for all $\bmo\in \Oo$ we have
\[
\hat{\gamma}_\bmo:= \frac{2}{C(\widehat{\Psi}{(\widehat{E})})} \exp\lc \sum_{(i,j)\in\widehat{E}}\widehat{\Psi}_{i,j}{(\widehat{E})} o_io_j \rc \Pkonv \gamma_\bmo(\Psi),  \quad n\to\infty.
\]
\end{thm}

{Combined with the previous results, this shows that, under suitable assumptions, all parameters and the graph that govern the dependence structure of an IHR process can be estimated consistently. The practical performance of those estimators is illustrated in simulations and a data application in the following sections.}

\section{Simulation study}
\label{secsimulation}

\subsection{Setup}
We conduct a simulation study to assess the finite sample properties of our estimators of the graph structure $G = (V,E)$ and the Ising weights $\Psi$. 

Across all experiments, the graphs are randomly drawn from the \BA{} model with attachment parameter $a\in\{1,2\}$, where $a=1$ corresponds to a tree and $a=2$ corresponds to a connected sparse graph with $2d-3$ edges. Moreover, for a given graph, the parameters of $\Theta$ on edges of $G$ are drawn uniformly on $[2,5]$. 

For the Ising model, we consider two regimes, the symmetric regime where all $\Psi_{i,j}=0$ and the asymmetric regime where each entry of $\Psi_{i,j}$ is drawn uniformly from $[-0.6,0.2]\cup[0.2,0.6]$ for $(i,j)\in E$. Across all simulations, we set $\balpha=\bm 3/\bm 2$ and $\bmc=\id$. For every sample size $n$ we set  $q_n:=n^{-3/10}$, i.e., we use the $(100 n^{-3/10})$-percent largest observations for the estimation of $\Gamma$. Since the entries of $\Gamma$ are estimated separately, it can happen that our estimate $\hat{\Gamma}$ is not conditionally negative definite,
in which case we project it onto the space of conditionally negative definite matrices.

Each experiment was repeated 50 times. The code to reproduce the simulation study can be found on \hyperlink{https://github.com/florianbrueck/graphicalmodelsforstableprocesses}{https://github.com/florianbrueck/graphicalmodelsforstableprocesses}.

\subsection{Graph recovery}

We compare several graph estimation methods on the simulated
 data described above.
In order to assess graph recovery between true edges $E$ and estimated edges $\hat E$, we use the $F_1$-score defined as
$$ F_1(E,\hat{E}) := \frac{2 \cdot |E\cap\hat{E}|}{2 \cdot |E\cap\hat{E}| + |  E^\complement\cap\hat{E} | +  | E\cap \hat{E}^\complement |}.
$$
As a first baseline method we use the minimum spanning tree based on
the L\'evy correlation $\chi$ from \cite{engelke2024levygraphicalmodels}.
We also use the same method but based on our estimator 
of the \HR{} parameter matrix $\Gamma$. 
For estimation of general graphs we use our new method based on the \texttt{EGLearn} algorithm \citep{engelkelalancettevolgushev2022}, either with neighborhood selection (NS) and graphical lasso (Glasso), 
which has a tuning parameter $\rho \geq 0$ controlling the sparsity of the estimated graphs. We denote the corresponding graph estimates and presision matirx by
$\hat{G}_\rho=(V,\hat{E}_\rho)$ and $\hat{\Theta}_\rho$, respectively.
To select the optimal $\rho$ in a data-driven way, we propose using the pseudo \HR{} log-likelihood \citep{rottger2023total} with 
penalty for model complexity 
$$ 
-2\left\{ \log (\vert\hat{\Theta}_\rho\vert_+)  +\frac{1}{2} \mathrm{tr}(\widehat{\Gamma} \hat{\Theta}_\rho)\right\} +  \text{IC}\{\hat E_\rho\} 
 $$
as an approximation of AIC or BIC, where 
 $$ \text{IC}(\hat E_\rho):=\begin{cases}
    2 |\hat{E}_\rho |,& \text{IC}=\text{AIC}, \\ 
    \log\lc  nq_n \rc |\hat{E}_\rho|,&\text{IC}=\text{BIC}.
 \end{cases} $$
 Maximizing this pseudo likelihood over a range of $\rho$ values yields a data-driven method of hyperparameter selection.  

 Figure \ref{fig:tree_f_1_score} shows boxplots of the $F_1$-score for dimension $d=10,20$ and sample size $n=2000,10000$ in the asymmetric regime for an underlying tree graph.
 We first observe that both minimum spanning tree methods
 outperform the general graph estimation technique. This is not 
 surprising since these methods are tailor-made for tree graphs.
 Interestingly, our tree estimator based on $\Gamma$ outperforms the estimator based on $\chi$; this is in line with the observations made
 in tree estimation in extreme value statistics \citep{engelkevolgushev2022}. A more detailed analysis of the performance of the minimum spanning tree algorithms with underlying weights $\Gamma$ and $\chi$ for smaller sample sizes is provided in Appendix \ref{appaddplots}, illustrating the superiority of method based on $\Gamma$ across all dimensions and sample sizes.
 Regarding our estimator for general graphs, we see that neighborhood selection consistently outperforms graphical lasso across all samples sizes and dimensions, and that the data-driven hyperparameter selection seems to work well, both for AIC and BIC.
 
  \begin{figure}[!htbp]
    \centering
    
        \centering
        \includegraphics[scale=0.65]{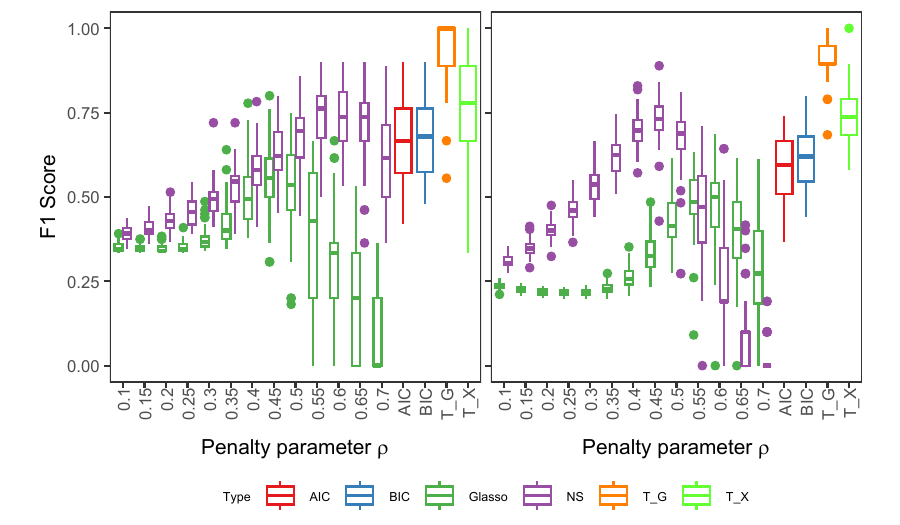}

        \centering
        \includegraphics[scale=0.65]{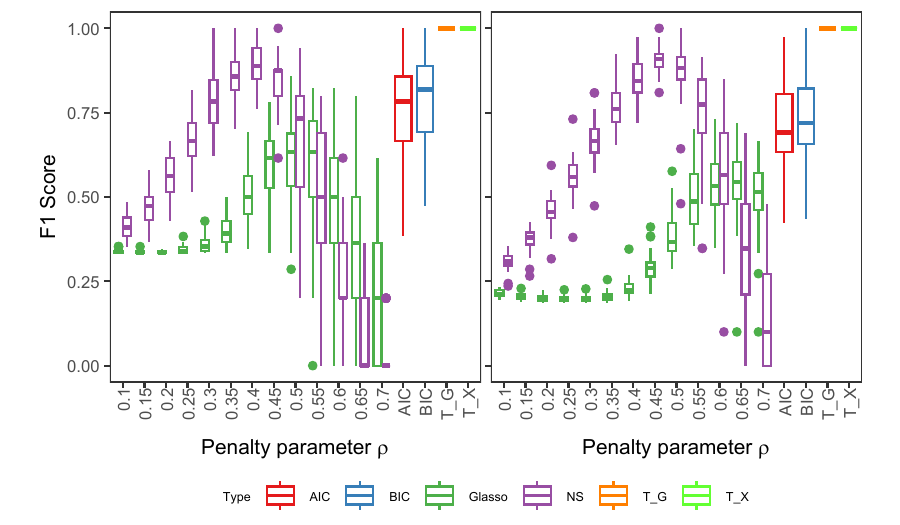}
    \caption{Boxplots of $F_1$-scores for estimation of tree graphs in the asymmetric regime for \texttt{EGLearn} based on neighborhood selection (NS) and graphical lasso (Glasso) for various penalty parameters and models selected by the two information criteria AIC and BIC, and the minimum spanning trees based on $\Gamma$ and $\chi$, for samples sizes $n=2000$ (top) and $n=10000$ (bottom) and dimensions $d=10$ (left) and $d=20$ (right).}

    \label{fig:tree_f_1_score}
\end{figure}

Figure \ref{fig:f_1_score_dim} shows the corresponding results for the same dimensions and sample sizes in the asymmetric regime for general graphs. Since now the underlying graph is no longer a tree,
the minimum spanning tree methods cannot recover well the denser
graph structure even for larger sample sizes.
On the contrary, the $F_1$-scores for the general graph estimation methods consistently improve for increasing sample sizes.
Again, neighborhood selection significantly outperforms the graphical lasso across all dimensions and sample sizes and should thus be considered as the preferred method.

  \begin{figure}[!htbp]
    \centering
    
        \centering
        \includegraphics[scale=0.75]{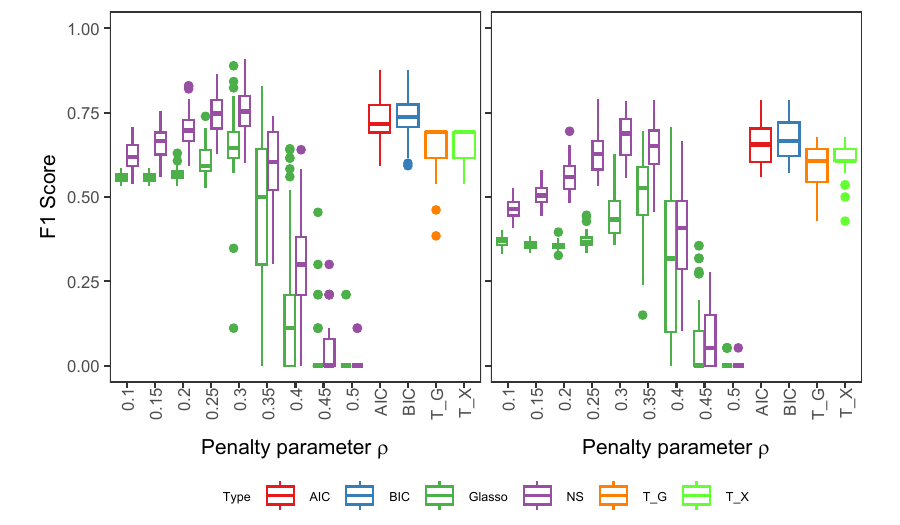}

        \centering
        \includegraphics[scale=0.75]{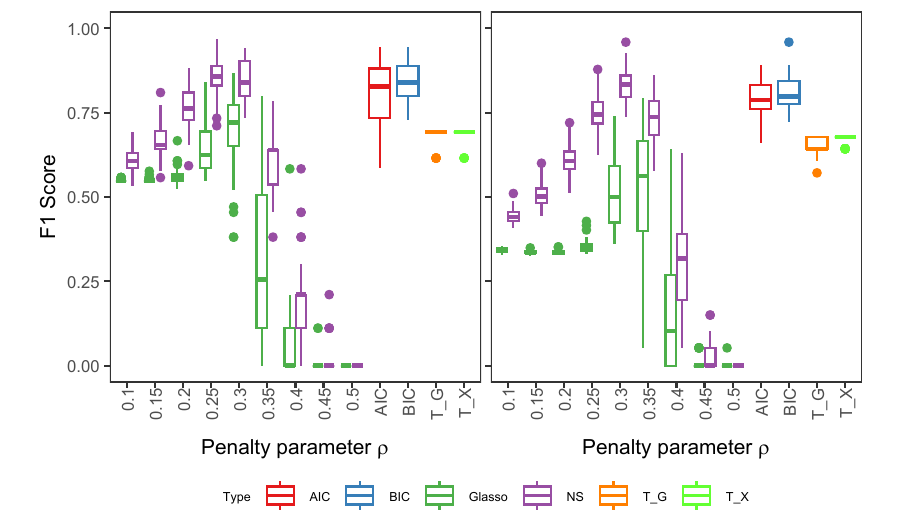}
    \caption{Boxplots of $F_1$-scores for estimation of \BA{} graphs with attachment parameter $a=2$ in the asymmetric regime for \texttt{EGLearn} based on neighborhood selection (NS) and graphical lasso (Glasso) for various penalty parameters and models selected by the two information criteria AIC and BIC, and the minimum spanning trees based on $\Gamma$ and $\chi$, for samples sizes $n=2000$ (top) and $n=10000$ (bottom) and dimensions $d=10$ (left) and $d=20$ (right).}
    \label{fig:f_1_score_dim}

\end{figure}
 
   Similar plots for the above experiments under the symmetric regime can be found in Appendix~\ref{appaddplots}. Interestingly, a general observation across all experiments is that the estimation of the graph works better in the asymmetric regime. A possible reason for this is  that in the asymmetric regime there are some orthants with a large number of observations,  which, as a consequence of our data-driven weighting scheme, yield a good estimate of $\Gamma_{i,j}^{(m,\bmo)}$ with a high weight. In contrast, in the symmetric regime, there are approximately $n/8$ observations available for each estimator $\hat{\Gamma}^{(m,\bmo)}_{i,j}$, which yields a higher bias on average in comparison to the estimates on orthants with many observations in the asymmetric regime.

\subsection{Estimation of Ising weights} \label{sec:est_ising}

In this subsection we focus on the estimation of the Ising parameters $\Psi$, which govern the asymmetry of the IHR model. We only consider the asymmetric regime as the parameters of $\Psi$ on the edges of the graph are non-zero in this case. We further assume that the true underlying graph is known in order to disentangle the performance of our estimation method for $\Psi$ from the graph estimation. For a given dimension $d$, we randomly generate one underlying graph from the \BA{} model with attachment parameter $a=2$ with a fixed corresponding parameter matrix $\Psi$. Across the 50 repetitions we assess the variation of the estimates of $\Psi$ around their true value. 

Since there is no closed form solution to the system of equations~\eqref{est_eq2} {we suggest a gradient ascent based procedure to obtain an approximate solution. To deal with possibly non-existing solutions, we further add a penalty to the original problem. Specifically, the system of equations~\eqref{est_eq2} gives the maximizer of the concave function 
\[
\hat f(\Psi) := { - \log C(\Psi) +} \sum_{(i,j) \in E} \widehat{\Cov}_{\Psi}(B_i,B_j) \psi_{i,j} 
\]
{over all $\Psi$ with $\psi_{i,j} = 0$ when $(i,j) \notin E$} whenever the maximizer exists. {To see this, recall that $C(\Psi)$ denotes the normalizing constant in~\eqref{defisingweight} and that $\partial \log C(\Psi)/\partial\psi_{i,j} = \e_\Psi[B_iB_j]$.} To deal with possible non-existence, we propose to instead maximize
\begin{equation}\label{eq:optpen}
\hat f(\Psi) - v \sum_{(i,j) \in E} |\Psi_{i,j}|
\end{equation}
where $v$ is a penalty parameter that we set to $v=0.05$ throughout all experiments. The corresponding sub-gradient is given by
\[
\widehat{\Cov}_{\Psi}(B_i,B_j) - \e_\Psi[B_iB_j]- v\sgn(\Psi_{i,j}), \quad (i,j) \in E.
\]
Since $\e_{\Psi}\lk B_iB_j \rk$ is not available in closed form, we replace it by an empirical estimate based on $5000$ MCMC samples 
of the Ising model. To perform the optimization of~\eqref{eq:optpen}, we opt here for the ADAM optimizer \citep{kingma2014}, with which we observe fast convergence and reliable estimates within less than $500$ iterations. An exemplary traceplot of our gradient ascent procedure is provided in Figure \ref{fig:traceplot} in the Appendix, which demonstrates quick convergence of the procedure.
}

Figure \ref{fig:boxplot_psi_dim_10_asym} shows boxplots of the estimates of $ \Psi_{i,j}$ for $(i,j)\in E$ together with the corresponding true values
for different dimensions and sample sizes. The variation of the parameter estimates and the bias decreases as the sample sizes increases. 

For lower sample sizes, we observe some outliers of extreme parameter estimates. This may be an artifact of the bias in the estimation of $\chi^{(o_1,o_2)}_{i,j}$. Overall, the estimation of the Ising weights seems to work well. In particular, the sign of the estimate of $\Psi_{i,j}$, and thus the direction of dependence of component $i$ and $j$, is always estimated correctly. Moreover, the ordering of the estimates $\hat{\psi}_{i,j}$ is largely in line with the true ordering of the $\psi_{i,j}$, which ensures that the ordering of the strengths of dependencies is reflected correctly.

 \begin{figure}[!htbp]
    \centering
        \includegraphics[trim = 0 2cm 0 1cm, clip = TRUE, width=\linewidth]{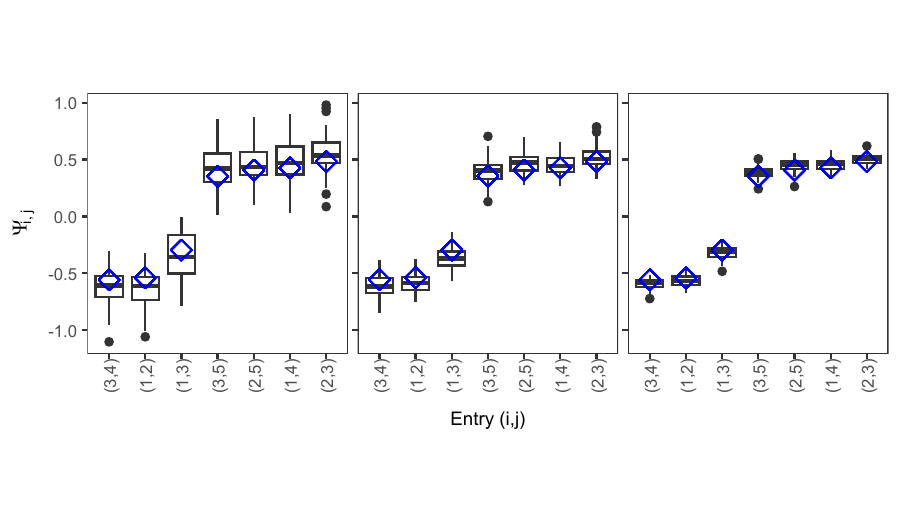}     
        \includegraphics[trim = 0 2cm 0 1cm, clip = TRUE, width=\linewidth]{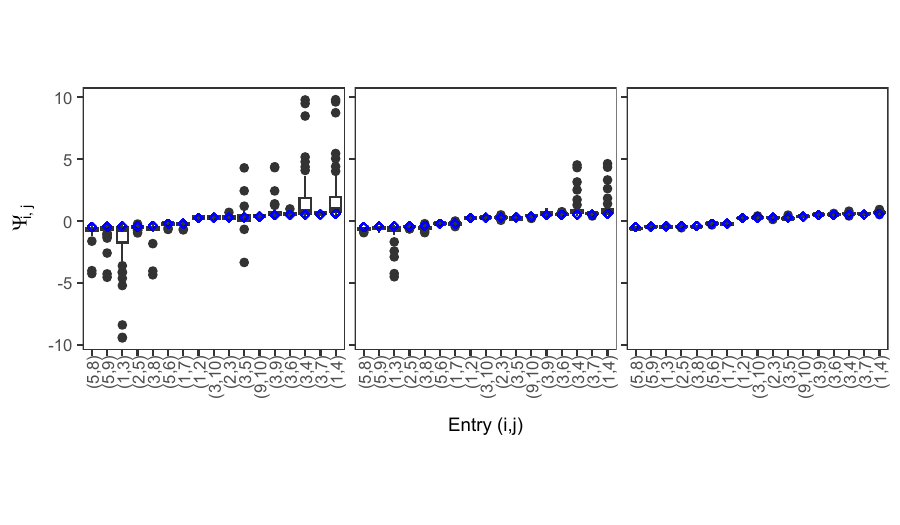}
    \caption{Boxplots of the parameter estimates of $\Psi$ along the edges of the graph $E$ for dimensions $d=5$ (top) and $d=10$ (bottem) and sample sizes $n=500$ (left), $n=2000$ (middle) and $n=10000$ (right). Displayed in ascending order of the true underlying parameter $\psi_{i,j}$ (blue).}
        \label{fig:boxplot_psi_dim_10_asym}
\end{figure}

\section{Application}
We apply our IHR L\'evy process to model $n=1509$ daily log-returns from stocks of $16$ American companies and compare the results to the minimum spanning tree approach with weights $\chi$ in \cite{engelke2024levygraphicalmodels}; for details on the data set we refer to the latter paper. 
We focus on assessing the estimated dependence structure of the data, since our model parameters $G$, $\Theta$ and $\Psi$ encode the dependence structure of the underlying Lévy process independently of the marginal parameters $\bac$.

We rely on the Black--Scholes model for L\'evy processes \citep[e.g.,][Chapter 11]{conttankov} that assumes the stock prices $(\mathbf S(t))_{t\geq 0}$ to follow an exponential of a $d$-dimensional L\'evy process $\mathbf L$ as
\[
\mathbf S(t) = \mathbf S(0)\exp\{\mathbf b t + \mathbf L(t)\}, \quad t \geq 0,
\]
where $\mathbf S(0)$ are the initial stock prices and $\mathbf b \in \mathbb{R}^d$ is a drift. We model $\bmL$ via an IHR L\'evy process with arbitrary underlying graph $G$. This generalizes the model of \cite{engelke2024levygraphicalmodels} whose method requires a tree graphical model.

Figure~\ref{fig:est_graphs} shows the estimated tree $\hat{T}^\chi$ from \cite{engelke2024levygraphicalmodels}, the estimated minimum spanning tree $\hat T^\Gamma$ based on our estimate of $\Gamma$, and our estimated general graph $\hat{G}$ based on AIC; the BIC graph in identical in this case.
Both estimated trees represent well the industry sectors of the assets as separate branches of the graph, however, with slightly different connections between the sectors. Moreover, the connections within a sector vary between the trees; this is probably due to strong dependence between these assets but the tree does not allow for additional edges. The general graph $\hat{G}$ contains additional  connections within and across industry sectors. The majority of a edges are still within the same industry sectors, showing the stronger dependence between stocks in the same industry. Interestingly, the tree $\hat T_\Gamma$ is almost a sub-graph of the general graph $\hat G$ with the exception of the edge between CFX and WFC which connects the Consumer Staples and the Energy sector, even though this is not enforced in any way.

 \begin{figure}[tb]
    \centering
    
    \begin{subfigure}{0.3\linewidth}
        \centering
        \includegraphics[width=\linewidth]{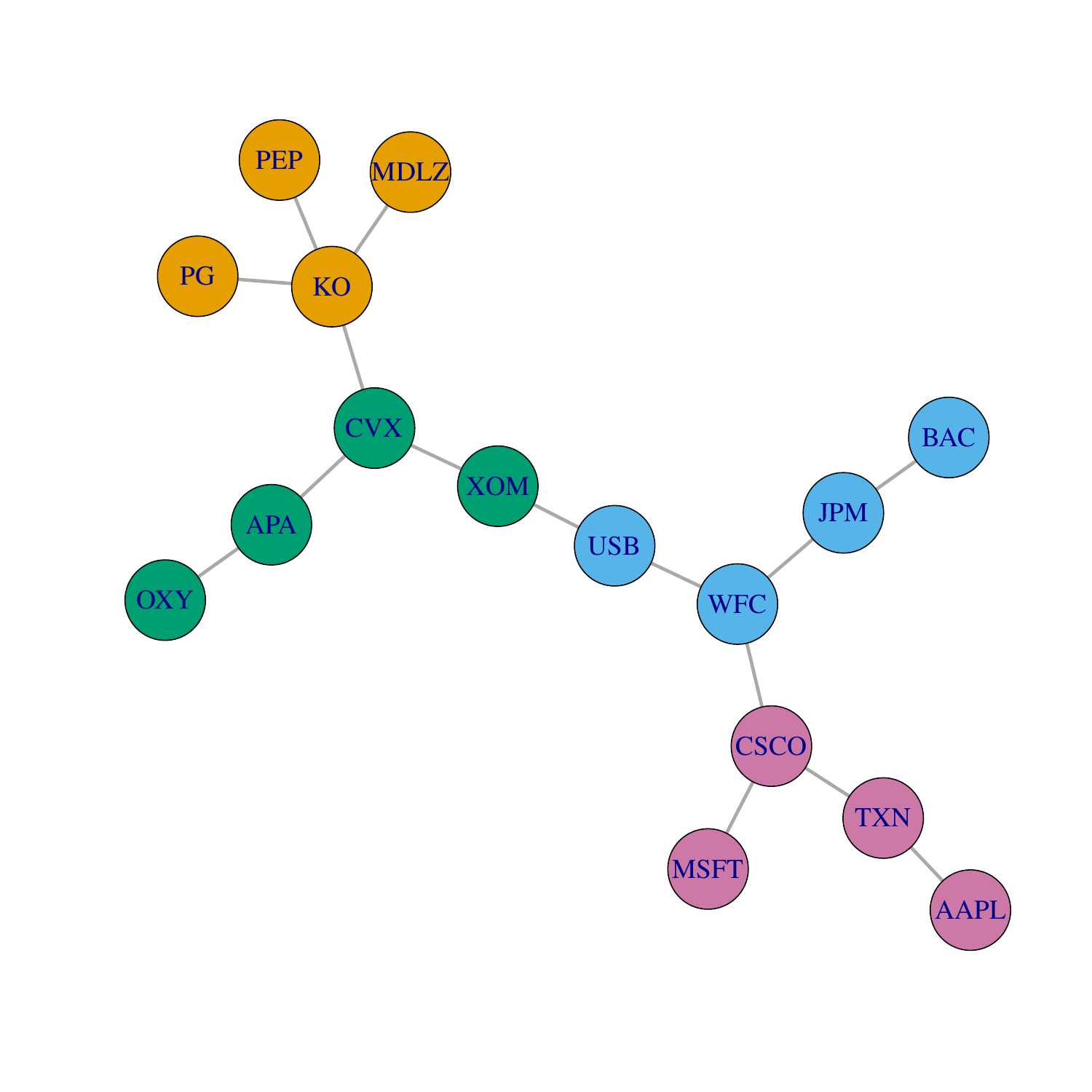}
        %\caption{First plot caption}
    \end{subfigure}
        \begin{subfigure}{0.3\linewidth}
        \centering
        \includegraphics[width=\linewidth]{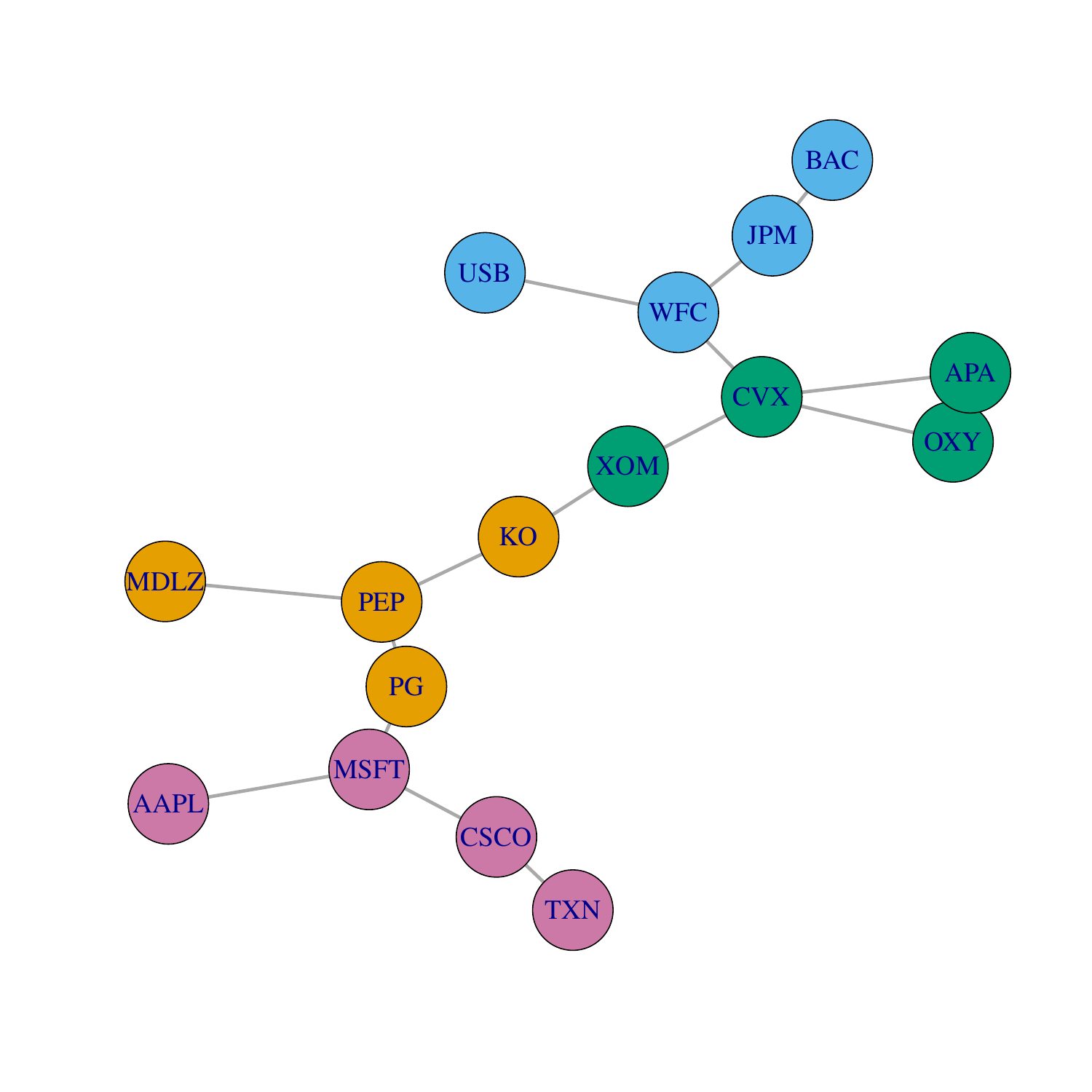}
        %\caption{First plot caption}
    \end{subfigure}
    \begin{subfigure}{0.3\linewidth}
                  \centering
        \includegraphics[width=\linewidth]{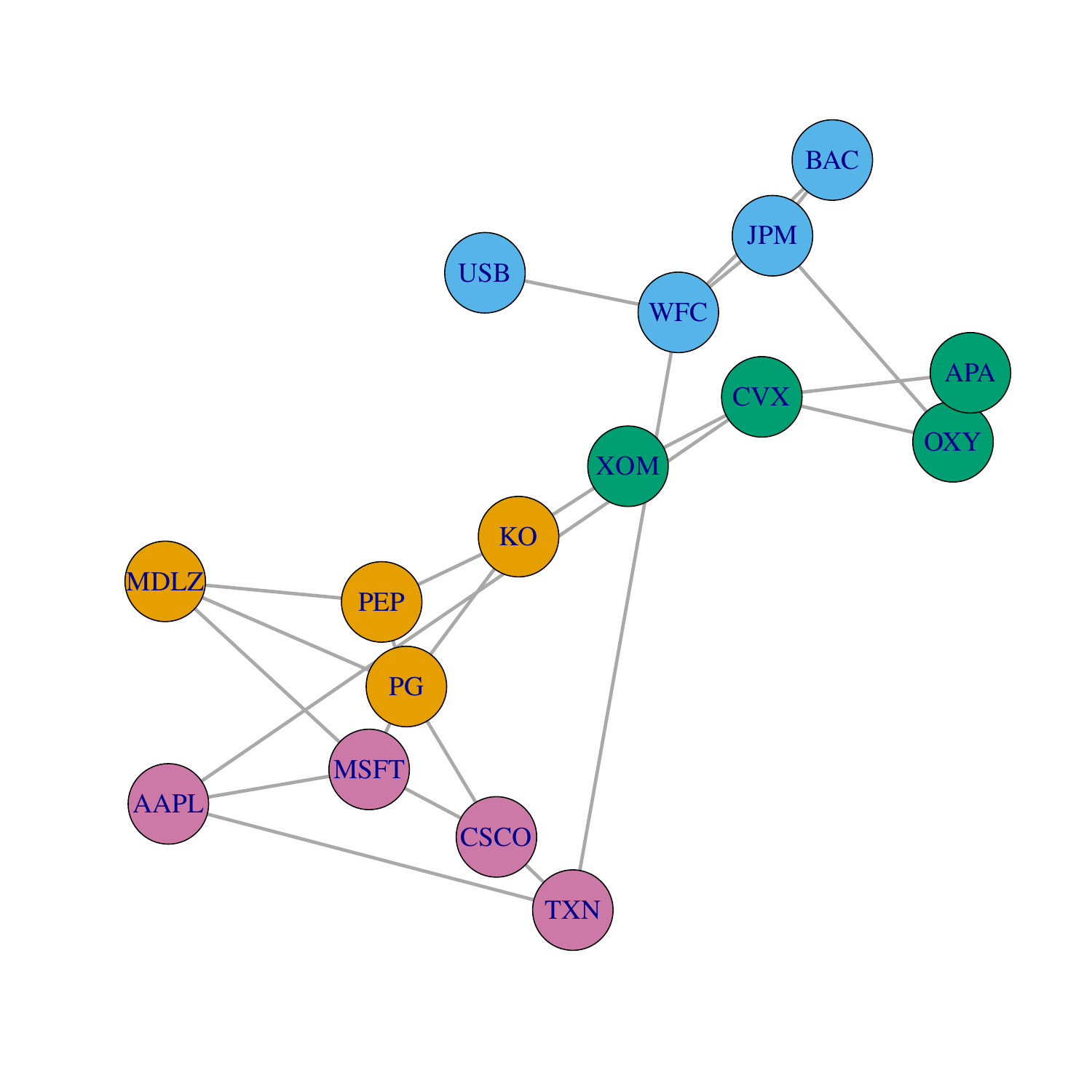}
    \end{subfigure}
    \caption{Estimated trees $\hat T^\chi$ \citep{engelke2024levygraphicalmodels} (left) and $\hat T^\Gamma$ (middle), and the estimated general graph $\hat{G}$ based on AIC (right). The colors encode the industry sectors of the underlying assets given by Financials (yellow),
Consumer Staples (blue), Technology (purple), Energy (green).}
    \label{fig:est_graphs}
\end{figure}

 To compare the goodness of fit of the estimated graphs 
 we compare the estimated and implied parameters of L\'evy correlation $\chi_{ij}$ as follows. 
The parameters matrix $\hat \Gamma^{\hat G}$ of a \HR{} L\'evy model on graph $\hat G = (V, \hat E)$ satisfies  
$\hat \Gamma^{\hat G}_{ij} = \hat \Gamma_{ij}$ for all $(i,j)\in\hat E$, where $\hat \Gamma$ is the empirical variogram matrix from~\eqref{defestimatorvariogrammatrix}. The values of 
$\hat \Gamma^{\hat G}_{ij}$ on non-edges $(i,j) \notin \hat E$ are implied by the completion from \cite{hentschelengelkesegers2023}. A measure of how well a 
graph structure explains the dependence of the overall
data is thus how well the variogram values implied by the 
the model matches the empirical values.
For plotting, we transform the variogram values to
the extremal correlation \citep{sch2003}, by exploiting the formula 
\[ \chi_{ij} = 2 - 2 \Phi(\sqrt{\Gamma_{ij}}/2),\]
for \HR{} distributions \citep[e.g.,][]{engelke2015}.
Figure~\ref{fig:data_app_comp_emp_vs_impl} shows
the empirical extremal correlation compared to those implied by the three different models.
The tree models have a clear bias and underestimate the empirical correlations. This is due to the fact that trees are very sparse models and therefore the non-adjacent nodes are relatively independent in this model class. The general graph overcomes this bias issue as it allows for denser graphs. In fact, the right amount of sparsity is selected in a data-driven way by AIC or BIC criteria.

 \begin{figure}
        \centering
        \includegraphics[width=\linewidth]{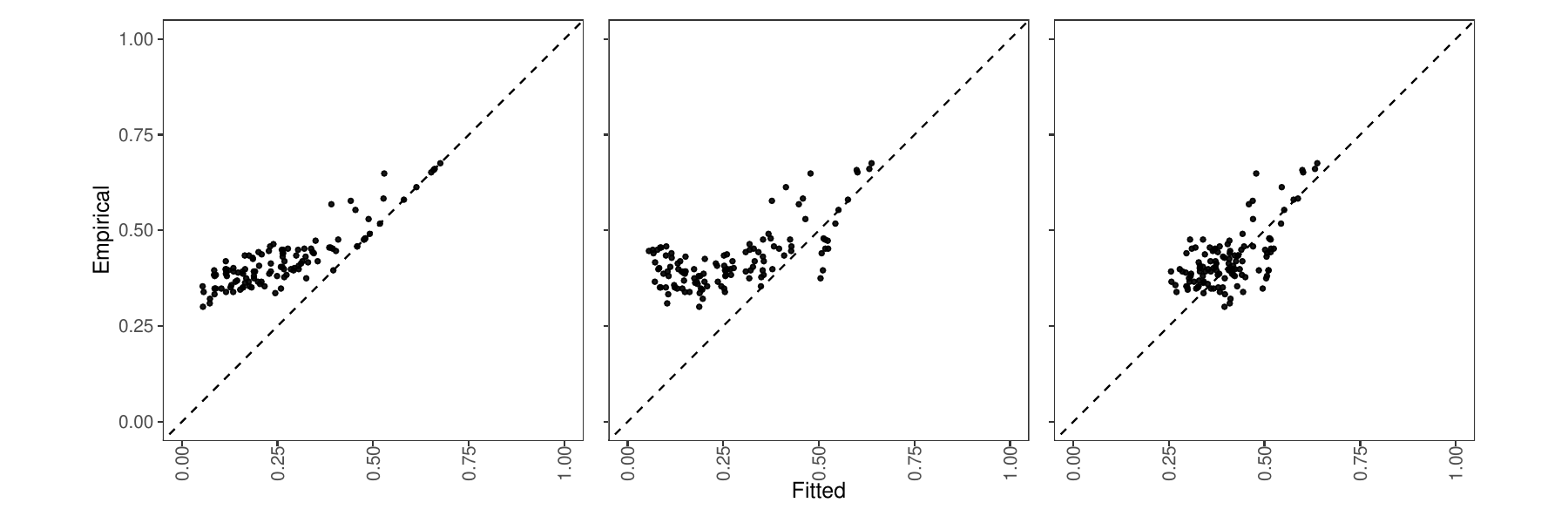}
    \caption{Implied vs estimated entries of $\chi$ for the minimung spanning tree with weights $\chi$ (left), $\Gamma$ (middle) and the general graph (right). }
    \label{fig:data_app_comp_emp_vs_impl}
\end{figure}

To assess our estimates of the orthant weights, the left panel in Figure \ref{fig:data_app_comp_emp_vs_impl2} shows a traceplot of the estimates for $\Psi$ using the gradient ascent described in Section~\ref{secsimulation}. All parameter estimates are non-negative, indicating strong positive correlation in the data, which is to be expected since stock data is known to be strongly positively correlated. To assess the adequacy of the parameter estimates of $\Psi$, Figure \ref{fig:data_app_comp_emp_vs_impl2} compares the mass on the bivariate positive and negative orthants for components $i,j$ implied by the graph structure given by $m_{i,j}=\sum_{o\in\Oo, o_io_j=1 } \gamma_\bmo(\Psi)/2$ with their empirical estimates given by
$$\hat m_{i,j}=\frac{\hat{\chi}^{(+,+)}+\hat{\chi}^{(-,-)}}{\hat{\chi}^{(+,+)}+\hat{\chi}^{(-,-)}+\hat{\chi}^{(+,-)}+\hat{\chi}^{(-,+)}}$$
as they may be interpreted as the strength of positive dependence implied by $\Psi$. The implied and estimated values of $m_{i,j}$ are close to the diagonal, demonstrating a good fit of our model of the orthant weights. All in all, our IHR process model seems to outperform the tree-based models on this dataset since it provides a superior or at least equally good fit across all metrics.
 \begin{figure}
 \centering
    \begin{subfigure}{0.25\linewidth}
        \centering
        \includegraphics[width=\linewidth]{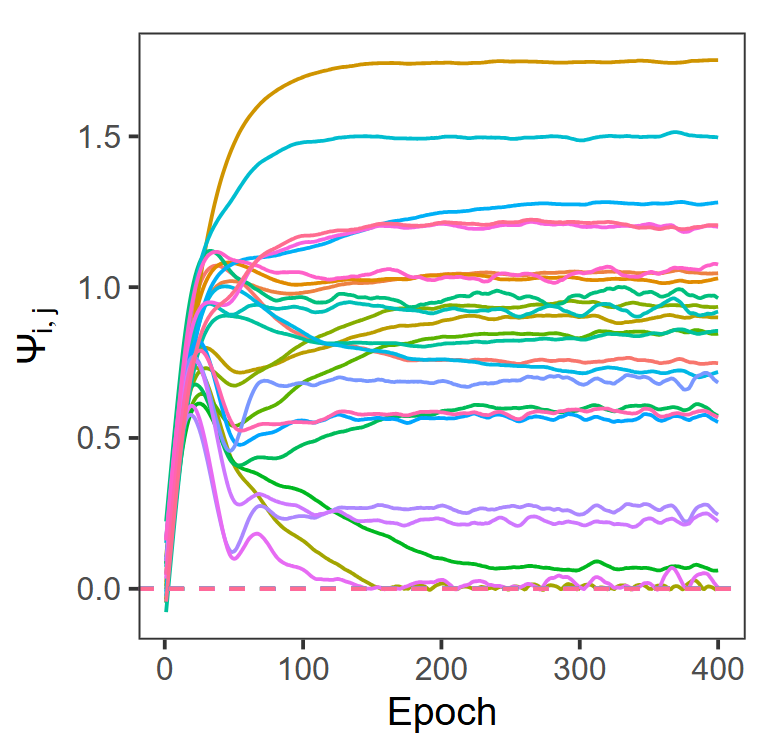}
        %\caption{First plot caption}
    \end{subfigure}
\hspace{-0.3cm}
        \begin{subfigure}{0.75\linewidth}
        \centering
        \includegraphics[width=\linewidth]{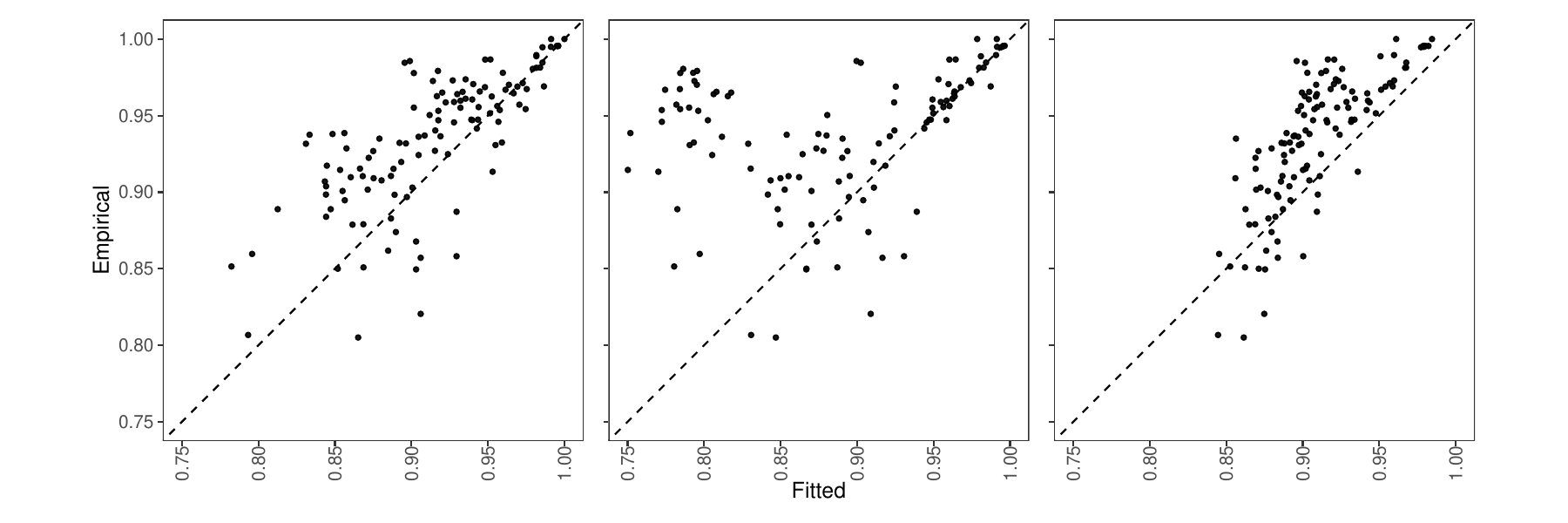}
        %\caption{First plot caption}
    \end{subfigure}

    \caption{Traceplot of the estimates of $\Psi$ on the estimated edges in $\hat{G}$ (left). Implied vs estimated $m_{i,j}$ for the minimum spanning trees based on $\chi$ (right,left), the minimum spanning trees based on $\Gamma$ (right,middle) and based on the general graph (right,right)}
    \label{fig:data_app_comp_emp_vs_impl2}
\end{figure}

\section{Conclusion}
In this paper, we introduce a family of PLMs which encodes the conditional independence properties of the associated Lévy process in terms of zero-entries of the symmetric \HR{} precision matrix $\Theta\in\R^{d\times d}$. This is an analog to the precision matrix $\Sigma^{-1}$ of a Brownian motion, which also encodes conditional independencies in its zero-entries. Thus, our model may be viewed as a continuous time analog of the Gaussian graphical model for non-Gaussian Lévy processes. To be able to incorporate asymmetric dependencies for the jumps of the Lévy process we introduce a new model for an orthant-weighting scheme of a PLM based on the Ising model, which retains the underlying graph and allows to asymmetrize a given PLM in a sparse way. We propose sparsistent estimation methods for the underlying parameter matrices $\Theta$ and $\Psi$, assuming that the marginal processes are stable with arbitrary stability indices. Importantly, the stability indices do not need to be identical across the margins and our proposed estimation procedure does not require any knowledge about these indices since it is solely based on the copula of the Lévy process at a fixed time point. Thus, our method models dependence and may be used in conjunction with any parameter estimation method for the marginal parameters of the stable processes to estimate the full data generating process.

\begin{acks}[Acknowledgments]
The authors would like to thank Axel B\"ucher and Mathias Vetter for helpful discussions regarding the proof of Lemma 4.3 in \cite{buechervetter2013}.
\end{acks}

\begin{funding}
Sebastian Engelke and Florian Brück were supported by the Swiss National Science Foundation under Grant 186858. Stanislav Volgushev was supported by a discovery grant (RGPIN-2024-05528) from NSERC of Canada.
\end{funding}

\bibliographystyle{imsart-nameyear.bst}
\bibliography{bib.bib}

\clearpage
\newpage
%%%%%%%%%%%%%%%%%%%%%%%%%%%%%%%%%%%%%%%%%%%%%%%%%%%%%%%%%%%%%%%%%%%%%%%%%%%%%%%%%%%%%%%%%%%
\appendix
%%%%%%%%%%%%%%%%%%%%%%%%%%%%%%%%%%%%%%%%%%%%%%%%%%%%%%%%%%%%%%%%%%%%%%%%%%%%%%%%%%%%%%%%%%%%

%%%%%%%%%%%%%%%%%%%%%%%%%%%%%%%%%%%%%%%%%%%%%%%%%%%%%%%%%%%%%%%%%%%%%%%%%%%%%%%%%%%%%%%%%%%%%

\addtocontents{toc}{\protect\setcounter{tocdepth}2} 

\tableofcontents

\section{Further properties of HR Lévy processes}
In this section, we discuss several properties of \HRL{} processes that will be used in subsequent proof sections. We start by introducing some additional notation.  We will frequently use the notation $\Lambda^{(\bag,\Theta)}$ or $\Lambda^{(\bag)}$ to emphasize that $\Lambda$ depends on $(\bag,\Theta)$ or $(\bag)$. Define
\begin{align}
g(\bmx;\Theta) &:= c(\Theta)\prodd \frac{1}{x_i^{1+1/d}} \exp\lc -\frac{1}{2}\log(\bmx)^\T \Theta\log(\bmx)+\log(\bmx)^\T r_\theta \rc \id_{\{ \bmx\in(0,\infty)^d\}}
\label{def:g1}
\end{align}
where the normalizing constant $c_\Theta$ is chosen such that $\int_{ (0,\infty)^{d-1}} g(\bmx,\Theta)\rmd \bmx_{\setminus i}=x_{i}^{-2}$. Note that $g$ is $2^{d-1}$ times the density of the normalized symmetric measure $\Lambda^*_{(sym)}$ from Example~\ref{ex:HR} restricted to the set $[0,\infty)^d$. Recalling~\eqref{defHRparetolevycop} we also see that 
$\lambda^\star(\bmx) = \sum_{\bmo \in \Oo} \gamma_\bmo\bm{1}_{\{\bmx \in  \bmo\}} g(|\bmx|;\Theta)$, where the absolute value is applied component-wise.

By \citet[][Equation (9)]{engelkehitz2020} we have that for an arbitrary $1\leq k\leq d$ the function $g$ can also be represented as
\begin{equation}\label{eq:galt}
g(\bmx;\Theta)= \id_{\{ \bmx\in(0,\infty)^d\}}  \varphi_{d-1}(\tilde{\bmx}_{\setminus k}; \Sigma^{(k)})  x_k^{-2}  \prod_{i \neq k} x_i^{-1},   
\end{equation} 
where $\varphi_p(\cdot; \Sigma)$ is the density of a centered $p$-dimensional normal distribution with covariance matrix $\Sigma$, $\tilde{\bmx}_{\setminus k} = \{\log(x_i / x_k) + \Gamma_{ik} / 2\}_{i \neq k}$, and $\Sigma^{(k)} = (\Theta_{\setminus\{ k\},\setminus \{k\}})^{-1}$.

Next, recall Definition~\ref{defnHRprocess} where we define an \HRL{} process with L\'evy measure $\Lambda^{(\bag,\Theta)}$. We now discuss several properties of the measure $\Lambda^{(\bag,\Theta)}$. To lighten notation, we will often drop the dependence of $\Theta$ and simply write $\Lambda^{(\bag)}$. 

\begin{prop}
\label{proplevymeasurehrprocess}
For every Hüsler-Reiss precision matrix $\Theta$, $\balpha\in(0,2)^d$, $\bmc\in(0,\infty)^{2d}$ and $\bgamma\in [0,\infty)^{2^d}$, $\Lambda^{(\bag, \Theta)}$ defines a valid L\'evy measure. It has a density with respect to Lebesgue measure on $\R^d$ which takes the form
\begin{align} \label{deflevydensityhüslerreissprocess}
\lambda(\bmx;\bag,\Theta) &= \sum_{\bmo\in\Oo}  \gamma_\bmo \id_{\{\bmx\in\bmo\}}   g\lc \lv \bm t_{1/\balpha,(1/\bmc)^{1/\balpha}}(\bmx)\rv;\Theta\rc \prodd\frac{ \alpha_i \vert x_i\vert^{\alpha_i-1}}{ c_i^{o_i}} 
\\
&= \sum_{\bmo \in \Oo} \gamma_\bmo \id_{\{\bmx\in\bmo\}} C_\bmo g(M_\bmo |\bmx|^{\balpha};\Theta) \prod_{i=1}^d |x_i|^{\alpha_i-1} \label{eq:hrlevydens_altrepr1}
\end{align} 
where\footnote{Here and in what follows, with a slight abuse of notation, $c_i^{o_i} = c_i^+$ if $o_i$ is positive and $c_i^-$ otherwise, and similarly for $c_i^{\sgn(x_i)}$.} $C_\bmo := \prodd \alpha_i/c_i^{o_i}$, $|\bmx|$ is interpreted as the component-wise absolute value and so is the power, i.e.
$(1/\bmc)^{1/\balpha}:=\lc \lc(c_i^+)^{-1/\alpha_i},(c_i^-)^{-1/\alpha_i}\rc\rc_{1\leq i\leq d}$, $M_\bmo = \mbox{diag}(1/c_1^{o_1},\dots,1/c_d^{o_d})$ and
\begin{align}
    &\bm t_{\balpha,\bmc}:\ \Rdo \to \R^d, \ \bmx\mapsto \Big( \sgn(x_1)\lc c_i^{\sgn(x_1)}\vert x_1\vert\rc ^{1/\alpha_1},\ldots,  \sgn(x_d)\lc c_i^{\sgn(x_d)}\vert x_d\vert \rc^{1/\alpha_d}\Big) . \label{defHRlevymeasuretrafo}
\end{align} 
Moreover, $\Lambda^{(\bag, \Theta)}$ is the push-forward measure of $\Lambda^{(\bm 1, \bm 1, \bm \gamma, \Theta)}$ under the map $t_{\balpha,\bmc}$.   
\end{prop}

\begin{proof}[Proof of Proposition~\ref{proplevymeasurehrprocess}]

Throughout the proof, we will drop the dependence of $\Lambda^{(\bag,\Theta)}$ on $\Theta$ and simply write $\Lambda^{(\bag)}$. We begin by deriving the density representation~\eqref{deflevydensityhüslerreissprocess}; the equality in~\eqref{eq:hrlevydens_altrepr1} then follows from straightforward computations. Define
\[
T(x)=\begin{cases} 
   (x,\infty), & x\geq 0, \\
    (-\infty,x], & x<0.
\end{cases}
\]
Observe that it suffices to prove the equality
\[
\Lambda^{(\bag)}\lc \times_{i=1}^d T(z_i)\rc = \int_{\times_{i=1}^d T(z_i)} \lambda(\bmx;\bag,\Theta) \rmd \bmx
\]
for all fixed orthants $\bmo$ and all $\bmz\in\bmo\cap\lc\R\setminus\{\bm 0\}\rc^d$. We have
\begin{align*}
\Lambda^{(\bag)}\lc \times_{i=1}^d T(z_i)\rc
&= \Lambda^{(\bm 1, \bm 1, \bgamma)}\lc \times_{i=1}^d T(\sgn(z_i)|z_i|^{\alpha_i}/c_i^{\sgn(z_i)}) \rc
\\
&= \Lambda^{(\bm 1, \bm 1, \bgamma)}\lc t_{\bm \alpha,\bm c}^{-1}\lc \times_{i=1}^d T(z_i)\rc\rc.
\end{align*}
This shows that $\Lambda^{(\bag)}$ is the push-forward of $\Lambda^{(\bm 1, \bm 1, \bgamma)}$ under $t_{\bm \alpha,\bm c}$. The map $t_{\bm \alpha,\bm c}$ is bijective, and continuously differentiable on the set $(\R\setminus\{0\})^d$. Noting further that the density of $\Lambda^{(\bm 1, \bm 1, \bgamma)}$ is given by 
\[
\bmx \mapsto \sum_{\bmo \in \Oo} \gamma_\bmo\bm{1}_{\{\bmx \in  \bmo\}} g(|\bmx|;\Theta),
\]
the formula for $\lambda$ follows by applying the usual density transformation formula on each orthant separately.

To prove that $\Lambda^{(\bag)}$ is a valid Lévy measure it is sufficient to check that $\Lambda^{(\bac,\bm 1)}$ is a valid Lévy measure, which amounts to checking that for for all $(c_i^+)\in(0,\infty)^d$ there is some open neighborhood $A$ of $\bm 0$ such that
\[
\int_{ A \cap\oinf^d}\Vert x\Vert_2^2\Lambda^{(\bac,\bm 1)}( \rmd \bmx) =\sum_{i=1}^d\int_{ A\cap[0,\infty)^d} x_i^2\Lambda^{(\bac,\bm 1)}( \rmd \bmx) <\infty.
\]
Choosing $A = (-1,1)^d$ we obtain by the change of variable formula for push-forward measures and the fact that $\Lambda^{(\bac,\bm 1)}$ is the push-forward of $\Lambda^{(\bm 1,\bm 1, \bm 1)}$ under $t_{\bac}$ 
\begin{align*}
\int_{ A \cap\oinf^d}\Vert x\Vert_2^2\Lambda^{(\bac,\bm 1)}( \rmd \bmx) &= \sum_{i=1}^d\int_{\tilde A\cap[0,\infty)^d} (c_i^+)^{2/\alpha_i}x^{2/\alpha_i}_i \Lambda^{(\id,\id,\id)}(\rmd\bmx)
\\
&\leq \sum_{i=1}^d(c_i^+)^{2/\alpha_i}\int_{0}^{(c^+_i)^{-1/\alpha_i}} x^{2/\alpha_i-2}_i \rmd x_i  <\infty,
\end{align*}
since $2/\alpha_i-2>-1$, where $\Tilde{A} = \times_{i=1}^d[-(c^-_i)^{-1/\alpha_i},(c^+_i)^{-1/\alpha_i}]^d$ is the pre-image of $A$ under $t_{\balpha,\bm c}$ and we used that the marginal densities of the \HR{} exponent measure are $x_i^{-2}I\{x_i > 0\}$ by homogeneity of the \HR{} exponent measure.
\end{proof}

Even though \HR{} processes are only self-similar, i.e.\ $(\g L(at))_{t\geq 0}\overset{d}{=} (a^{H}\g L(t))_{t\geq 0}$ for some $H>0$, when $\balpha=(\alpha,\ldots,\alpha)$, they obey the following ``generalized self-similarity'' property.
\begin{thm}
\label{thmselfsimHRprocess}
    Let  $\lc \g L(t)\rc_{t\geq 0}$ denote an \HR{} process with drift $\bm 0$.
    For every $\bm \tau\in\R^d$, $\bm \delta\in[0,\infty)$ and $a>0$ we have 
    $$ \lc\lc \delta_i L_i(at) +at\tau_i \rc_{\leqd}\rc_{t\geq 0} \overset{d}{=}  \lc\lc a^{1/\alpha_i}\delta_i L_i(t)+ t\lc a\tau_i +\delta_i b_i(a)\rc  \rc_{\leqd}\rc_{t\geq 0},$$
    where for $i \in  V$
    \begin{equation}\label{eq:deffunctionb}
    b_i(a):= a^{\frac{1}{\alpha_i}} \int_{\Rdo} \bmo(x_i)\lc c^{\bmo(x_i)}_i \vert x_i\vert \rc^{\frac{1}{\alpha_i}} \lc \id_{\{\Vert t_{\balpha,\bmc}(\bmx) \Vert_2 \leq 1\}}-  \id_{\{\Vert t_{\balpha,\bmc}(a\bmx) \Vert_2 \leq 1\}} \rc    \Lambda^{(\bm 1, \bm 1, \bgamma)}(\rmd\bmx). 
    \end{equation}
\end{thm}

\begin{proof}[Proof of Theorem~\ref{thmselfsimHRprocess}]
 It suffices to prove the claim for $ \bm\tau=\bm 0$. Thus, let $ \bm\tau=\bm 0$, $a>0$ and first set $\bm \delta=\bm 1$.
    The characteristic function of $\g L(at)$ is given by $\exp(q(\bmz))$ where
    \begin{align*}
    q(\bmz)
    &=at \sum_{\bmo\in\Oo} \gamma_\bmo \int_{\bmo \cap \Rdo} \Big\{\exp\lc i \bmz^\intercal t_{\balpha,\bmc}(\bmx)\rc- 1-i\bmz^\intercal t_{\balpha,\bmc}(\bmx) \id_{\{\Vert t_{\balpha,\bmc}(\bmx) \Vert_2 \leq 1\}} \Big\} g(\vert\bmx\vert;\Theta)\rmd \bmx
    \\
    &= t \sum_{\bmo\in\Oo} \gamma_\bmo \int_{\bmo \cap \Rdo} \Bigg\{ \exp\lc i\sum_{j=1}^d z_j \bmo(ax_j) \lc c_j^{\bmo(ax_j)}ax_j\rc^{1/\alpha_j}\rc- 1
    \\
    &\quad\quad\quad\quad\quad\quad -i\sum_{j=1}^d z_j \bmo(ax_j) \lc c_j^{\bmo(ax_j)}ax_j\rc^{1/\alpha_j} \id_{\{\Vert t_{\balpha,\bmc}(a\bmx) \Vert_2 \leq 1\}} \Bigg\} g(a\vert \bmx\vert ;\Theta)a^{d+1}\rmd \bmx
    \\
    &\overset{\star}{=}t  \sum_{\bmo\in\Oo} \gamma_\bmo \int_{\bmo \cap \Rdo} \Big\{ \exp\lc i \bmz(a,\balpha)^\intercal t_{\balpha,\bmc}(\bmx)\rc- 1-
    \\
    &\hspace{3cm}i\bmz(a,\balpha)^\intercal t_{\balpha,\bmc}(\bmx) \id_{\{\Vert t_{\balpha,\bmc}(\bmx) \Vert_2 \leq 1\}} \Big\} g(\vert\bmx\vert;\Theta)\rmd \bmx
    \\
    &\quad +it \sum_{\bmo\in\Oo} \gamma_\bmo \int_{\bmo \cap \Rdo} \bmz^\intercal \mathbf{v}_{a,\balpha,\bmc}(\bmx)\lc \id_{\{\Vert t_{\balpha,\bmc}(\bmx) \Vert_2 \leq 1\}}-  \id_{\{\Vert t_{\balpha,\bmc}(a\bmx) \Vert_2 \leq 1\}} \rc g(\vert \bmx\vert ;\Theta)  \rmd \bmx\\
    &= \tilde q(\bmz)
    +i \bmz^\intercal(t \bm h(a)),
    \end{align*}
    where $\bmz(a,\balpha):=\lc a^{1/\alpha_i}z_i\rc_{\leqd}$, $\mathbf{v}_{a,\bac} := (a^{1/\alpha_i}[\bm t_{\bac}(\bmx)]_i)_{i \in  V}$, in $\star$ we have used that $g$ is $-(d+1)$-homogeneous,  $\tilde q$ is such that $\exp(\tilde q(\bmz))$ is the characteristic function of $(a^{1/\alpha_i}L_i(t))_{i \in  V}$
    and $\bmb$ is as defined in the statement of the theorem. Finally, using $\tilde{\bmz}= \lc\delta_i z_i\rc_{i \in  V}$ in the calculations above we obtain the claim.
\end{proof}

\begin{lem}
\label{lem:posorthantprob}
Let $\g L$ denote a \HRL{}$(\Theta,\bm \gamma,\bm\alpha,\g c)$ process with $\gamma_\bmo > 0$ for all $\bmo\in \Oo$. Then, for every $t>0$ and every orthant $\bmo\in \Oo$ we have $P(\g L(t)\in\bmo)>0$. 
\end{lem}
\begin{proof}
Throughout this proof, we simplify notation and denote by $\Lambda$ the PLM of $\g L$.
    Fix a $\delta>0$ and decompose $(\g L(t))_{t\geq 0}\overset{d}{=}\lc \bmX(t)+\g P(t)\rc_{t\geq 0}$, where $\bmX$ is a Lévy process independent of the  compound Poisson process $\bm P$ with Lévy measure $\id_{\{\bmx\in B^\complement_\infty(\delta)\cap \bmo\}}\Lambda(\rmd\bmx)$. We have that $\g P(t)\overset{d}{=}\sum_{i=1}^{N(t)} \g J_i$  where $N(t)\sim  \text{Poi}\lc t \Lambda\lc \bmo \cap B^\complement_\infty(\delta)\rc \rc$ and $(\g J_i)_{i\in\N}$ is an i.i.d.\ sequence of random vectors which are independent of $N(t)$ and have distribution
    $$
    P(\g J_i\in \cdot) = \frac{\Lambda\big( \cdot\cap\bmo\cap B^\complement_\infty(\delta)\big) }{\Lambda\big( \bmo\cap B^\complement_\infty(\delta)\big)}. 
    $$ 
    Fix $t > 0$. Since $\mathbf{X}(t)$ is real-valued, there exists an $M>0$ such that $P(\|\mathbf{X}(t)\|_\infty \leq M) \geq 1/2$. Define 
    $$
    A := \Big\{\bmx\in \bmo \big| \min_{i} |x_i| \geq M + 2\delta
    \Big\}. 
    $$
    Then $\Lambda(A)>0$ since $\gamma_\bmo>0$ and therefore $\Lambda\big( \cdot \cap B^\complement_\infty(\delta)\big)$ has a density with respect to Lebesgue measure which is strictly positive on $B^\complement_\infty(\delta)$. Note also that $\|\mathbf{X}(t)\|_\infty \leq M$ together with $\bm c \in \bmo$ and $
    \min_i |c_i| > M+\delta$ implies $\g X(t) + \bm c \in \bmo \cap B^\complement_\infty(\delta)$. Thus, conditioning on $N(t)=1$, we obtain
    \begin{align*}
        P\lc \g L(t)\in\bmo\rc
        &\geq P\lc \|\bmX(t)\|_\infty \leq M,  \g P(t)\in  A\rc
        =P\lc \|\bmX(t)\|_\infty \leq M \rc P\lc \g P(t)\in  A\rc\\
        &\geq P\lc\|\bmX(t)\|_\infty \leq M \rc P\lc N(t)=1, J_1\in A\rc\\
        &=P\lc\|\bmX(t)\|_\infty \leq M \rc P\lc N(t)=1\rc \frac{ \Lambda\lc A\rc}{\Lambda\lc \bmo\cap  B^\complement_\infty(\delta) \rc}
        \\
        &>0,
    \end{align*} 
    which completes the proof.
\end{proof}

\newpage

\subsection{Properties of the HR (Lévy) density}
\label{sec:propHRLdens}

In this section, we derive several regularity properties of the Lévy density of a Hüsler-Reiss process that will later be needed to control the small time behavior of the Hüsler-Reiss Lévy process. The proofs are collected at the end of this section.

Recall that $B_\infty(\delta)$ stands for the closed ball of radius $\delta$ with respect to $\|\cdot\|_\infty$. 
Moreover, to keep the notation concise, we denote $\partial_{i}f(\bmx)=\frac{\partial}{\partial x_1} f(\bmx)$ and $\partial_{i,j}f(\bmx)=\frac{\partial^2}{\partial x_i\partial x_j}f(\bmx)$ for any function $f:\R^d\to \R$ whose corresponding partial derivatives exist.

In the proofs, we will mostly use the following representation of the Lévy density of a \HR{} process $\lambda$ from (\ref{deflevydensityhüslerreissprocess}), which by~\eqref{eq:galt} holds for any $k = 1,\dots,d$
\begin{align*}
\lambda(\bmx;\balpha,\bmc,\bgamma,\Theta) &= \sum_{\bmo \in \Oo} \gamma_\bmo C_\bmo I_\bmo(\bmx) g(M_\bmo |\bmx|^{\balpha}) \prod_{i=1}^d |y_i|^{\alpha_i-1}
\\
& = \sum_{\bmo \in \Oo} \gamma_\bmo \tilde C_\bmo  I_\bmo(\bmx) h_k(|\bmx|;\bm\mu_\bmo^{(k)},\Sigma^{(k)}) |y_k|^{-\alpha_k}\prod_{i=1}^d |y_i|^{-1}
\end{align*}
where $I_\bmo(\bmx)=\id_{\{\bmo(\bmx)=\bmo\}}$ is an indicator of whether $\bmx$ lands in the corresponding open orthant, $C_\bmo, \tilde C_\bmo$ are constants that depend on $\bmc,\balpha$ and the orthant $\bmo$, $|\bmx|$ is interpreted as the component-wise absolute value and so is the power, i.e. $|\bmx|^{\balpha} = (|x_1|^{\alpha_1},\dots,|x_d|^{\alpha_d})$, $\bm\mu_\bmo^{(k)}$ is a vector of constants that depends on $k, \Theta, \bmo$, and
\[
M_\bmo = \mbox{diag}(1/c_1^{o_1},\dots,1/c_d^{o_d}).
\]
We also define 
$$h_k(\bmx;\mathbf{v}, A) = \varphi_{d-1}\big(\log(\bmy_{\backslash k}^{\bm \alpha_{\backslash k}}/y_k^{\alpha_k}) - \mathbf{v};A\big).$$
Finally, for every $\delta>0$ define
$$ [0,\infty)^d_\delta:=\{ \bmx\in [0,\infty)\mid \Vert \bmx\Vert_\infty>\delta \} $$
and
$$ Z_\delta:=\{\bmx\in \R^d\setminus B_\infty(\delta)\mid x_j=0 \text{ for some } 1\leq j\leq d\} \text{ and }Z^+_\delta:=Z_\delta\cap [0,\infty)^d,$$
where $Z_0:=\cup_{\delta>0}Z_\delta=\{\bmx\in \R^d\setminus \{0\}\mid x_j=0 \text{ for some } 1\leq j\leq d\}. $

\begin{lem}\label{lem:tech1}
Assume $\alpha_i \in (0,2), i \in  V $, $\min_{i \in  V} c_i^+ \wedge c_i^- > 0$. For every $\kappa > 0$ and $\mathbf{v}\in\R^{d-1}$ there exists a constant $C_\kappa = C_{\kappa}(\bm\alpha,\Theta,d,\mathbf{v})$ such that for all $\bmy \in (0,\infty)^d$ and for $k=1,\dots,d$
\begin{equation}\label{eq:boundphi}
\varphi_{d-1}(\log(\bmy_{\backslash k}^{\bm \alpha_{\backslash k}}/y_k^{\alpha_k}) - \mathbf{v};\Sigma^{(k)}) \leq C_\kappa  \prod_{i\neq k} \Big(\frac{y_i^{\alpha_i}}{y_k^{\alpha_k}} \wedge \frac{y_k^{\alpha_k}}{y_i^{ \alpha_i}} \Big)^\kappa.   
\end{equation}
\end{lem}

Now, we can derive a corresponding bound for the density of the HR process.

\begin{lem}\label{lem:propdens}
Assume $\alpha_i \in (0,2), i \in  V $, $\min_{i \in  V} c_i^+ \wedge c_i^- > 0$. The density $\lambda$ of $\Lambda^{(\bag,\Theta)}$ is continuous on $\Rdo$  with $\lambda(\bmx)=0$ for all $\bmx\in Z_0$. Further, for any $\xi > 0$ there exists a constant $C_{0,\xi}$ that also can depend on $d, \Theta, \bm\gamma, \bm\alpha, \bmc$ such that for any $k = 1,\dots,d$
\begin{equation}\label{eq:boundf}
\lambda(\bmx) \le \id_{\{x_k \neq 0\}} C_{0,\xi} |x_k|^{-\alpha_k(1+\sum_{i=1}^d \frac{1}{\alpha_i})} \prod_{i=1}^d \Big(\frac{|x_i|^{\alpha_i}}{|x_k|^{\alpha_k}} \wedge \frac{|x_k|^{\alpha_k}}{|x_i|^{ \alpha_i}} \Big)^\xi.    
\end{equation}
In particular, $\lambda$ is bounded on $B^\complement_\infty(\delta)$ with a bound that solely depends on $d$, $\Theta$, $\bgamma$, $\bmc$, $\balpha$, $\delta$.    
\end{lem}

This bound allows us to derive the following regularity properties of the density of a \HRL{} process.

\begin{lem}\label{lem:partder1} Denote by $\lambda$ the density of $\Lambda^{(\bag,\Theta)}$. Assume $\alpha_i \in (0,2), i \in  V $, $\min_{i \in  V} c_i^+ \wedge c_i^- > 0$. The following statements hold.
\begin{enumerate}
\item[(i)] The partial derivatives of $\lambda$ exist for every $\bmx \in Z_0$ and satisfy $\partial_j \lambda(\g x) = 0$ for $\bmx \in Z_0$.
\item[(ii)] For any $\xi > 0$ there exists a constant $C_{3,\xi}$ depending on $\xi, \Theta,\bm\alpha,\bmc,d,\bm\gamma$ only such that for all $\bmx \in \R^d\backslash\{\bm 0\}$
\begin{equation}\label{eq:boundpartdevfin}
\Big| \partial_j \lambda(\bmx) \Big| \leq  C_{3,\xi} \id_{\{x_k \neq 0\}} x_k^{- 1- \alpha_k(1+\alpha_j^{-1} + \sum_{i=1, i\neq k}^d \alpha_i^{-1})} \prod_{i\neq k} \Big(\frac{x_i^{\alpha_i}}{x_k^{\alpha_k}} \wedge \frac{x_k^{\alpha_k}}{x_i^{ \alpha_i}} \Big)^\xi, \quad k = 1,\dots,d.
\end{equation}
In particular, $\partial_j \lambda(\bmx)$ is uniformly bounded on $B^\complement_\infty(\delta)$.
\end{enumerate}
\end{lem}

We will later need regularity properties of the density $\lambda$ of the measure $\Lambda^{(\bag,\Theta)}$ when a ball around $\bm 0$ is smoothly cut out of its support. To this purpose, pick a smooth (infinitely differentiable) function $s_\epsilon: \R^d \to [0,1]$ which has support $  B_\infty^\complement(\eps)$ and $s_\epsilon(\bmx)=1$ for all $\bmx\in B_\infty^\complement(2\eps)$. Define
\begin{equation}\label{def:h}
h_\epsilon:\R^d\to\R^d;\quad \bmx\mapsto s_\epsilon(\bmx)\lambda(\bmx)
\end{equation}
and note that $\int_{\R^d} h_\epsilon(\bmx)\rmd\bmx<\infty$, since $\lambda$ is a Lévy density and $s$ has support bounded away from $0$. 

\begin{lem}
\label{lem:regsmootheddens}
Assume $\alpha_i \in (0,2), i \in  V $, $\min_{i \in  V} c_i^+ \wedge c_i^- > 0$. The following is true.
\begin{enumerate}
    \item[$(i)$] $h_\epsilon$ is continuous, with uniformly bounded uni- and bivariate marginal densities.
   \item[$(ii)$] $h_\epsilon$ has continuous and uniformly bounded first order partial derivatives.
   \item[$(iii)$] For every $j\in \{1,\ldots,d\}$ there exists a function $M_j(\bmx_{\setminus j})$ such that 
   $$ \lv \partial_j h_\epsilon(\bmx)\rv \leq M_j(\bmx_{\setminus j}) $$
   with $\int_{\R^{d-1}} M_j(\bmx_{\setminus j})\rmd\bmx_{\setminus j} <\infty$ and $x_i\mapsto \int_{\R^{d-2}\setminus} M_j(\bmx_{\setminus j})\rmd\bmx_{\setminus i,j}$ is bounded on bounded sets for all $i\not=j$.
\end{enumerate}

The constants in the bounds depend only on $d,\epsilon,\bgamma,\bmc,\balpha,\Theta$.
\end{lem}

\subsubsection{Proofs for Section~\ref{sec:propHRLdens}}

\subsubsection*{Proof of Lemma~\ref{lem:tech1}}
\begin{proof}
First note that for $\eta := \min_k \lambda_{min}([\Sigma^{(k)}]^{-1}) > 0$ we have $\bmx^\top [\Sigma^{(k)}]^{-1} \bmx \geq \eta \|\bmx\|^2$. Thus
\begin{align*}
\varphi_{d-1}(\log(\bmy_{\backslash k}^{\bm \alpha_{\backslash k}}/y_k^{\alpha_k}) - \mathbf{v};\Sigma^{(k)}) &\lesssim \exp(-\eta \|\log(\bmy_{\backslash k}^{\bm \alpha_{\backslash k}}/y_k^{\alpha_k}) - \mathbf{v}\|^2)
\\
&= \prod_{i \neq k} \exp\big(-\eta (\log(y_{i}^{ \alpha_i}/y_k^{\alpha_k}) - v_i)^2 \big).
\end{align*}
We will bound each term in the product separately. Since all terms have the same structure it suffices to bound one. Note that $|t| > (2|a|)\vee (4\kappa/\eta)$ implies $\eta(t+a)^2 \geq \kappa |t|$, since $|t|>2|a|$ implies $|t+a| > |t|/2$, which implies 
\[
\eta(t+a)^2 \geq \eta(t/2)^2 = |t| (|t|\eta/4) \geq \kappa |t|.
\]
Thus for $|\log(y_{i}^{ \alpha_i}/y_k^{\alpha_k})| > (2\max_{j} |v_j|)\vee(4\kappa/\eta)$, we have
\[
\exp\big(-\eta (\log(y_{i}^{ \alpha_i}/y_k^{\alpha_k}) - v_i)^2 \big) \leq \Big(\frac{y_i^{\alpha_i}}{y_k^{\alpha_k}} \wedge \frac{y_k^{\alpha_k}}{y_i^{ \alpha_i}} \Big)^\kappa.
\]
since for $\log(y_{i}^{ \alpha_i}/y_k^{\alpha_k}) > 0$ we get the bound $(y_{i}^{ \alpha_i}/y_k^{\alpha_k})^{-\kappa}$ and for $\log(y_{i}^{ \alpha_i}/y_k^{\alpha_k}) < 0$ we get the bound $(y_{i}^{ \alpha_i}/y_k^{\alpha_k})^{\kappa}$.

For $|\log(y_{i}^{ \alpha_i}/y_k^{\alpha_k})| \leq (2\max_{j} |v_j|)\vee(4\kappa/\eta)$, the exponential term is bounded while $\Big(\frac{y_i^{\alpha_i}}{y_k^{\alpha_k}} \wedge \frac{y_k^{\alpha_k}}{y_i^{ \alpha_i}} \Big)^\kappa$ is bounded away from zero. In this case the the bound can be achieved by the choice of a proper constant. This completes the proof of~\eqref{eq:boundphi}.

\end{proof}

\subsubsection*{Proof of Lemma~\ref{lem:propdens}}
\begin{proof}
We first prove the bound in~\eqref{eq:boundf}. Due to the structure of $\lambda$, it suffices to consider $\bmy \in (0,\infty)^d$, all other orthants can be handled by similar arguments. Fix $\bmy \in (0,\infty)^d$. Let $\bm\mu = \bm\mu_\bmo^{(k)}$ for $\bmo$ corresponding to $(0,\infty)^d$. We have
\begin{align}
h_k(\bmy;\bm\mu,\Sigma^{(k)}) y_k^{-\alpha_k}\prod_{i=1}^d y_i^{-1} 
&= h_k(\bmy;\bm\mu,\Sigma^{(k)}) y_k^{-\alpha_k(1+\sum_{i=1}^d \frac{1}{\alpha_i})}  \prod_{i=1}^d \Big(\frac{y_k^{\alpha_k}}{y_i^{\alpha_i}}\Big)^{1/\alpha_i} \notag
\\
&\leq C_\kappa y_k^{-\alpha_k(1+\sum_{i=1}^d \frac{1}{\alpha_i})} \prod_{i=1}^d \Big(\frac{y_k^{\alpha_k}}{y_i^{\alpha_i}}\Big)^{1/\alpha_i} \Big(\frac{y_i^{\alpha_i}}{y_k^{\alpha_k}} \wedge \frac{y_k^{\alpha_k}}{y_i^{ \alpha_i}} \Big)^\kappa.  \label{eq:boundfpt1}
\end{align}
where the second line follows from~\eqref{eq:boundphi}. 
Picking $\kappa = \xi + \max_{i=1}^d 1/\alpha_i$, combined with similar arguments for all other orthants, yields~\eqref{eq:boundf} since $0\leq a \leq b$, 
\[
x^a(x\wedge x^{-1})^b = x^{a+b}\id_{\{x \in (0,1)\}} + x^{a-b}\id_{\{x \geq 1\}} \leq (x \wedge x^{-1})^{b-a} \leq 1.
\]
This shows that $\lambda$ is uniformly bounded on $B^\complement_\infty(\delta)$ for a bound that depends only on $\delta,\bm\alpha,\Theta,\bmc,\bm\gamma,d$: indeed, if $\bmy \in B^\complement_\infty(\delta)$ then there exists a $k$ such that $|y_k| \geq \delta$. Apply the bound in~\eqref{eq:boundf} with this $k$.

Next, we prove that $\lambda$ is continuous on $\R^d\backslash\{0\}$. Continuity is obvious at any point with no zero entries, and it remains to prove continuity on $Z_0^+$. Fix $\bmy \in Z_0^+$ and consider an arbitrary sequence $\bmy_n \to \bmy$. Then there exists a $k$ such that $y_k > 0$. Since $y_k > 0$, there exists a $\delta > 0$ such that for all $n$ sufficiently large $y_{n,k} \geq \delta$. For all such $n$, apply~\eqref{eq:boundf} with $\xi > \max_{i=1}^d 1/\alpha_i$ to obtain $\lambda(\bmy_n) \to 0$ since at least one entry of $\bmy$ is zero so that the product tends to zero. 

\end{proof}

\subsubsection*{Proof of Lemma~\ref{lem:partder1}}
\begin{proof}

\textit{Proof of (i)} Fix $\bmy \in Z_0$. If there exists $i \neq j$ so that $y_i = 0$, then $\lambda(\bmy + \bm e_jh) = 0$ for any $h \in \R$ such that $\bmy+\bm e_jh\neq \bm 0$, so $\partial \lambda(\bmy)/\partial y_j = 0$ in this case. It remains to consider the case $y_j = 0$, $y_i \neq 0$ for $i \neq j$. Then for $h>0$
\[
\frac{\lambda(\bmy + h\bm e_j) - \lambda(\bmy)}{h} = \frac{\lambda(\bmx)}{x_j}
\]
where $\bmx = \bmy + h \bm e_j$. Using the bound in~\eqref{eq:boundphi} with $\bmx$ in place of $\bmy$, with $k\neq j$, and arguing as we did when deriving~\eqref{eq:boundfpt1} this time with $\delta = y_k > 0$ and $\kappa > 1/\alpha_j + \max_{i=1}^d 1/\alpha_i$, we find that 
\[
\frac{\lambda(\bmy + h\bm e_j) - \lambda(\bmy)}{h} \to 0, \quad h \downarrow 0.
\]
Thus the right-side partial derivative of $\lambda$ at $\bmy$ exists and is zero. The same argument works for the left-hand partial derivative, showing that the $j$'th partial derivative of $\lambda$ at $\bmy$ exists and satisfies $\partial_j \lambda(\bmy)= 0$. 

\medskip

\textit{Proof of (ii)} First we compute the partial derivatives of $\lambda$ on $(0,\infty)^d$. Let $\bm\mu = \bm\mu_\bmo^{(k)}$ for $\bmo$ corresponding to $(0,\infty)^d$. Then, for $i \neq k$ we obtain
\begin{align*}
&\partial_i \Big[h_k(\bmy;\bm\mu,\Sigma^{(k)}) y_k^{-\alpha_k} \prod_{j=1}^d y_j^{-1}\Big]
\\
&= - h_k(\bmy;\bm\mu,\Sigma^{(k)}) y_k^{-\alpha_k} y_i^{-1} \prod_{j=1}^d y_j^{-1} + \Big[\partial_i h_k(\bmy;\bm\mu,\Sigma^{(k)}) \Big] y_k^{-\alpha_k} \prod_{j=1}^d y_j^{-1}
\\
&= - h_k(\bmy;\bm\mu,\Sigma^{(k)}) y_k^{-\alpha_k}\Big(\prod_{j=1}^d y_j^{-1}\Big) \Big[y_i^{-1} + \frac{\alpha_i}{y_i} \sum_{1 \leq j \leq d, j \neq k} \theta_{ij}(\log(y_j^{\alpha_i}/y_k^{\alpha_k}) - \mu_j) \Big].
\end{align*}
For $k=j$ we obtain
\begin{align*}
&\partial_k \Big[h_k(\bmy;\bm\mu,\Sigma^{(k)}) y_k^{-\alpha_k-1} \prod_{i\neq k} y_i^{-1}\Big]
\\
&= - (\alpha_k+1) h_k(\bmy;\bm\mu,\Sigma^{(k)}) y_k^{-\alpha_k-2} \prod_{i\neq k} y_i^{-1} + \Big[\partial_k h_k(\bmy;\bm\mu,\Sigma^{(k)}) \Big] y_k^{-\alpha_k} \prod_{i=1}^d y_i^{-1}
\\
&= - h_k(\bmy;\bm\mu,\Sigma^{(k)}) y_k^{-\alpha_k}\Big(\prod_{i=1}^d y_i^{-1}\Big) \Big[\frac{\alpha_k+1}{y_k} - \frac{\alpha_k}{y_k} \sum_{1 \leq i,\ell \leq d, i,\ell \neq k} \theta_{i\ell}(\log(y_i^{\alpha_i}/y_k^{\alpha_k}) - \mu_\ell \Big].
\end{align*}
There exists a universal constant $C$ such that $|\log(x)| \leq C(x + x^{-1})$. With this, we obtain that there exist constants $C$ depending only on $\bm\alpha, \Theta,d$ and $\tilde C_\kappa$ depending additionally on $\kappa$ such that for $\bmy \in (0,\infty)^d$ and $j =1,\dots,d$
\begin{align*}
&\Big| \partial_j \Big[h_k(\bmy;\bm\mu,\Sigma^{(k)}) y_k^{-\alpha_k-1} \prod_{i\neq k} y_i^{-1}\Big] \Big|
\\
&\leq C h_k(\bmy;\bm\mu,\Sigma^{(k)}) y_k^{-\alpha_k} y_j^{-1}  \Big(\prod_{i=1}^d y_i^{-1}\Big) \Big(1 + \sum_{1\leq i\leq d, i\neq k} \frac{y_i^{\alpha_i}}{y_k^{\alpha_k}} + \frac{y_k^{\alpha_k}}{y_i^{\alpha_i}} \Big)
\\
&= C h_k(\bmy;\bm\mu,\Sigma^{(k)}) y_k^{-\alpha_k(1+\alpha_j^{-1} + \sum_{i=1}^d \alpha_i^{-1})} \Big(\frac{y_k^{\alpha_k}}{y_j^{\alpha_j}}\Big)^{1/\alpha_j}  \Big(\prod_{i=1}^d \Big(\frac{y_k^{\alpha_k}}{y_i^{\alpha_i}}\Big)^{1/\alpha_i} \Big)
\\
&\quad \times\Big(1 + \sum_{1\leq i\leq d, i\neq k} \frac{y_i^{\alpha_i}}{y_k^{\alpha_k}} + \frac{y_k^{\alpha_k}}{y_i^{\alpha_i}} \Big)
\\
&\leq \widetilde C_\kappa y_k^{-\alpha_k(1+\alpha_j^{-1} + \sum_{i=1}^d \alpha_i^{-1})} \Big(\frac{y_k^{\alpha_k}}{y_j^{\alpha_j}}\Big)^{1/\alpha_j}  \Big(\prod_{i=1}^d \Big(\frac{y_k^{\alpha_k}}{y_i^{\alpha_i}}\Big)^{1/\alpha_i} \Big) \Big(1 + \sum_{1\leq i\leq d, i\neq k} \frac{y_i^{\alpha_i}}{y_k^{\alpha_k}} + \frac{y_k^{\alpha_k}}{y_i^{\alpha_i}} \Big)
\\
&\quad \times\Big( \prod_{i\neq k} \Big(\frac{y_i^{\alpha_i}}{y_k^{\alpha_k}} \wedge \frac{y_k^{\alpha_k}}{y_i^{ \alpha_i}} \Big)^\kappa \Big) 
\\
&=\widetilde C_\kappa y_k^{-\alpha_k(1+\alpha_j^{-1} + \sum_{i=1}^d \alpha_i^{-1})}  \Big(\frac{y_k^{\alpha_k}}{y_j^{\alpha_j}}\Big)^{1/\alpha_j}  \Big(\prod_{i=1}^d \Big(\frac{y_k^{\alpha_k}}{y_i^{\alpha_i}}\Big)^{1/\alpha_i} \Big) \Big(1 + \sum_{1\leq i\leq d, i\neq k} \frac{y_i^{\alpha_i}}{y_k^{\alpha_k}} + \frac{y_k^{\alpha_k}}{y_i^{\alpha_i}} \Big)
\\
&\quad \times \Big( \prod_{i\neq k} \Big(\frac{y_i^{\alpha_i}}{y_k^{\alpha_k}} \wedge \frac{y_k^{\alpha_k}}{y_i^{ \alpha_i}} \Big)^\kappa \Big).
\end{align*}
where we have used (\ref{eq:boundfpt1}) for the last inequality.

Next, observe that setting $z_j = y_k^{\alpha_k}/y_j^{\alpha_j}$ and noting that $z_k = 1$ we need to bound terms of the form
\begin{align*}
z_j^{1/\alpha_j} \Big(1 + \sum_{1\leq \ell\leq  d,\ell\neq k} (z_\ell + z_\ell^{-1})\Big) \prod_{i \neq k} z_i^{1/\alpha_i}(z_i \wedge z_i^{-1})^\kappa .    
\end{align*}
Recognizing that this is a sum of  terms of the form
\[
z_\ell^{b_l} \Big(\prod_{i\neq k} [z_i^{\tilde a_i} (z_i \wedge z_i^{-1})^\kappa]\Big) = \Big(\prod_{i\neq k} [z_i^{ a_i} (z_i \wedge z_i^{-1})^\kappa]\Big)
\]
for appropriate $ a_i\in\R$, we will bound each term in the sum by the same bound. Then we can bound the whole term by a constant multiple of the bound of the individual terms. 

The largest value of $a$ that can occur is $a_{max} = 1+2\max_i 1/\alpha_i$. We have for $0\leq a \leq a_{max} < \kappa$, noting that $x \wedge x^{-1} \leq 1$
\[
x^a(x\wedge x^{-1})^\kappa = x^{a+\kappa}\id_{\{x \in (0,1)\}} + x^{a-\kappa}\id_{\{x \geq 1\}} \leq  (x \wedge x^{-1})^{\kappa-a} \leq (x \wedge x^{-1})^{\kappa-a_{max}}.
\]
It could also happen that $a$ is negative, then the smallest possible value is $-1+\min_i 1/\alpha_i$. We have $|-1+\min_i 1/\alpha_i| \leq 1 < a_{max}$, hence a similar argument as above applies with $a_{min} \leq a \leq 0$. Thus setting $\kappa = a_{max} + \xi$ allows us to bound each term in the product by $(z_i \wedge z_i^{-1})^\xi$.

Thus, we obtain
\begin{equation}\label{eq:finboundpartder}
\Big| \partial_j \Big[h_k(\bmy) y_k^{-\alpha_k-1} \prod_{i\neq k} y_i^{-1}\Big] \Big| \leq C_{2,\xi} y_k^{- 1- \alpha_k(1+\alpha_j^{-1} + \sum_{i=1, i\neq k}^d \alpha_i^{-1})} \prod_{i\neq k} \Big(\frac{y_i^{\alpha_i}}{y_k^{\alpha_k}} \wedge \frac{y_k^{\alpha_k}}{y_i^{ \alpha_i}} \Big)^\xi
\end{equation}
where $C_{2,\xi}$ depends on $\Theta, \bm\alpha, d, \xi$ only. Note that this bound holds for any $k$, including $k=j$. The same argument works for all other open orthants, after possibly enlarging the constant. If $\bmy$ has at least one zero entry but is not identically zero, both the partial  derivative of $\lambda$ and the bound are zero. This completes the proof of~\eqref{eq:boundpartdevfin}. 
\end{proof}

\subsubsection*{Proof of Lemma~\ref{lem:regsmootheddens}}
\begin{proof}
Throughout the proof, fix $\xi > \max_i 1/\alpha_i$. We will repeatedly use the following elementary fact, which follows by substitution
\begin{align}
\int_\R \Big(\frac{|x|^{\alpha_i}}{|y_k|^{\alpha_k}} \wedge \frac{|y_k|^{\alpha_k}}{|x|^{ \alpha_i}} \Big)^\xi \rmd x 
&= 2 \int_0^\infty \Big(\frac{x}{|y_k|^{\alpha_k/\alpha_i}} \wedge \frac{|y_k|^{\alpha_k/\alpha_i}}{x}\Big)^{\xi\alpha_i} \rmd x \notag
\\
&= 2|y_k|^{\alpha_k/\alpha_i} \int_0^\infty x^{\xi\alpha_i}\wedge x^{-\xi\alpha_i} \rmd x = C_\xi(\bm\alpha) |y_k|^{\alpha_k/\alpha_i} \label{eq:helpint}
\end{align}
where the last equality follows by the choice of $\xi$.

\textit{Proof of part (i)} Continuity is obvious since $\lambda$ is continuous on $\R^d\setminus \{0\}$. Let $\lambda_{i,j}$ denote the Lévy density of the HR process corresponding to margin $\{i,j\}$. To obtain bounded bivariate marginal densities of $h_\epsilon$ observe that
\begin{align*}
\int_{\R^{d-2}}h_\epsilon(\bmx)\rmd\bmx_{\setminus i,j}&=  \id_{\{\vert x_i\vert\geq \eps\text{ or }\vert x_j\vert\geq\eps\}} \int_{\R^{d-2}}s_\eps(\bmx)\lambda(\bmx) \rmd\bmx_{\setminus i,j}
\\
&\quad \quad + \id_{\{\vert x_i\vert, \vert x_j\vert< \eps\}}\int_{\R^{d-2}\setminus B_\infty(\eps)}s_\epsilon(\bmx)\lambda(\bmx)\rmd\bmx_{\setminus i,j} 
\\
&\leq  \id_{\{ (x_i,x_j)\in B_\infty^\complement(\eps)\}}\lambda_{i,j}(x_i,x_j)+ \int_{\R^{d-2}\setminus B_\infty(\eps)}\lambda(\bmx)\rmd\bmx_{\setminus i,j}
\end{align*}
Since the multivariate margins of an \HR{} process are \HR{} processes, Lemma \ref{lem:propdens} with $d=2$ implies that $\lambda_{i,j}$ is uniformly bounded on $B_\infty^\complement(\epsilon)$. Therefore, it suffices to bound the last term. Note that
\begin{align*}
\int_{\R^{d-2}\setminus B_\infty(\eps)}\lambda(\bmx)\rmd\bmx_{\setminus\{i,j\}} &\leq \sum_{k \neq i,j} \int_{|x_k| > \eps}\int_{\R^{d-3}} \lambda(\bmx) \rmd x_{\setminus\{i,j,k\}} \rmd x_k
\end{align*}
and
\begin{align*}
&\int_{|x_k| > \eps}\int_{\R^{d-3}} \lambda(\bmx) \rmd x_{\setminus\{i,j,k\}}\rmd x_k 
\\
&\leq C_{0,\xi} \int_{|x_k| > \eps} |x_k|^{-\alpha_k(1+\sum_{i=1}^d \frac{1}{\alpha_i})} \int_{\R^{d-3}} \prod_{\ell\neq i,j,k} \Big(\frac{|x_\ell|^{\alpha_\ell}}{|x_k|^{\alpha_k}} \wedge \frac{|x_k|^{\alpha_k}}{|x_\ell|^{ \alpha_\ell}} \Big)^\xi \rmd x_{\setminus\{i,j,k\}} \rmd x_k 
\\
&= 2 C_{0,\xi} C_\xi(\balpha)^{d-3} \int_{\eps}^\infty x_k^{-1 - \alpha_k- \alpha_k/\alpha_i - \alpha_k/\alpha_j } \rmd x_k < \infty.
\end{align*}
where we used~\eqref{eq:boundf} in the second line and~\eqref{eq:helpint} in the third line. The argument for the univariate marginal density is similar.

\medskip

\textit{Proof of part (ii)} 
The partial derivatives of $h_\epsilon$ are given by
$$ 
\partial_jh_\epsilon(\bmx)= \lambda(\bmx)\partial_j s_\epsilon(\bmx)+s_\epsilon(\bmx)\partial_j \lambda(\bmx)
$$
By Lemma \ref{lem:partder1}(ii) $\lambda$ has uniformly bounded derivatives on $B_\infty^\complement(\eps)$, which implies that $s_\eps(\bmx)\partial_j \lambda(\bmx)$ is uniformly bounded as $s_\eps$ vanishes on $B_\infty(\eps)$. Moreover, $s_\eps$ is constant on $B_\infty^\complement(2\eps)$ which implies that $\partial_j s_\eps(\bmx)=0$ for all $\bmx\in B_\infty^\complement(2\eps)$ and for all $\bmx\in B_\infty(\eps)$. Since $s_\eps$ is smooth, $\partial_j s_\eps$ is continuous on $ B_\infty(2\eps)\setminus  B_\infty(\eps)$ and thus it is bounded. Moreover, $\lambda$ is bounded on $B_\infty(2\eps)\setminus  B_\infty(\eps)$ and therefore $\lambda(\bmx) \partial_j s_\eps(\bmx)$ is uniformly bounded. Continuity follows from the continuity of $\lambda$ and $\partial_j \lambda$ on $B_\infty^\complement(\eps)$.

\medskip

\textit{Proof of part (iii)}

 Observe that
\[
\partial_j h_\epsilon(\bmx)= \lambda(\bmx)\partial_j s_\epsilon(\bmx)+s_\epsilon(\bmx)\partial_j \lambda(\bmx).
\]
Since $\partial_j s_\eps(\bmx)=0$ for all $\bmx\in B_\infty^\complement(2\eps)$ and for all $\bmx\in B_\infty(\eps)$, and since by Lemma~\ref{lem:propdens} $\lambda(\cdot) \id_{\{\|\cdot\|_\infty \geq \eps\}}$ is uniformly bounded by a constant, say $C_{\lambda,\eps}$, we have
\begin{align*}
\Big| \lambda(\bmx)\partial_j s_\epsilon(\bmx)\Big| &\le \id_{\{\|\bmx\|_\infty \leq 2\eps\} }\Big(\sup_{\|\bmy\|_\infty \geq \eps} \lambda(\bmy)\Big) \sup_{\bmy \in \R^d} \Big|\frac{\partial}{\partial y_j}s_\epsilon(\bmy)\Big| 
\\
&\le C_{f,\eps} \id_{\{\|\bmx_{\setminus j}\|_\infty \leq 2\eps\}} \sup_{\bmy \in \R^d} \Big|\frac{\partial}{\partial y_j}s_\epsilon(\bmy)\Big|. 
\end{align*}
Let $C_{s,\lambda,\eps} := C_{\lambda,\eps} \sup_{\bmx \in \R^d} \Big|\partial_j s_\epsilon(\bmx)\Big|$. Then
\[
\Big|\partial_j h_\epsilon(\bmx)\Big| \le C_{s,\lambda,\eps}\id_{\{\|\bmx_{\setminus j}\|_\infty \leq 2\eps\}} + \Big|\id_{\{\|\bmx\|_\infty > \eps\}} \partial_j \lambda(\bmx)\Big|.
\]
Note that
\[
\id_{\{\|\bmx\|_\infty > \eps\}} \leq \id_{\{|x_j|\ge \eps\}}\id_{\{\|\bmx_{\setminus j}\|_\infty \leq \eps\}} + \sum_{m\neq j} \id_{\{|y_m| \ge \eps\}}.
\]
Then by~\eqref{eq:boundpartdevfin}
\[
\Big| \id_{\{|x_j|\ge \eps\}}\id_{\{\|\bmx_{\setminus j}\|_\infty \leq \eps\}} \partial_j \lambda(\bmx)\Big| \leq C_{3,\xi} \eps^{-1-\alpha_k(1+\alpha_j^{-1}+\sum_{i=1,i\neq k}^d\alpha_i^{-1})}\id_{\{\|\bmx_{\setminus j}\|_\infty \leq \eps\}}.
\]
Finally, for any $k \neq j$ by~\eqref{eq:boundpartdevfin}, noting that each term in the product is bounded from above by $1$ so we can drop terms,
\[
\Big| \id_{\{|y_k|\ge \eps\}} \frac{\partial}{\partial y_j}\lambda(\bmy)\Big| \leq C_{3,\xi} \id_{\{|y_k| \geq \eps\}}  y_k^{- 1- \alpha_k(1+\alpha_j^{-1} + \sum_{i=1, i\neq k}^d \alpha_i^{-1})} \prod_{i\neq k,j} \Big(\frac{|y_i|^{\alpha_i}}{|y_k|^{\alpha_k}} \wedge \frac{|y_k|^{\alpha_k}}{|y_i|^{ \alpha_i}} \Big)^\xi.
\]
In summary, defining
\[
C_{4,\eps,\xi} := C_{3,\xi} \eps^{-1-\alpha_k(1+\alpha_j^{-1}+\sum_{i=1,i\neq k}^d)\alpha_i^{-1}} + C_{s,\lambda,\eps}
\]
we have
\begin{align*}
&\Big|\frac{\partial}{\partial y_j}h_\epsilon(\bmy)\Big| 
\\
\le 
&C_{4,\eps,\xi} \id_{\{\|\bmy_{\setminus j}\|_\infty \leq 2\eps\}} + \sum_{k \neq j} C_{3,\xi} \id_{\{|y_k| \geq \eps\}}  y_k^{- 1- \alpha_k(1+\alpha_j^{-1} + \sum_{i=1, i\neq k}^d \alpha_i^{-1})} \prod_{i\neq k,j} \Big(\frac{|y_i|^{\alpha_i}}{|y_k|^{\alpha_k}} \wedge \frac{|y_k|^{\alpha_k}}{|y_i|^{ \alpha_i}} \Big)^\xi.
\end{align*}
Define the right-hand side above as $M_j(\bmy_{\setminus j})$.

Now we prove that $M_j$ has the claimed properties. The term $C_{4,\eps,\xi}\id_{\{\|\bmy_{\backslash j}\|_\infty \leq 2\eps\}}$ causes no problems, so we focus on one term from the sum. By the choice of $\xi$ and~\eqref{eq:helpint}
\begin{align*}
&\int_\eps^\infty \int_{\R^{d-2}} y_k^{- 1- \alpha_k(1+\alpha_j^{-1} + \sum_{i=1, i\neq k}^d \alpha_i^{-1})} \prod_{i\neq k,j} \Big(\frac{|y_i|^{\alpha_i}}{|y_k|^{\alpha_k}} \wedge \frac{|y_k|^{\alpha_k}}{|y_i|^{ \alpha_i}} \Big)^\xi d\bmy_{\setminus\{k,j\}}dy_k
\\
=~& C_\xi(\bm\alpha)^{d-2} \int_\eps^\infty y_k^{-1-2\alpha_k\alpha_j^{-1}} dy_k < \infty. 
\end{align*}
This proves that $\int_{\R^{d-1}} M_j(\bmy_{\backslash j}) d\bmy_{\backslash j} < \infty$. 
Boundedness of $x_i \mapsto \int_{\R^{d-2}} M_j(\bmy_{\backslash j}) d \bmy_{\backslash \{i,j\}} $ on bounded sets (and, in fact, globally) follows by a similar computation, again noting that $\Big(\frac{|y_j|^{\alpha_i}}{|y_k|^{\alpha_\ell}} \wedge \frac{|y_\ell|^{\alpha_k}}{|y_j|^{ \alpha_i}} \Big)^\xi \leq 1$. 
\end{proof}
%%%%%%%%%%%%%%%%%%%%%%%%%%%%%%%%%%%%%%%%%%%%%%%%%%%%%%%%%%%%%%%%%%%%%%%%%%%%%%%%%%%%%%%%%%%

\newpage
%%%%%%%%%%%%%%%%%%%%%%%%%%%%%%%%%%%%%%%%%%%%%%%%%%%%%%%%%%%%%%%%%%%%%%%%%%%%%%%%%%%%%%%%%%%
\section{A uniform bound for the small time behavior of a multivariate Lévy process}
\label{app_proof_univ_conv_levy_measure}

The goal of this section is to prove a general uniform (in $t$ and $\bmy$) bound of the form
$$ \Big\vert P\lc  \g L(t) \in [\bmy,\bm\infty) \rc -t\nu\lc [\bmy,\bm\infty)  \rc \Big\vert \leq Ct^2.
$$
Our results can be viewed as an extension and correction of Lemma 4.3 in \cite{buechervetter2013}, who derive a concentration inequality for the small time behavior of bivariate Lévy processes with only positive jumps, which is in turn a generalization of uni-variate results of \cite{figueroahoudre2009}. We first provide some details on the issues that we found in the proof of Lemma 4.3 \cite{buechervetter2013}, which were confirmed in correspondence with the authors. Specifically, \cite{buechervetter2013} consider the following regularity conditions.

\begin{ass}[Assumption 4.1 of \cite{buechervetter2013}]\label{ass:bv}
Let $\mathbf{X}$ be a bivariate L\'evy process with Lévy measure $\nu$. 
The following assumptions on $\nu$ are in order:
\begin{enumerate}
    \item[(i)] $\nu$ has support $[0, \infty)^2 \setminus \{(0,0)\}$.
    \item[(ii)] On this set it takes the form $\nu(d\mathbf{u}) = s(\mathbf{u})\, d\mathbf{u}$ for a positive L\'evy density $s$ 
    which satisfies
    \[
    \sup_{\mathbf{u} \in M_\eta} \left( |s(\mathbf{u})| + \|\nabla s(\mathbf{u})\| \right) < \infty
    \]
    for any $\eta \in (0,\infty)^2$, where
    \[
    M_\eta = (\eta, \infty)^2 \cup \big( \{0\} \times (\eta, \infty) \big) \cup \big( (\eta, \infty) \times \{0\} \big),
    \]
    and $\nabla s$ denotes the gradient of $s$ on $(\eta, \infty)^2$ and the univariate derivative on the stripes through $0$, respectively.
\end{enumerate}
\end{ass}

Under these assumptions they claim the following result.

\begin{lem*}[Lemma 4.3 of \cite{buechervetter2013}]
Suppose that (i) and (ii) of Assumption~\ref{ass:bv} hold and let $\delta > 0$ be fixed. 
Then there exist constants $K = K(\delta)$ and $t_0 = t_0(\delta)$ such that the uniform bound
\[
\left| \mathbb{P}\left( X_t^{(1)} \ge x_1,\, X_t^{(2)} \ge x_2 \right) 
 - t\, \nu\big( [x_1, \infty) \times [x_2, \infty) \big) \right| < K t^2
\]
holds for all $\mathbf{x} = (x_1, x_2) \in [\delta, \infty)^2 
\cup \big( \{-\infty\} \times [\delta, \infty) \big) 
\cup \big( [\delta, \infty) \times \{-\infty\} \big)$ 
and $0 < t < t_0$.
\end{lem*}

However, it turns out that the proof of \cite{buechervetter2013} is flawed. Specifically, Assumption 4.1 does not suffice to ensure that the function $g$ in equation (A1) from \cite{buechervetter2013} is twice continuously differentiable, which is needed to apply Ito's lemma in the version that is used in (A1) in the latter paper. Further details are described in Section~\ref{app_counterexample} below. In order to prove our Theorem \ref{thmunifconvlevymeasure}, which corrects and generalizes Lemma 4.3 in \cite{buechervetter2013}, we needed to strengthen the requirements on the Lévy density to include twice continuous differentiability and impose the existence of certain majorants, which resulted in Assumption~\ref{assmpsmootheneslevydens}.

\subsection{Additional details on issues in the proof of Lemma 4.3 in \cite{buechervetter2013}}
\label{app_counterexample}

Let $\g L$ denote a bivariate Lévy process without drift and Gaussian component and Lévy measure concentrated on $[0,\infty)^2$. Define 
$f(x,y)=h(x)g(y)$,
where $g(y)$ is a smooth and bounded probability density with bounded derivatives of a random variable on $[1,\infty)$. Note that we can choose $g$ such that $\lim_{x\to 1}g(x)=0$. Further define
\begin{align*}
  h(x)&=\id_{\{x\geq 1\}}\exp(-(x-1))(x-1)^2\sin(1/(x-1))^2 = \id_{\{x\geq 1\}}\exp(-(x-1)) \int_{1}^{x} q^\prime(z)\rmd z 
\end{align*}
where
\[
q^\prime(z) =
\begin{cases} 
2(z-1) \sin\lc 1/(z-1)\rc^2- 2\sin(1/(z-1))\cos\lc 1/(z-1)\rc & \text{if } z > 1\\
0 & \text{if } z \leq 1
\end{cases}
\]
It is easy to see that $C:=\int_0^\infty h(x)\rmd x<\infty$, $\lim_{x\to 1}h(x)=0$ and that $h(x)$ is bounded and differentiable with bounded derivative $-h(x)+ \exp(-(x-1))q^\prime(x)$ everywhere. It is important to note that $h$ has discontinuous derivative at $1$ as $q^\prime$ is not continuous at $1$. 
This implies that $f$ is continuous, bounded and its partial derivatives exist and are bounded. The same is true for the margins $f_1(x)=h(x)$ and $f_2(y)=Cg(y)$ obtained by integrating out the second and first components, respectively. Therefore, the gradients of $f$ and of its marginal densities $f_1$ and $f_2$ exist and are bounded in the sense that 
\begin{align}\label{lemma_bound}
   \forall\delta>0:\  \sup_{(x,y)\in [\delta,\infty)^2} (|f(x,y)|  + \| \nabla f(x,y) \| ) < \infty,
\end{align}
as required in Assumption 4.1(ii) in \cite{buechervetter2013}. 

Let $s$ be a bivariate smooth Lévy density that satisfies the assumptions of \cite[Lemma 4.3]{buechervetter2013}, i.e., $s$ has support on $[0,\infty)^2\setminus \{\mathbf 0\}$ and its marginal Lévy densities $s_1$ and $s_2$ satisfy the bound~\eqref{lemma_bound}. Assume further that 
\[
\frac{\partial}{\partial x} \int_{v \geq 1-t} s(1-x,v) dv =  \int_{v \geq 1-t} \frac{\partial}{\partial x} s(1-x,v) dv
\]
for all $t, x<1/2$.

Then $s+f$ is a Lévy density which satisfies these assumptions as well. But $s+f$ has discontinuous partial derivative $\frac{\partial}{\partial x}s(x,y)+ h^\prime(x)g(y)$ and therefore, for every $\epsilon<1/2$ and every smooth  $c_\epsilon$ such that
$$\id_{\{\bmx \in B^\infty(\epsilon/2)\}} \leq c_{\epsilon}(\bmx)\leq \id_{\{ \bmx\in B^\infty(\epsilon)\}}$$
the function
$$ G(x,y)=\int_{\bmz\geq \bm 1 -(x,y)} (1-c_\epsilon(z_1,z_2))\lc s(z_1,z_2)+f(z_1,z_2)\rc \rmd z_1\rmd z_2 $$
has second order partial derivative for $x,y<1/2$ given by
\begin{align*}
    \frac{\partial^2}{(\partial x)^2}G(x,y)&= -\int_{z_2\geq 1- y} \frac{\partial}{\partial x}s(1-x,z_2)+\frac{\partial}{\partial x}h(1-x)g(z_2) \rmd z_2\\
    &=-\int_{z_2\geq 1-y} \frac{\partial}{\partial x}s(1-x,z_2) \rmd z_2 -  \frac{\partial}{\partial x}h(1-x)\int_{z_2\geq 1- y}g(z_2) \rmd z_2 
\end{align*}  which is not continuous at $(x,y)=\bm 0$ because of the discontinuity of $\frac{\partial}{\partial x}h$ at 1.

This precludes an application of Ito's lemma to the function $G$ ($g$ with $\bmx=\bm 1$ in their notation) in Equation A.1 in the proof of Lemma 4.3 in \cite{buechervetter2013} as all standard versions of Ito's lemma require the continuity of the second order derivative, see e.g.\ Theorem 4.4.7 in \cite{applebaum2009}. This shows that the conditions stated in Lemma 4.3 in \cite{buechervetter2013} do not suffice to prove their claim. The problem is not easily fixable by resorting to a variants of Ito's lemma that does not require continuity of the derivative, since those formulas deviate from the standard Ito formula through the appearance of non-standard terms that would require additional arguments that we were unable to produce.

Thus, we have opted to strengthen the requirements on the Lévy density to be able to apply Ito's formula in the proof of our Theorem \ref{thmunifconvlevymeasure}, which corrects and generalizes Lemma 4.3 in \cite{buechervetter2013}.

\subsection{A general regularity transfer from the density to the CDF}

Here, we present a general criterion under which certain smoothness properties of the density transfer to its corresponding CDF and lower-dimensional densities.

\begin{lem}
\label{lemdiffintegrateddens}
Assume that $f: \R^d \to [0,\infty)$ is a Lebesgue integrable function such that $\int_{\R^d} f(\bmy)\rmd \bmy<\infty$. If $d=1$ assume that $f$ is bounded and continuously differentiable with bounded derivative. If $d\geq 2$ assume that
\begin{enumerate}
\item[(i)] $f$ has uniformly bounded univariate and bivariate margins $\lc f_i\rc_{\leqd}$ $\lc f_{i,j}\rc_{1\leq i\not=j\leq d}$ where $f_i(x) := \int_{\R^{d-1}} f(\bmx) d\bmx_{\backslash i}$ and $f_{i,j}(x) := \int_{\R^{d-2}} f(\bmx) d\bmx_{\backslash \{i,j\}}$.
\item[(ii)] $f$ is continuous  and $x_i\mapsto \partial_i f(\bmx)$ is continuous for all $\leqd$.
\item[(iii)] For every $a \in \R$ and $j=1,\dots,d$ there exists $\eps(a,j)>0$ and a function $M_{j,a,\epsilon}: \R^{d-1} \to [0,\infty)$  such that $\int_{\R^{d-1}} M_{j,a,\epsilon}(\bmx_{\setminus j}) \rmd\bmx_{\setminus j} < \infty$ and for all $1\leq i\not=j\leq d$  $y_i \mapsto \int_{\R^{d-2}} M_{j,a,\epsilon}(\bmx_{\setminus j}) \rmd \bmx_{\setminus  i,j}$ is bounded on bounded sets and
\[
\sup_{|x_j-a|\leq \eps} \Big|\partial_j f(x_1,\ldots,x_d) \Big| \leq M_{j,a,\epsilon}(\bmx_{\backslash j}).
\]
\item[(iv)] For all $1\leq j\leq d$ $x_j \mapsto \int_{\R^{d-1}} \lv\partial_j f(\bmx)\rv \rmd\bmx_{\setminus j}$ is uniformly bounded.
\end{enumerate}
Then for all $I\subset\{1,\ldots,d\}$
 $$ F_I: \R^{\vert I\vert}\to[0,\infty); \quad  \bmx \mapsto \int_{\{\bmy\in \R^{d}\mid \bmy_I\leq \bmx\}} f(\bmy)\rmd \bmy$$
 is twice continuously differentiable with uniformly bounded first and second order partial derivatives.
\end{lem}

\begin{proof}[Proof of Lemma~\ref{lemdiffintegrateddens}]

The claim is trivial for $d=1$. Therefore, assume $d \geq 2$ and observe that $F_I(\bmx_I)=F_{ V}(\tilde{\bmx}_I)=:F(\tilde(\bmx)_I)$, where $\tilde{x}_I:=\lc x_i \id_{\{i\in I\}}+\infty\id_{\{i\not\in I\}}\rc_{\leqd}$. Thus, we can prove the general claim by proving that 
 $$ F(\bmx): (-\infty,\infty]^d \to [0,\infty); F(\bmx)= \int_{\{\bmy\in \R^{d}\mid \bmy\leq \bmx\}} f(\bmy)\rmd \bmy$$
is twice continuously differentiable with bounded first and second order partial derivatives in its real-valued arguments. Since the assumptions and statements are equi-variant under permuting the entries of vectors, we will without loss of generality only consider partial derivatives w.r.t.\ $x_1, x_2\in\R$, whereas the case that only one $x_i$ is finite is also implicitly covered by only considering those parts of the proof that are relevant to this claim.

We start with proving the existence of the first order partial derivatives. For any fixed $\bmz \in (-\infty,\infty]^{d-1}$, the function
\[
x \mapsto h_{\bmz}(x) = \int_{\bmy \leq \bmz} f(x,\bmy) d\bmy
\]
is uniformly bounded by $f_1(x)$ since $f$ is non-negative. It is also continuous since there exists an $\eps>0$ such that for $|x-x'|\leq \eps$ 
\[
|h_{\bmz}(x)-h_{\bmz}(x')| \leq \int_{\bmy \leq \bmz} |f(x,\bmy) - f(x',\bmy)| d\bmy \leq |x-x'| \int_{\R^{d-1}} M_{1,x,\epsilon}(\bmy) d\bmy,
\]
by the mean value theorem.
Since $F(x,\bmz) = \int_{(-\infty,x]} h_{\bmz}(t) dt$, the existence of the first partial derivative follows by the (Lebesgue version of) the fundamental theorem of calculus and we have $\partial_1 F(x,\bmz) = \int_{\bmy \leq \bmz} f(x,\bmy) d\bmy$.

We now prove that $x \mapsto \partial_1 F(x,\bmz)$ is again differentiable for any fixed $\bmz$. This follows by the Leibniz rule for differentiating under the integral sign: by assumption, the function $x \mapsto f(x,\bmy)$ is differentiable for every $\bmy$ and there exists some $\eps>0$ such that the derivative satisfies 
\[
\sup_{\vert x^\prime -x\vert<\eps}\Big|\partial_1 f(x,\bmy)\Big| \leq M_{1,x,\eps}(\bmy) 
\]
which has finite integral by assumption. In particular, we obtain
\[
\partial_{1,1} F(x,\bmz) = \int_{\bmy \leq \bmz} \partial_1 f(x,\bmy) d\bmy.
\]
Moreover, 
$$\lv \partial_{1,1} F(x,\bmz)\rv\leq \int_{\R^{d-1}}\lv \partial_1 f(x,\bmy) \rv \rmd\bmy$$
is uniformly bounded in $x$ by assumption. Therefore, $\partial_ {1,1}F(x,\bmz)$ is uniformly bounded.

To compute $\partial_{2,1} F(x,y,\bmz)$, we again apply the Lebesgue version of the fundamental theorem of calculus. Fix $x \in \R ,\bmz \in (-\infty,\infty]^{d-2}$. Define $g_{x,\bmz}(t) := \int_{\bmy \leq \bmz} f(x,t,\bmy) d\bmy$ and observe that $\partial_1 F(x,t,\bmz) = \int_{(-\infty,t]} g_{x,\bmz}(s) ds$. If we can show that $s \mapsto g_{x,\bmz}(s)$ is continuous everywhere, it will follow that $g_{x,\bmz}(t)$ is equal to $\partial_{2,1} F(x,t,\bmz)$. Since there exists some $\eps=\eps(s,2)$ such that for all $|a_n| \leq \eps$ and $\tilde s_n = \tilde s_n(x,\bmy)$ between $s$ and $s+a_n$ the mean-value theorem implies
\begin{align*}
|g_{x,\bmz}(s+a_n)-g_{x,\bmz}(s)| &= \Big|\int_{\bmy \leq \bmz} f(x,s+a_n,\bmy) - f(x,s,\bmy) d\bmy\Big|
\\
& \leq |a_n| \int_{\bmy \leq \bmz} \lv \partial_2 f(x,\tilde s_n(x,\bmy),\bmy) \rv \rmd\bmy
\\
&\leq |a_n| \int_{\bmy \leq \bmz} M_{2,s,\eps}(x,\bmy) \rmd\bmy \leq |a_n| \int_{\R^{d-2}} M_{2,s,\eps}(x,\bmy) \rmd\bmy.
\end{align*}
Since $x\mapsto  \int_{\R^{d-2}} M_{2,s,\eps}(x,\bmy) \rmd\bmy $ is finite continuity of $s \mapsto g_{x,\bmz}(s)$ follows.
This proves differentiability of $t \mapsto \frac{\partial}{\partial_1}F(x,t,\bmz)$ with respect to $t$ and shows
\[
\partial_{2,1} F(x,t,\bmz) = \int_{\bmy \leq \bmz} f(x,t,\bmy) d\bmy \leq f_{1,2}(x,t) 
\]
is uniformly bounded

It remains to establish continuity of all second order partial derivatives, which then implies that $F$ is twice continuously differentiable. Starting with second order partial derivatives with respect to the same component choose a sequence $(x_n,\bmz_n) \to (x,\bmz)$ and observe 
\begin{align*}
\partial_{1,2} F(x,\bmz) -\partial_{1,1} F(x_n,\bmz_n) &= \int_{\{\bmy \leq \bmz_n\}} \partial_1 f(x,\bmy) - \partial_1 f(x_n,\bmy) d\bmy
\\
&\quad  + \int_{\{\bmy \leq \bmz, \bmy \not \leq \bmz_n\}} \partial_1 f(x,\bmy) \rmd\bmy
\\
&\quad - \int_{\{\bmy \leq \bmz_n, \bmy \not\leq \bmz\}} \partial_1 f(x,\bmy) \rmd\bmy.
\end{align*}
Assuming that $|x-x_n|\leq \eps(x,1)$, all integrands are dominated in absolute value by $2M_{1,x,\eps}(\bmy)$. Moreover, using that $x\mapsto\frac{\partial}{\partial_1}f(x,\bmz)$ is continuous, all integrands converge to zero Lebesgue almost-everywhere. Thus, the right-hand side tends to zero by dominated convergence and continuity follows.

For the continuity of mixed second order partial derivatives, note that
\begin{align*}
\partial_{2,1} F(x,t,\bmz) - \partial_{2,1} F(x_n,t_n,\bmz_n) 
&= \int_{\{\bmy \leq \bmz_n\}} (f(x,t,\bmy) - f(x_n,t_n,\bmy)) \rmd\bmy
\\
&\quad  + \int_{\{\bmy \leq \bmz, \bmy \not \leq \bmz_n\}} f(x,t,\bmy) \rmd\bmy
\\
&\quad - \int_{\{\bmy \leq \bmz_n, \bmy \not\leq \bmz\}} f(x,t,\bmy) \rmd\bmy.     
\end{align*}
The integrand in the second integral is bounded by $f(x,t,\bmy)$ which is non-negative and integrates to $f_{1,2}(x,t) < \infty$. Since the indicator converges pointwise to the indicator of a Lebesgue null set when $\bmz_n\to \bmz$, the integral converges to zero by dominated convergence. The same argument applies to the third integral. For the first integral, note that by the multivariate mean-value theorem we have
\[
f(x_n,t_n,\bmy) - f(x,t,\bmy) = (x_n-x)\partial_1 f(\tilde x_n,\tilde t_n,\bmy) + (t_n -t) \partial_2 f(\tilde x_n,\tilde t_n,\bmy)
\]
where the point $(\tilde x_n, \tilde t_n)$ lies on the line segment connecting $(x,t)$ and $(x_n,t_n)$. Assuming that $|x_n-x| \leq \eps(1,x)=:\eps_1, |t_n-t| \leq \eps(2,t)=:\eps_2$, this expression is upper bounded by $|x-x_n|M_{1,x,\eps_1}(\tilde t_n,\bmy) + |t-t_n| M_{2,t,\eps_2}(\tilde x_n,\bmy)$. Since $\tilde x_n, \tilde t_n$ stay in a bounded set, we obtain 
\begin{multline*}
\Big|\int_{\{\bmy \leq \bmz_n\}} (f(x,t,\bmy) - f(x_n,t_n,\bmy)) \rmd\bmy \Big|
\\
\leq |x-x_n|\int_{\R^{d-2}} M_{1,x,\eps_1}(\tilde t_n,\bmy) \rmd\bmy + |t-t_n| \int_{\R^{d-2}} M_{2,t,\eps_2}(\tilde x_n,\bmy) \rmd\bmy \to 0
\end{multline*}
since both integrals are bounded by assumption. This completes the proof of continuity of the mixed partial derivative. 
\end{proof}

In particular, the smoothed version of the \HRL{} density $h_\eps$ satisfies the conditions of Lemma \ref{lemdiffintegrateddens}.
\begin{prop}
\label{cordensregHR}
Assume $\alpha_i \in (0,2), i \in  V $, $\min_{i \in  V} c_i^+ \wedge c_i^- > 0$. The function $h_\eps(\cdot)$ defined in~\eqref{def:h} satisfies the conditions of Lemma \ref{lemdiffintegrateddens} for every $\eps>0$.
\end{prop}

\begin{proof}[Proof of Proposition~\ref{cordensregHR}]
Observe that Lemma \ref{lem:regsmootheddens}$(i),(ii)$ are exactly conditions $(i)$ and $(ii)$ of Lemma \ref{lemdiffintegrateddens}. Lemma \ref{lem:regsmootheddens}$(iii$) implies that for all $1\leq j\leq d$, $a\in\R$ and $\delta>0$ 
$$ 
\sup_{|x_j-a|\leq \delta} \Big|\partial_j  h_\epsilon(x_1,\ldots,x_d) \Big| \leq M_{j}(\bmx_{\setminus j}
)
$$
with $x_i\mapsto \int_{\R^{d-2}} M_j(\bmx_{\setminus j})\rmd \bmx_{\setminus i,j}$ being bounded on bounded sets for every $i\not=j$. Thus, $h_\epsilon$ satisfies condition $(iii)$ of Lemma \ref{lemdiffintegrateddens}. Further,
 $$x_j \mapsto \int_{\R^{d-1}} \lv\partial_j  h_\epsilon(\bmx)\rv \rmd\bmx_{\setminus j} \leq \int_{\R^{d-1}} M_j(\bmx_{\setminus j})\rmd\bmx_{\setminus j}<\infty$$
 is uniformly bounded, which verifies condition (iv) of Lemma \ref{lemdiffintegrateddens}
\end{proof}

\subsection{ Proof of Theorem \ref{thmunifconvlevymeasure} and Proposition \ref{propunifconvHRprocess}}

Finally, we are ready to prove the main results of this section.

\begin{proof}[Proof of Theorem \ref{thmunifconvlevymeasure}]

The general structure of our proof follows the proof of \citet[][Lemma 4.3]{buechervetter2013} with the appropriate modifications to obtain the more general statements of Theorem \ref{thmunifconvlevymeasure} and corrections to account for the fact that their Lemma 4.3 was not proved as stated.

Let us denote the sets of rectangles consisting of solely unbounded open and closed intervals as
\begin{align*}
\mathcal{V}_d &:= \Big\{ \times_{i=1}^d S_i(y_i) \ \Big\vert\ S_i(y_i) \in \{[y_i,\infty), (y_i,\infty), (-\infty,y_i), (-\infty,y_i], \R\} \Big\}
\end{align*}
and the set of rectangles such that at least one half open interval does not intersect with $[-\delta,\delta]^\complement$ as
\begin{align*}
\mathcal{V}_{d,\delta} &:= \Big\{ \times_{i=1}^d S_i(y_i) \in \Vv_d \ \Big\vert\ \exists i: S_i(y_i) \subset [-\delta,\delta]^\complement \Big\}.
\end{align*}
Rectangles $R \subset B_\infty^\complement(\delta)$ of the form $W \times V$, $V \in \mathcal{V}_{d-1}$ for an arbitrary interval $W$ can be obtained from sets in $\mathcal{V}_{d,\delta}$ by at most one operation of the form $A\setminus B$ where $B \subset A$. Such rectangles can in turn be used to construct rectangles of the form $W_1 \times W_2 \times V \subset B_\infty^\complement(\delta)$ where $V \in \Vv_{d-2}$. Iterating this for at most $d$ steps we can obtain arbitrary rectangles which are subsets of $B_\infty^\complement(\delta)$. Since for arbitrary measures we have $\mu(A\setminus B) = \mu(A)-\mu(B)$ whenever $B \subset A$, $\mu(B) < \infty$ and there are finitely many options to construct $R$ via the described procedure it suffices to prove
$$ 
\sup_{V \in \Vv_{d,\delta}} \Big\vert P\lc \g L(t) \in V \rc -t\nu(V )  \Big\vert \leq C t^2. 
$$
All sets $V$ of the above form can in turn be obtained from sets of the form $[\bmy,\bm \infty)$ with $y_1 > \delta$ by reflections of components, reorderings, or limits of nested sequences of sets. Thus, it suffices to prove that
$$ 
\sup_{\bmy\in M_{1,\delta}}\Big\vert P\lc \g L(t) \geq \bmy \rc -t\nu\lc [\bmy,\bm \infty)  \rc \Big\vert \leq C t^2, 
$$
with $M_{1,\delta}:= \{\bmy\in [-\infty,\infty)^d \mid y_1>\delta\}$ for an arbitrary Lévy process $\bmL$ which satisfies the assumption of Theorem \ref{thmunifconvlevymeasure},  since reflections and reorderings of components of a Lévy process again yield a Lévy process whose Lévy density satisfies Assumption~\ref{assmpsmootheneslevydens} and there are only finitely many ways to permute and reflect the components. Therefore, from now on, w.l.o.g.\ we assume that $\bmy\in M_{1,\delta}$ and set $I_{<\infty}=\{i\in V \mid y_i\in \R\}$.
    
Denote the triplet of $\g L$ by $(\bmb, \Sigma,\Lambda)$. Next, choose $\epsilon<\min\{\delta/2;1\}$ such that additionally $B_\infty(\eps) \subset B_2(1)$ and define two independent Lévy processes $\g X$ and $\g Y$ via the triplets $(\bmb-\bmb_\epsilon,\Sigma,c_\epsilon(\bmx)\Lambda(\rmd \bmx))$ and $(\bmb_\epsilon,0,(1-c_\epsilon(\bmx))\Lambda(\rmd\bmx))$, where $c_\epsilon\in S_\eps$ [recall the definition of $S_\eps$ in~\eqref{S_eps}] is chosen such that $(1-c_\epsilon) \lambda$ satisfies the conditions of Lemma \ref{lemdiffintegrateddens}, and $\bmb_\epsilon=\int_{B_2(1)} \bmx (1-c_{\epsilon}(\bmx))\lambda(\bmx)  \rmd \bmx$. Then, $\g L\overset{d}{=} \g X+ \g Y$ where $\g Y\overset{d}{=} \sum_{i=1}^{N_t} \bm \xi_i$ is a compound Poisson process whose Lévy measure has total mass $K_\epsilon:=\int_{\R^d} (1-c_{\epsilon}(\bmx))\lambda(\bmx)\rmd \bmx$ and $N_t\overset{d}{=} Poi(tK_\epsilon)$ and $\g X$ has jumps bounded by $1$. 

We first bound $P\lc  \vert X_1(t) \vert \geq y  \rc$ for all $y >\delta$. To this end, we apply \cite[Lemma 2]{ruschendorf2002}. Note that the results in the latter reference are for Lévy processes $Z$ with characteristic function of the form
\begin{equation}\label{eq:ruewoernCF}
\e[\exp(izZ(t))] = \exp\Big(itz\mu + tz^2\sigma^2/2 + t \int_{\R\setminus\{0\}} \Big( e^{izu} - 1 - \frac{izu}{1+u^2} \Big) \eta(\rmd u)  \Big).
\end{equation}
Define $\eta(\rmd u)$ as the univariate measure corresponding to the first margin of $c_\eps(u) \Lambda(\rmd u)$ and let 
\[
\tilde b_1 :=  b_1-b_{\eps,1}  + \int_{\R\setminus\{0\}} u\lc\bm 1_{|u|\leq1} - \frac{1}{1+u^2}\rc \eta(du).
\]
The univariate process $X_1$ with triplet $(\tilde b_1, \sigma^2, \eta(\rmd u))$ has a characteristic function as in~\eqref{eq:ruewoernCF}, so the results from \cite[Lemma 2]{ruschendorf2002} are applicable in our setting (with possible changes to some of the constants) since the drift $\bmb$ is arbitrary. A close look at the proof of their result reveals that it suffices to pick any $A$ (in their notation) such that the support of $\eta$ is contained in $[-A,A]$. In our case, set $A = \delta/2$ and we can pick (in their notation) $a = 4/\delta$. Then for $t<1$ and $y \geq \delta$ we have $ay > 2$, thus $t^{ay} \leq t^2$ and so 
\[
P\lc  \vert X_1(t) \vert \geq y  \rc \leq C_1t^2
\]
since the function $x \mapsto \exp(4x\delta^{-1}(1-\log x))$ is uniformly bounded on $[\delta, \infty)$. Therefore, for all $\bmy\in M_{1,\delta}$ we have 
    \begin{align}
        P\lc  \g X(t) \geq \bmy \rc \leq P\lc \vert X_1(t)\vert \geq \delta \rc \leq C_1t^2  \label{eqndecomp1}
    \end{align} 
    for all $t\in(0,1)$.

    Next, conditioning on the number of jumps of $Y$, we compute
    \begin{align*}
        P\lc \g L(t)>\bmy \rc&=\exp\lc -K_\epsilon t\rc \sum_{i=0}^\infty \frac{(K_\epsilon t)^i}{i!} P\lc \g X(t)+\sum_{j=1}^i \bm \xi_i \geq \bmy \rc \\
        &= \exp\lc -K_\epsilon t\rc P\lc \g X(t)>\bmy \rc + \exp\lc -K_\epsilon t\rc K_\epsilon t P\lc \g X(t)+ \bm \xi_1 \geq \bmy \rc \\
        &+ \exp\lc -K_\epsilon t\rc \sum_{i=2}^\infty \frac{(K_\epsilon t)^i}{i!} P\lc \g X(t)+\sum_{j=1}^i \bm \xi_i \geq \bmy \rc
    \end{align*}
    As just derived above, $ \exp\lc -K_\epsilon t\rc P\lc \g X(t)>\bmy \rc$ is smaller than $C_1t^2$ for all $t\in(0,\Bar{q}_1)$ and  $\bmy\in M_{1,\delta}$. Moreover, for the third term observe that 
    $$\sum_{i=2}^\infty \frac{(K_\epsilon t)^i}{i!}=(K_\epsilon t)^2\sum_{i=0}^\infty \frac{(K_\epsilon t)^i}{(i+2)!}\leq (K_\epsilon t)^2\sum_{i=0}^\infty \frac{(K_\epsilon t)^i}{i!}=(K_\epsilon t)^2 \exp(K_\epsilon t) $$
    which gives that 
    \begin{align}
         \exp\lc -K_\epsilon t\rc \sum_{i=2}^\infty \frac{(K_\epsilon t)^i}{i!} P\lc X(t)+\sum_{j=1}^i \xi_i >\bmy \rc \leq K_\epsilon^2 t^2 \label{eqndecomp2}
    \end{align} 
    for all $t>0$ and $\bmy\in[-\infty,\infty)^d$. Therefore, it remains to bound 
    $$
    D(t) :=\exp\lc -K_\epsilon t\rc K_\epsilon t P\lc  X(t)+\xi_1>\bmy \rc.
    $$

    As a next step, observe that for $K_\epsilon=0$ there is nothing to show. Therefore, assume that $K_\epsilon>0$ and consider the function 
    \begin{align*}
        g_{\bmy}:\R^d \to [0,1];\ g_{\bmy}(\mathbf{u})&=P\lc \mathbf{u}+\xi_1\geq \bmy \rc= K_\epsilon^{-1} \int_{\{ \mathbf{u}-\bmy \geq  \bmx\}} (1-c_\epsilon(-\bmx))\lambda(-\bmx)  \rmd \bmx
    \end{align*} 
    and obtain
    $$ \e\lk g_{\bmy} (X(t))\rk=P\lc X(t)+\xi_1>\bmy \rc, $$
    since $\xi_1$ and $X$ are independent. By our assumptions, $(1-c_\epsilon(-\bmx))\lambda(-\bmx)$ satisfies the conditions of Lemma \ref{lemdiffintegrateddens}. Therefore,  $\mathbf{u}\mapsto g_{\bmy}(\mathbf{u})$ is twice continuously differentiable for every $\bmy\in
    [-\infty,\infty)^d$ with uniformly bounded first and second order derivatives. Note that $g_{\bmy}$ only depends on $\mathbf{u}_{I_{<\infty}}$ and therefore all partial derivatives w.r.t.\ a component $j\not\in I_{<\infty}$ vanish. Let us denote this bound by $C_2$ and note that the bound can be chosen to be independent of $\bmy$ and solely dependent on $\delta$ through $\epsilon$.

    Note that we can decompose $X$ into a continuous part $X^{(c)}(t):=t(\bmb-\bmb_\epsilon)+ W(t)$ and a discontinuous part 
    $$
    X^{(d)}(t)=\int_0^t\int_{B_2(1)}\bmx  c_\epsilon(\bmx)\lc N_X- \mathcal{L}\otimes \Lambda\rc(\rmd (s,\bmx)),
    $$
    where $W$ is a Brownian motion with covariance matrix $\Sigma=\lc \Sigma_{i,j}\rc_{1\leq i,j\leq d}$, $N_X$ is a PRM with intensity $\mathcal{L}\otimes \Lambda$, $\mathcal{L}$ denotes the Lebesgue measure and we have used that $\epsilon$ is chosen such that $B_\infty(\eps)\subset B_2(1)$. Since $g_{\bmy}$ is twice continuously differentiable we can apply the multivariate Ito formula for Lévy processes given in \cite[Theorem 4.4.7]{applebaum2009} to obtain
    \begin{align*}
        &g_{\bmy}(X(t))
        \\
        &=g_{\bmy}(\bm 0)+ \sum_{\leqd} \int_0^t \partial_i g_{\bmy}(X_i(s-)) \rmd X^{(c)}_i(s)+ \frac{1}{2}\sum_{1\leq i,j\leq d}\Sigma_{i,j}\int_0^t  \partial_{i,j}g_{\bmy}(X(s-)) \rmd s \\
        &~+\int_0^t\int_{B_2(1)} \lc g_{\bmy}(X(s-)+\bmx)-g_{\bmy}(X(s-))\rc  c_\epsilon(\bmx)\lc N_X- \mathcal{L}\otimes \Lambda\rc(\rmd (s,\bmx)) \\
        &~+\int_0^t\int_{B_2(1)} \lc g_{\bmy}(X(s-)+\bmx)-g_{\bmy}(X(s-)) -\bmx^\intercal\nabla g_{\bmy}(X(s-)) \rc  c_\epsilon(\bmx) \mathcal{L}\otimes \Lambda(\rmd (s,\bmx))\\
        &=g_{\bmy}(\bm 0)+ \sum_{\leqd} \int_0^t \partial_i g_{\bmy}(X_i(s-)) \rmd X_i(s)+ \frac{1}{2}\sum_{1\leq i,j\leq d} \Sigma_{i,j}\int_0^t  \partial_{i,j}g_{\bmy}(X(s-)) \rmd s \\
        &~+\int_0^t\int_{B_2(1)} \lc g_{\bmy}(X(s-)+\bmx)-g_{\bmy}(X(s-))   -\bmx^\intercal\nabla g_{\bmy}(X(s-)) \rc  c_\epsilon(\bmx)\lc N_X- \mathcal{L}\otimes \Lambda\rc(\rmd (s,\bmx)) \\
        &~+\int_0^t\int_{B_2(1)} \lc g_{\bmy}(X(s-)+\bmx)-g_{\bmy}(X(s-)) -\bmx^\intercal\nabla g_{\bmy}(X(s-)) \rc  c_\epsilon(\bmx) \mathcal{L}\otimes \Lambda(\rmd (s,\bmx))\\
&=: T_1+T_2+T_3+T_4+T_5,
\end{align*}
where we recall that $\partial_i g$ stands for the partial derivative of $g$ with respect to the $i$-th entry of its argument and similar for $\partial_{i,j}g$, $\nabla g_{\bmy}:=\lc \partial_i g_{\bmy}\rc_{\leqd}$ and  $\nabla^2 g_{\bmy}:=\lc \partial_{i,j} g_{\bmy}\rc_{1\leq i,j\leq d}$, defining  $\partial_i g_{\bmy}=\partial_{i,j} g_{\bmy}=0$ if at least one $i,j\not\in I$. 

First, observe that 
    \begin{align}
       T_1=g_{\bmy}(\bm 0)=\int_{\{ \bmx \geq\bmy \}} (1-c_\epsilon(\bmx))\lambda(\bmx) \rmd \bmx  /K_\epsilon=\int_{\{ \bmx \geq \bmy \}} \lambda(\bmx) \rmd \bmx  /K_\epsilon=\Lambda\lc [\bmy,\bm \infty) \rc /K_\epsilon, \label{eqnest1} 
    \end{align}    
    since $y_1 \geq  \delta>\epsilon$.

Next, we bound $T_2$. Since $\partial_i g_{\bmy}$ is bounded by $C_2$ which only depends on $\delta$, and since $\partial_i g_{\bmy}$ is continuous, we have that $\partial_i g_{\bmy}(X_i(s-))$ is predictable with respect to the filtration generated by $(X(s))_{s \geq 0}$. Thus, in the notation of \cite[Chapter 4]{applebaum2009}, $\partial_i g_{\bmy}(X_i(s-))$ is in $\mathcal{H}(2,E)$ with $E := B_2(1)\setminus\{\bm 0\}$. Moreover, $\widetilde {\g X}(s) := \g X(s) - (\bmb - \bmb_\eps)s$ is a martingale-valued measure of type $(2,\rho)$, again using terminology from \cite[Chapter 4, Example 4.1.1]{applebaum2009}. By Theorem 4.2.3 in the latter reference, 
\[
H_i(t) := \int_0^t \partial_i g_{\bmy}(X_i(s-)) \rmd \widetilde{X}_i(s)
\]
is a square integrable martingale and in particular $\e[H_i(t)] = 0$.

Finally, note that
\begin{align}
\bigg\vert  \sum_{\leqd} \int_0^t \partial_i  g_{\bmy}(X_i(s-)) (b_i-b_{\epsilon,i})\rmd s \bigg\vert &\leq t C_2\Vert b-b_\epsilon\Vert_1  \label{eqnest2}
\end{align}
where $C_3:=\Vert b-b_\epsilon\Vert_1$ does not depend on $t$. Therefore, $\vert \e\lk T_2(t)\rk \vert  \leq tC_2C_3 $.

Third, since $\max_{i,j \in  V}|\partial_{i,j} g_{\bmy}|$ is bounded by $C_2$ we obtain that 
\begin{align}
\vert T_3 \vert \leq  t C_2  d^2 \max_{1\leq i,j\leq d} \Sigma_{i,j}  \label{eqnest3}
\end{align} 
which implies $\vert \e\lk T_3\rk \vert \leq C_2 C_4 t $, where $C_4:= d^2  \max_{1\leq i,j\leq d} \Sigma_{i,j} $.
   
It remains to treat $T_4$ and $T_5$. First, we apply a first order Taylor expansion on $g_{\bmy}$ at $X(s-)$ with corresponding random Lagrange remainder $Z(X(s-))$ to obtain that
\[
g_{\bmy}\lc X(s-) +\bmx \rc -g_{\bmy} \lc  X(s-) \rc -\bmx^\intercal \nabla g_{\bmy}\lc  X(s-) \rc=\bmx^\intercal \nabla^2 g_{\bmy}\lc Z(X(s-))  \rc \bmx.
\]
Observe that for an arbitrary square matrix $M \in \R^{d\times d}$ we have
\[
|\bmx^\top M \bmx| \leq \|\bmx\|_1 \|M \bmx\|_1 \leq \|\bmx\|_1^2 \max_{i,j} |M_{ij}| \leq d \|\bmx\|_2^2 \max_{i,j} |M_{ij}|.
\]
Thus we have almost surely
\[
\lv \bmx^\intercal \nabla^2 g_{\bmy}\lc Z(X(s-))  \rc  \bmx \rv \leq C_2 d \Vert \bmx\Vert_2^2.
\]
Therefore,
    \begin{align}
        \lv\e\lk T_5\rk\rv&\leq \int_0^t \int_{B_2(1)} dC_2\Vert \bmx\Vert^2_2  c_\epsilon (\bmx) \mathcal{L}\otimes \Lambda (\rmd s,\rmd \bmx)=dC_2 t\int_{B_2(1)} \Vert \bmx\Vert^2_2  c_\epsilon (\bmx)  \Lambda(\rmd \bmx) =C_2 C_5t \label{eqnest4}
    \end{align}
    where $C_5:= d \int_{B_2(1)} \Vert \bmx\Vert^2_2  c_\epsilon (\bmx)  \Lambda(\rmd \bmx)<\infty$ since $\Lambda$ is a Lévy measure.

To treat $T_4$, observe that
\begin{align*}
T_4(t) &= \int_0^t \int_{B_2(1)} \Big( g_{\bmy}\lc X(s-) +\bmx \rc -g_{\bmy} \lc  X(s-) \rc -\bmx^\intercal \nabla g_{\bmy}\lc  X(s-) \rc \Big)
\\
& \hspace{2.5cm} \times c_\epsilon (\bmx) \lc N_X -\mathcal{L}\otimes \Lambda\rc (\rmd s,\rmd \bmx)
\end{align*} 
is, seen as process of $t$, a square integrable martingale by \cite[Proposition 8.8]{conttankov}, since from the arguments above it follows that, using the Fubini-Tonelli Theorem,
    \begin{align*}
    &\e\Big[\int_0^t \int_{B_2(1)} \Big( g_{\bmy}\lc X(s-) +\bmx \rc -g_{\bmy} \lc  X(s-) \rc -\bmx^\intercal \nabla g_{\bmy}\lc  X(s-) \rc \Big)^2   c_\epsilon (\bmx) \mathcal{L}\otimes \Lambda (\rmd s,\rmd \bmx)\Big]
    \\
    & = \int_0^t \int_{B_2(1)} \e\Big[\Big( g_{\bmy}\lc X(s-) +\bmx \rc -g_{\bmy} \lc  X(s-) \rc -\bmx^\intercal \nabla g_{\bmy}\lc  X(s-) \rc \Big)^2 \Big]  c_\epsilon (\bmx) \mathcal{L}\otimes \Lambda (\rmd s,\rmd \bmx)
    \\
    &<\infty,
    \end{align*}
    which implies $\e\lk T_4\rk=0$.

    Combining (\ref{eqnest1}), (\ref{eqnest2}), (\ref{eqnest3}) and (\ref{eqnest4}) gives 
    \begin{align}
        \big\vert \e\lk g_{\bmy}(X(t))\rk -\Lambda\lc [\bmy,\bm\infty)\rc/K_\epsilon \big\vert \leq t\max \Big\{ C_2 C_3;C_2 C_4;C_2 C_5 \Big\}=:t\Bar{C} 
    \end{align} 
    where $\Bar{C}$ does not depend on $t$ but can depend on $\delta, \eps$ and the other constants in the proof above.     
    
    Using that $0 \leq 1-\exp(-K_\epsilon t)\leq K_\epsilon t$ for all $t\geq 0$ we obtain 
    \begin{align} \nonumber
        &\vert D(t) -t\Lambda\lc [\bmy,\bm\infty)\rc\vert
        \\
        &=\big\vert \exp(-K_\epsilon t)K_\epsilon t P\lc X(t)+\xi_1 >\bmy \rc - t\Lambda\lc[\bmy,\bm\infty)\rc \big\vert  \nonumber\\
        &\leq   \big\vert  \exp(-K_\epsilon t) K_\epsilon t\big( \e\lk g_{\bmy}(X(t))\rk -\Lambda\lc[\bmy,\bm\infty)\rc/K_\epsilon\big) \big\vert  +\big\vert \big( \exp(-K_\epsilon t)-1\big) t\Lambda\lc[\bmy,\bm\infty) \rc \big\vert  \nonumber\\
        &\leq K_\epsilon \Big( \bar{C}+ \Lambda\big( B_2^\complement(\delta)\big) \Big)  t^2 \label{eqndecomp3}
    \end{align} 
    for all $\bmy\in M_{1,\delta}$ and $t\in(0,\bar{q_1})$.

    Thus, choosing $K:=C_1+K_\epsilon\lc K_\epsilon+  \bar{C}+ \Lambda\big(B_\infty^\complement(\delta)\big)\rc $,  (\ref{eqndecomp1}), (\ref{eqndecomp2}), and (\ref{eqndecomp3}) give
    \begin{align}
        \big\vert P\lc \g L(t)>\bmy\rc- t\Lambda\lc [\bmy,\bm\infty)\rc\big\vert \leq Kt^2 \label{eqnunif}
    \end{align} 
    for all $ \bmy\in M_{1,\delta}$ and $t\in (0,\bar{q})$.

It remains to show that there is a constant $\bar{K}$ such that (\ref{eqnunif}) holds for all $t>0$. Using that $P(\g L(t)>y)\leq 1$ and $\Lambda\lc [\bmy,\bm\infty)\rc\leq \Lambda\big(B_\infty^\complement(\delta)\big)$ for all $\bmy\in M_{1,\delta}$, we obtain
\[
\big\vert P\lc \g L(t)>\bmy\rc- t\Lambda\lc [\bmy,\bm\infty)\rc\big\vert \leq 1 + t \Lambda\big(B_\infty^\complement(\delta)\big).
\]
Setting $K_1 = 1/\bar{q}^2+\Lambda\big(B_\infty^\complement(\delta)\big)/\bar{q}$ yields for all $t \geq \bar q, \bmy \in M_{1,\delta}$
\[
K_1t^2 = \Big( 1/\bar{q}^2+\Lambda\big(B_\infty^\complement(\delta)\big)/\bar{q} \Big)t^2 \geq 1 + t \Lambda\big(B_\infty^\complement(\delta)\big) \geq \big\vert P\lc \g L(t)>\bmy\rc- t\Lambda\lc [\bmy,\bm\infty)\rc\big\vert\quad.
\]
Setting $\bar{K}=\max\{K_1,K\}$ the claim follows.
\end{proof}

We immediately obtain the corresponding result for HR processes.

\begin{proof}[Proof of Proposition \ref{propunifconvHRprocess}]
The proof immediately follows from observing that Proposition \ref{cordensregHR} implies that for every $\epsilon> 0$ there is a $c_\epsilon \in S_\epsilon$ such that $(1-c_\epsilon) \lambda_{\rm HR}$ satisfies the conditions of Lemma \ref{lemdiffintegrateddens}.
\end{proof}

\newpage

\section{Other Proofs}
\label{app_proofs}

\subsection{Proof of Lemma \ref{lemHrprocessclass}}
\begin{proof} 
Obviously setting $\bmb \neq 0$ only changes the drift, but not the Lévy measure of $\g L$. Thus, we can w.l.o.g.\ assume that $\bm\tau=\bmb=0$. 

Recall the definition of the function $\bm t_{\bac}$ from~\eqref{defHRlevymeasuretrafo}. Note that densities of \HR{} are closed under marginalization, so
\[
\int_{\R^d\setminus \{\bm 0\}} q(\bmx_{\setminus I}) g(\bmx;\Theta) \rmd \bmx = \int_{\R^{d-|I|} \setminus \{\bm 0\}} q(\bmx) g(\bmx;\tilde \Theta(I)) \rmd \bmx
\]
where $\tilde \Theta(I)$ is the parameter matrix which is given by $-(P\Gamma_{\setminus I, \setminus I}P)^{+}/2$ where $\Gamma$ denotes the $\Gamma$ matrix corresponding to $\Theta$; see \citet[][Example 7]{engelkehitz2020}.

Now, let $\bma,\bmz\in\Rdo $ and denote $\bmz(\bma)=\lc a_iz_i\rc_{1\leq i\leq d}$ as well as the enumeration of the orthant in $\R^{d^\prime}$ as $\Oo_{d^\prime}=\{-1,1\}^{d^\prime}$. First, we consider $\bma \geq 0$ and denote $I=\{i\mid a_i=0\}$. Then the characteristic function of $ \bma \odot \g L=\lc a_iL_i(t)\rc_{\leqd}$ is given by
    \begin{align*}
        &\e\lk \exp\lc i\bmz^\intercal (\bma \odot \g L(t))\rc \rk 
        \\
        &=\exp\Bigg( t \sum_{\tilde{\bmo} \in\Oo_{d-\vert I\vert}} \sum_{\overset{\bmo\in\Oo}{\bmo_{ V\setminus I}=\tilde{\bmo}}}\gamma_\bmo \int_{\Rdo\cap\bmo} \Big[ \exp\Big( i \sum_{j\in  V\setminus I} z_j \bmo(x_j)\lc  a_j^{\alpha_j} c_j^{\bmo(x_j)}\vert x_j\vert \rc^{1/\alpha_j}  \Big) -1
        \\
        &\hspace{3cm}-i\Big\{ \sum_{j\in  V\setminus I} z_j \bmo(x_j)\lc  a_j^{\alpha_j} c_j^{\bmo(x_j)}\vert x_j\vert \rc^{1/\alpha_j}  \Big\}  \id_{\{\Vert t_{\balpha,\bmc}(\bmx)\Vert_2 \leq 1\}} \Big] g(\vert\bmx\vert;\Theta) \rmd \bmx \Bigg)
        \\
        &=\exp\Bigg( t \sum_{\tilde{\bmo} \in\Oo_{d-\vert I\vert}} \gamma_{\tilde{\bmo}}(\bma)\int_{\R^{d-\vert I\vert}\setminus\{\bm 0\}\cap\tilde{\bmo}} \Big[ \exp\lc i \bmz_{ V\setminus I}^\intercal t_{\balpha_{ V\setminus I},\bmc(\bma)_{ V\setminus I}} (\bmx)  \rc -1\\
        &\hspace{1.5cm} -i\bmz_{ V\setminus I}^\intercal t_{\balpha_{ V\setminus I},\bmc(\bma)_{ V\setminus I}} (\bmx)   \id_{\{\Vert t_{\balpha_{ V\setminus I},\bmc(\bma)_{ V\setminus I}}(\bmx)\Vert_2 \leq 1\}} \Big] g(\vert\bmx\vert;\tilde \Theta(I)) \rmd \bmx +i\bmz_{ V\setminus I}^\intercal\bm\zeta(\bma)\Bigg)
 \end{align*}
where 
$\bmc(\bma)_{ V\setminus I}=\lc \lc   a_i^{\alpha_i} c^{-}_i,  a_i^{\alpha_i} c^{+}_i\rc\rc_{ V\setminus I}$, $ \gamma_{\tilde{\bmo}}(\bma):=  \sum_{\overset{\bmo\in\Oo}{\bmo_{ V\setminus I}=\tilde{\bmo}}}\gamma_\bmo $ and for all $i\in V\setminus I$ 
\begin{multline*} [\bm \zeta(\bma)]_i:= \sum_{\tilde{\bmo} \in\Oo_{d-\vert I\vert}}\sum_{\overset{\bmo\in\Oo}{\bmo_{ V\setminus I}=\tilde{\bmo}}}\gamma_\bmo  \int_{\Rdo\cap\bmo}  [t_{\balpha_{ V\setminus I},\bmc(\bma)_{ V\setminus I}} (\bmx)]_i
\\
\lc \id_{\{\Vert t_{\balpha_{ V\setminus I},\bmc(\bma)_{ V\setminus I}}(\bmx)\Vert_2 \leq 1\}}        -  \id_{\{\Vert t_{\balpha,\bmc}(\bmx)\Vert_2 \leq 1\}} \rc g(\vert\bmx\vert;\Theta) \rmd \bmx .
\end{multline*}
Thus, noting that $\gamma_{\tilde{\bmo}}(\bma)$ satisfy~\eqref{weight_norm}, $\bma \odot \g L$ is a $(d-\vert I\vert)$ dimensional HR process.

To obtain the general claim it remains to show that $\bma \odot \g L$ is an HR process for all $\bma\in\{-1,1\}^d$. Repeating the same calculations as above we obtain
    \begin{align*}
        &\e\lk \exp\lc i\bmz^\intercal (\bma \odot \g L(t))\rc \rk \\
        &=\exp\Bigg( t  \sum_{\bmo\in\Oo}\gamma_\bmo \int_{\Rdo\cap\bmo} \Big[\exp\Big\{ i \sum_{j \in  V} z_j \bmo(a_j)\bmo(x_j)\big(  c_j^{\bmo(x_j)}\vert x_j\vert \big)^{1/\alpha_j}  \Big\} -1\\
        &\hspace{3cm} -i\Big\{ \sum_{j \in  V} z_j \bmo(a_j)\bmo(x_j)\lc  c_j^{\bmo(x_j)}\vert x_j\vert \rc^{1/\alpha_j} \Big\}   \id_{\{\Vert t_{\balpha,\bmc}(\bmx)\Vert_2 \leq 1\}} \Big] g(\vert\bmx\vert;\Theta) \rmd \bmx \Bigg)\\
        &=\exp\Bigg( t \sum_{\bmo \in\Oo} \hat{\gamma}_{\bmo(\bma)}\int_{\Rdo \cap \bmo} \Big[ \exp\lc i \bmz^\intercal t_{\balpha,\hat{\bmc}(\bma)} (\bmx)  \rc -1
        \\
        &\hspace{3cm} -i\bmz^\intercal t_{\balpha,\bmc(\bma)} (\bmx)   \id_{\{\Vert t_{\balpha,\hat{\bmc}(\bma)}(\bmx)\Vert_2 \leq 1\}} \Big] g(\vert\bmx\vert;\Theta) \rmd \bmx \Bigg)
 \end{align*}
where $\hat{\gamma}_{\bmo}(\bma)=\gamma_{\bmo(\bma)}$ with $\bmo(\bma)=\lc a_io_i\rc_{\leqd}$ and $\hat{\bmc}(\bma)=\lc  c_i^{-\sgn(a_i)} , c_i^{\sgn(a_i)}\rc$.  
Therefore, $\bma \odot \g L$ is an \HRL{} process and the claim follows.

\end{proof}

\subsection{Proof of Theorem \ref{thmglobalmarkovprop}}

\begin{proof} 

Define $\mathcal{R} := \{R = \times_{v \in  V} R_v: \bm 0 \notin \overline{R}, \Lambda(R) > 0\}$ and for any fixed $R \in \mathcal{R}$ define $\g Z^R$ as a random vector with distribution $\Lambda(\cdot\cap R)/\Lambda(R)$. The measure $\Lambda = \Lambda^{(\bag)}$ satisfies Assumption [A2] from \cite{engelke2024levygraphicalmodels}, so for any fixed $i\neq j$ we have 
\begin{equation}\label{eq:indepequiv1}
L_i\perp L_j \mid L_{ V\setminus\{i,j\}} \iff Z_i^R \perp Z_j^R \mid Z_{ V\setminus\{i,j\}}^R~\forall R \in \mathcal{R}
\end{equation}
by Theorem 4.3 in \cite{engelke2024levygraphicalmodels}. On each $R \in \mathcal{R}$, $\Lambda$ has a strictly positive Lebesgue density since $\gamma_\bmo > 0$ for all $\bmo$ which is satisfied for IHR processes. Assuming that $\Theta_{i,j} = 0$ and $\psi_{i,j} = 0$, the density of $\Lambda$ factorizes (this is evident from the definition), so the conditional independence for $\g Z^R$ in~\eqref{eq:indepequiv1} holds by the classical Hammersley--Clifford Theorem for positive densities with respect to a product measure.

Now assume that the conditional independence in~\eqref{eq:indepequiv1} holds, which is equivalent to the $\Lambda$ conditional independence defined in~\eqref{defcondind3}. By \citet[][Prop.~4.4]{engelke2024levygraphicalmodels} we know that this is further equivalent to the corresponding conditional independence 
for the PLM $\Lambda^\star$. Recall from~\eqref{defHRparetolevycop} that the asymmetric HR PLM admits a density with respect to Lebesgue measure, and this density takes the form
\[ 
\lambda^\star(\bmx)= h(\bmx) g(\bmx) \bm 1_{\{\bmx \neq \bm 0\}} ,  \quad \bmx \in\mathbb R^d,
\]
with $h(\bmx) = \lambda^\star_{(\text{sym})}(\bmx)$
and $g(\bmx) = \sum_{\bmo\in\mathcal O} \gamma_\bmo \id_{\{\bmx\in \bmo\}}$.

Choose an arbitrary {set $R \in \mathcal{R}$} contained in {$(0,\infty)^d$}. {Conditional independence of $\g Z^R$ combined with the classical Hammersley--Clifford Theorem for positive densities with respect to a product measure implies that the density of $\g Z^R$ factorizes. Since this is true for all $R$ bounded away from zero, and since $h$ restricted to $(0,\infty)^d$ is proportional to the density of the HR exponent measure, it follows that $\Theta_{ij} =0$. Moreover, we have $h(\bmx) = h_{V\setminus\{j\}}(|\bmx_{V\setminus\{j\}}|)h_{V\setminus\{i\}}(|\bmx_{V\setminus\{i\}}|)$ where $|\bmx|$ is interpreted as absolute value taken component-wise.}

Now from \citet[][Theorem 4.6]{engelkeivanosvsstrokorb2022}, it follows that $\lambda^\star$ factorizes if we have the 
conditional independence in~\eqref{eq:indepequiv1};
note that the second condition of their theorem is satisfied since the sub-face corresponding to $A\cup B$ does not have any mass in the IHR model. Therefore, we can write 
\[ 
g(\bmx) = \frac{\lambda^\star_{V\setminus\{j\}}(\bmx_{V\setminus\{j\}})\lambda^\star_{V\setminus\{i\}}(\bmx_{V\setminus\{i\}})}{h_{V\setminus\{j\}}(|\bmx_{V\setminus\{j\}}|)h_{V\setminus\{i\}}(|\bmx_{V\setminus\{i\}}|)}, \quad \bmx\in (\mathbb R \setminus \{\bm 0\})^d; 
\]
note that the factors of $h$ are the same in all orthants. {In particular, for any $\bmy \in \{-1,1\}^d$ we have $h_{V\setminus\{j\}}(|\bmy_{V\setminus\{j\}}|)h_{V\setminus\{i\}}(|\bmy_{V\setminus\{i\}}|) = \lambda^\star_{(\text{sym})}(\vert \bmy\vert )=:C > 0$ is constant. Moreover, each such $\bmy$ can be uniquely identified with an orthant $\bmo$, so we have shown
\begin{equation}\label{eq:gammaofac1}
\gamma_\bmo =  g(\bmo) = f_1(\bmo_{V\setminus\{i\}})f_2(\bmo_{V\setminus\{j\}}),
\end{equation}
where $f_1, f_2$ are some strictly positive functions. By definition of $\gamma_\bmo$ it is easy to see that 
\[
\gamma_\bmo = {e^{2\psi_{i,j}o_io_j}} \tilde f_1(\bmo_{V\setminus\{i\}}) \tilde f_2(\bmo_{V\setminus\{j\}}) 
\]
where again $\tilde f_i$ are strictly positive {and depend on $\Psi$}. Combining this with~\eqref{eq:gammaofac1} we have 
\begin{equation}\label{eq:factorexp}
e^{2\psi_{i,j}o_io_j} = \frac{f_1(\bmo_{V\setminus\{i\}})f_2(\bmo_{V\setminus\{j\}})}{\tilde f_1(\bmo_{V\setminus\{i\}}) \tilde f_2(\bmo_{V\setminus\{j\}})} = c_1(o_i)c_2(o_j)
\end{equation}
for some positive functions $c_1,c_2$ on $\{-1,1\}$ where the second equation follows since the left-hand side depends only on $o_i,o_j$. An elementary computation shows that this is only possible if $\psi_{i,j} = 0$.
\footnote{Write $\psi$ for $2\psi_{i,j}$. If~\eqref{eq:factorexp} was true, we would have $c_1(+)c_2(+) = c_1(-)c_2(-) = e^\psi$ and $c_1(-)c_2(+) = c_1(+)c_2(-) = e^{-\psi}$. Then $1 = e^\psi e^{-\psi} = c_1(-)c_2(-)c_1(+)c_2(-) = c_1(+)c_2(+)c_1(-)c_2(+) $ so $c_2(-)^2 = c_2(+)^2$ and thus $c_2(-) = c_2(+)$ since both are positive. This also implies $c_1(-) = c_1(+)$. But then $e^\psi = e^{-\psi}$, so $\psi = 0$.   
}
This completes the proof.}

\end{proof}

\subsection{Proof of Theorems~\ref{thm:secondorder} and ~\ref{thmconsitencygammaestimator}}

Before proving the theorem we need two preparatory Lemmas.

\begin{lem}
    \label{lemtechproofs}
    \begin{enumerate}
        \item[i)]   Let $\g L$ denote an HR process without drift
        and let $b_i$ be defined as in Theorem \ref{thmselfsimHRprocess}, equation~\eqref{eq:deffunctionb}. Then, for all $t< 1$ and $i \in  V$
     \begin{align*}
        \vert b_i(t) \vert &\leq        C+ 2^{d}(c_i^+ \vee c_i^-)^{1/\alpha_i}  \Bigg( \frac{1}{  1/\alpha_i-1 } \Bigg( \frac{1}{\lc t(c_i^+ \wedge c_i^-)\rc^{1/\alpha_i-1}} -1\Bigg)\id_{\{ \alpha_i\not=1\}} 
        \\
        &\quad -\log\lc t(c_i^+ \wedge c_i^-)  \rc \id_{\{ \alpha_i=1\}}  \Bigg) 
     \end{align*}
     for a constant $C$ that may depend on $\balpha,\bmc,\bgamma$ and $\Theta$.
     \item[ii)] Let $\g L$ denote an HR processes with drift $\bm \tau\in\R^d$.
     Then $\g L(q)\overset{d}{=} \lc q^{1/\alpha_i}L_i(1)+c_i(q)\rc_{\leqd}$ where $c_i(q)=\tau_i(q-q^{1/\alpha_i})+q^{1/\alpha_i}b_i(q)$ for $b_i$ as in part (i). Further, for all $0<\zeta<\min\{1/\alpha_i,1\}$ there exists a constant $C=C(\bm \tau)$ such that $c_i(q)\leq Cq^{\min\{1/\alpha_i;1\}-\zeta} $ for all $q\in(0,1)$. 
     \item[iii)]For all $a>0,\bmx>0$ (component-wise) we have
        \begin{align*}
            \Lambda^{(\bag,\Theta)}\lc  \{ \bmy\in\Rdo \mid y_i>a^{-1/\alpha_i}x_i~~\forall i \in  V\} \rc= a\Lambda^{(\bag,\Theta)}\lc (\bmx,\bm \infty) \rc.
        \end{align*}
    \end{enumerate}
 \end{lem}

\begin{proof}[Proof of Lemma~\ref{lemtechproofs}]
\begin{enumerate}
         \item[i)] Observe that $\id_{\{\Vert t_{\balpha,\bmc}(t\bmx) \Vert_2 \leq 1\}}=\id_{\{\sumd \lc c_i^{\bmo(x_i)}t \vert x_i\vert \rc^{2/\alpha_i}  \leq 1\}}$ and for $t < 1$ ,  noting that for $0 < a < 1$, $\|t_{\balpha,\bmc}(a\bmx)\|_2 < \|t_{\balpha,\bmc}(\bmx)\|_2$, we have for $i \in  V$
        \begin{align*}
            &\lv \id_{\{\Vert t_{\balpha,\bmc}(\bmx) \Vert_2
            \leq 1\}}-  \id_{\{\Vert t_{\balpha,\bmc}(t\bmx) \Vert_2 \leq 1\}}\rv 
            \\
            &=\id_{\{ 1<  \Vert t_{\balpha,\bmc}(\bmx) \Vert_2,  \Vert t_{\balpha,\bmc}(t\bmx) \Vert_2\leq 1\}}\\
            &=\id_{\{1<  \Vert t_{\balpha,\bmc}(\bmx) \Vert_2,  \Vert t_{\balpha,\bmc}(t\bmx) \Vert_2\leq 1, \vert x_i\vert \leq 1\}}+\id_{\{ 1<  \Vert t_{\balpha,\bmc}(\bmx) \Vert_2, \Vert t_{\balpha,\bmc}(t\bmx) \Vert_2\leq 1, \vert x_i\vert > 1\}} \\
            &\leq \id_{\{ \exists j\in\{1,\ldots,d\} \text{ s.t. } 
            , \vert x_i\vert \leq 1 \}}
            +\id_{\{ \lc c_i^{\bmo(x_i)}t\vert x_i\vert \rc^{2/\alpha_i}\leq 1, \vert x_i\vert > 1\}} \\
            &\leq \id_{\{  \vert x_i\vert\leq 1/(t\min\{c_i^+,c_i^-\})  , \vert x_i\vert > 1\}}  + \sum_{j=1}^d \id_{\{ \vert x_j\vert > d^{-\alpha_j/2}\max\{ c_j^+,c_j^-\}^{-1} , \vert x_i\vert \leq 1\}}           
        \end{align*} 
        We get
        \begin{align*} 
                &\lv b_i(t)\rv = \Bigg\vert\int_{\Rdo}   \bmo(x_i)\lc c_i^{\bmo(x_i)}\vert x_i\vert \rc^{1/\alpha_i} \ \lc \id_{\{\Vert t_{\balpha,\bmc}(\bmx) \Vert_2 \leq 1\}}-  \id_{\{\Vert t_{\balpha,\bmc}(t\bmx) \Vert_2 \leq 1\}} \rc \Lambda^{(\bm 1,\bm 1,\bgamma)}(\rmd\bmx)  \Bigg\vert\\
                &\leq   \sum_{j=1}^d \max\{c_j^+,c_j^-\}^{1/\alpha_i} \int_{\{ \vert x_j\vert > d^{-\alpha_j/2}\max\{ c_j^+,c_j^-\}^{-1} \}}\Lambda^{(\bm 1,\bm 1,\bgamma)}(\rmd\bmx)\\
                &+ \int_{\{  \vert x_i\vert\leq 1/(t\min\{c_i^+,c_i^-\})  , \vert x_i\vert > 1\}}  \lc c_i^{\bmo(x_i)}\lv x_i\rv \rc^{1/\alpha_i}  \Lambda^{(\bm 1,\bm 1,\bgamma)}(\rmd\bmx)  \\
                &\leq   C+2^{d}(\max_{\bmo\in\Oo}\gamma_\bmo )(\max\{c_i^+,c_i^-\})^{1/\alpha_i} \int_{1}^{1/(t\min\{c_i^+,c_i^-\})}   x_i^{1/\alpha_i-2} \rmd x_i  \\
                &\leq  C+2^{d}(\max\{c_i^+,c_i^-\})^{1/\alpha_i} \lc 1/\alpha_i-1 \rc^{-1}\lc \lc t\min\{c_i^+,c_i^-\}\rc^{-(1/\alpha_i-1)} -1\rc\id_{\{ \alpha_i\not=1\}} \\
                &+2^{d}(\max\{c_i^+,c_i^-\})^{1/\alpha_i} \lc -\log\lc t\min\{c_i^+,c_i^-\}  \rc\rc \id_{\{ \alpha_i=1\}}\\
                &=  C+ 2^{d}(\max\{c_i^+,c_i^-\})^{1/\alpha_i}  \Bigg( \frac{1}{  1/\alpha_i-1 } \lc \frac{1}{\lc t\min\{c_i^+,c_i^-\}\rc^{1/\alpha_i-1}} -1\rc\id_{\{ \alpha_i\not=1\}} \\
                &-\log\lc t\min\{c_i^+,c_i^-\}  \rc \id_{\{ \alpha_i=1\}}  \Bigg) 
        \end{align*}
        where we used that $\gamma_\bmo\in[0,1]$ and set 
        $$ 
        C:=C(\bac):=\sum_{j=1}^d \max\{c_j^+,c_j^-\}^{1/\alpha_i} \int_{\{ \vert x_j\vert > d^{-\alpha_j/2}\max\{ c_j^+,c_j^-\}^{-1} \}}\Lambda^{(\bm 1,\bm 1,\bgamma)}(\rmd\bmx). 
        $$
            \item[ii)] The first part follows from an application of Corollary \ref{cortimescalingHRprocess} since $\g L(q)\overset{d}{=} \tilde{\g L}(q)+q\bm\tau, q > 0$, where $\tilde{\g L}$ is an HR process without drift. The second claim follows directly from $i)$ 
            
            \item[iii)]
            We have
            \begin{align*}
                \Lambda^{(\bag)}\lc  \{\bmy: y_i>a^{-1/\alpha_i}x_i \ \forall i \} \rc&= \gamma_{\bm 1}\int \id_{\{\lc c_i^+y_i\rc^{1/\alpha_i}>a^{-1/\alpha_i}x_i \ \forall i\}} g(\bmy;\Theta)\rmd\bmy\\
                &=\gamma_{\bm 1}\int \id_{\{\lc c_i^+ ay_i\rc^{1/\alpha_i}>x_i \ \forall i\}} g(\vert\bmy\vert ;\Theta)\rmd\bmy\\
                &=\gamma_{\bm 1}\int \id_{\{\lc c_i^+ y_i\rc^{1/\alpha_i}>x_i \ \forall i\}} g(\vert\bmy\vert/a;\Theta)a^{-d}\rmd\bmy\\
                &=\gamma_{\bm 1}\int\id_{\{\lc c_i^+ y_i\rc^{1/\alpha_i}>x_i \ \forall i\}} g(\vert\bmy\vert;\Theta)a^{d+1}a^{-d}\rmd\bmy\\
                &=\gamma_{\bm 1} a\int\id_{\{\lc c_i^+ y_i\rc^{1/\alpha_i}>x_i \ \forall i\}} g(\vert\bmy\vert;\Theta)\rmd\bmy\\
                &=a\Lambda^{(\bag)}\lc\{  y_i>x_i \ \forall i \} \rc\\
            \end{align*}
            since the density $g$ of the HR exponent measure is $-(d+1)$ homogeneous.
      
\end{enumerate}
\end{proof}

Next, we prove that we can uniformly control the tail behavior $\bmX_{\bmo}(\Delta)$, which is the main technical ingredient for the verification of the assumption of \cite{engelkelalancettevolgushev2022}. To simplify the theoretical exposition in the following note that scaling of HR processes are HR processes and it is therefore sufficient to prove $\bmX_{\bm 1}(\Delta)$ is in the domain of attraction of a multivariate HR Pareto distribution with precision matrix $\Theta$ to obtain the same statement for every $\bmX_\bmo(\Delta)$. Further define the monotone decreasing marginal survival functions
\begin{align*}
    G_i(z)&=\Lambda^{(\bag,\Theta)}\lc  x_i > z, x_j>0\ \forall j\not=i \rc=\int_{\{x_i>z,x_j>0\forall j\}} \Lambda^{(\bag,\Theta)}(\rmd\bmx)=\frac{\gamma_{\bm 1}c_i^+}{z^{\alpha_i}}, \quad z > 0,
\end{align*} 
and its monotone decreasing inverse as $G_i^{-1}(z)=\lc \frac{\gamma_{\bm 1}c_i^+}{z}\rc^{1/\alpha_i}$. 

For $\emptyset\neq I\subset  V$ define 
\begin{equation}
R_I(\bmx_I) :=\Lambda_{\rm HR}\big( \{\bmy\in[0,\infty)^d\mid y_i\geq 1/x_i\ \forall i\in I \big)\quad \bmx_I \in [0,\infty)^{\vert I\vert},    
\end{equation}
where $\Lambda_{\rm HR}$ denotes the exponent measure of the \HR{} distribution, which can also be expressed as $\Lambda_{\rm HR}(\cdot) = \Lambda^{(\id,\id,\id)}(\cdot \cap (0,\infty)^d)$. Together with Lemma~\ref{lemtechproofs} (iii) this implies
\begin{equation}\label{eq:Rhomogeneous}
R_I(a\bmx_I) = a R_I(\bmx_I), \quad a> 0, \bmx_I \in [0,\infty)^{\vert I\vert}.    
\end{equation}
Since the margins of $\Lambda_{\rm HR}$ are given by $\Lambda_{\rm HR}(\{\bmx| x_i > z\}) = 1/z$, we have for any $\delta > 0$, $\bmx \in [0,\infty)^d$, any $\emptyset \neq I \subset  V$ and any $j \in [|I|]$
\begin{align*}
R_I(\bmx) - R_I(\bmx + \bm e_j\delta) &= \Lambda_{\rm HR}\big(\big\{\bmy \mid (\bmy_I)_j \in \big( \{(\bmx_I)_j+\delta\}^{-1},(\bmx_I)_j^{-1}\big)], (\bmy_{I})_{\setminus j} \geq \bm 1/(\bmx_{I})_{\setminus j}\big\}\big)
\\
&\leq \Lambda_{\rm HR}\big(\big\{\bmy \mid (\bmy_I)_j \in \big(\{(\bmx_I)_j+\delta\}^{-1},(\bmx_I)_j^{-1}\big]\big\}\big)
\\
&= \delta.
\end{align*}
Thus $R_I$ is Lipshitz continuous with respect to $\|\cdot\|_1$ Lipshitz constant $1$, i.e.
\begin{equation}\label{eq:RisLipshitz}
|R_I(\bmx) - R_I(\bmy)| \leq \|\bmx - \bmy\|_1.  \end{equation}

For the following proof also note that for $\bmx \in (0,\infty)^d$
\begin{align}\notag
R_I(\bmx_I) &= \Lambda_{\rm HR}\big(\{\bmy\in[0,\infty)^d\mid y_i\geq 1/x_i\ \forall i\in I\}\big)
\\ \nonumber
&= \gamma_1 \Lambda_{\rm HR}( \{\bmy\mid y_i\geq \gamma_1/x_i\ \forall i\in I \}) \\
&= \gamma_1\Lambda_{\rm HR}( \{\bmy\mid (c_i^+ y_i)^{1/\alpha_i}\geq \lc \gamma_1c_i^+/x_i\rc^{1/\alpha_i} \ \forall i\in I \})
\nonumber \\
&=\Lambda^{(\bag,\Theta)}( \{\bmy\mid y_i\geq   G_i^{-1}(x_i) \ \forall i\in I \})  \label{eq:LamHRLamGen} 
\end{align}

\begin{lem}\label{lem:HRsecorder}
For $d\geq 2$, $\emptyset\neq I\subset  V$ and $\Delta>0$ let $\g L$ denote an HR process with $\gamma_{\bm 1}>0$ and consider the random vector $\bmX(\Delta)$ with law 
\[
P(\bmX(\Delta) \in \cdot) := P(\g L(\Delta) \in \cdot\mid L_i(\Delta)>0, i \in  V)
\]
with associated marginal distribution functions $F_i^+(z)=P\lc X_i(\Delta)\leq z\rc$. Then, for any $\zeta < 1 \wedge \min_{i \in  V} \alpha_i^{-1}$ there exists a constant $C >0$ that depends on $\zeta$ and the other model parameters but not on $q$ such that for all $q\in(0,1)$ and $\bmx_I\in[0,1]^{\vert I\vert}$
$$ 
D_I(\bmx_I,q) := \left|  P\lc F_i^+ \lc X_i(\Delta)\rc >1-qx_i\ \forall\ i\in I\rc -q R_I(\bmx_I) \right|  \leq C q^{1+\zeta} .
$$
\end{lem}

\begin{proof}[Proof of Lemma~\ref{lem:HRsecorder}] 
Using that the class of HR processes is closed under marginal scaling and translation (Lemma~\ref{lemHrprocessclass}) in combination with Corollary \ref{cortimescalingHRprocess}, it follows that for every HR process $\g L$ there exists an HR process $\tilde{\g L}$ such that $\g L(\Delta)\overset{d}{=} \Tilde{\g L}(1)$. Thus, we can w.l.o.g.\ assume $\Delta=1$ in the following. Moreover, is suffices to show the claim for arbitrary $\emptyset\neq I$ as there are only finitely many $I\subset  V$. Thus, let $\emptyset\neq I$ be arbitrary in the following.

Define $ K_1:=P(L_i(1)>0\ \forall i \in  V)$ and denote the marginal distribution functions of $\bmX(1)$ as
$F_{i}^+ (z):=P\lc X_i(1)\leq z\rc$. We have $\g L(q)\overset{d}{=} \lc q^{1/\alpha_i}L_i(1)+c_i(q)\rc_{i \in  V}$ for some functions $c_i$ by Lemma~\ref{lemtechproofs} ii). Observe that for $z>0$ 
\begin{align}
    F_{i}^+ (z)&:=\frac{P \lc L_i(1)\leq z, L_j(1)>0\ \forall j \rc}{K_1}=1-\frac{P \lc L_i(1)> z, L_j(1)>0\ \forall j\neq i \rc}{K_1} \nonumber\\
    &=1- \frac{P \lc q^{-1/\alpha_i} \lc L_i(q)-c_i(q)\rc  > z,q^{-1/\alpha_i} \lc L_j(q)-c_j(q)\rc  >0\ \forall j\neq i \rc}{K_1} \nonumber\\
    &=1- \frac{P \lc   L_i(q) > q^{1/\alpha_i}z+c_i(q), L_j(q)>c_j(q)\ \forall j\not=i \rc}{K_1} \nonumber\\
    &=:1-\frac{H_i\lc q,q^{1/\alpha_i}z+c_i(q) \rc}{K_1} \label{eqnmargindistposlevy}
\end{align}
where we used Lemma \ref{lemtechproofs} ii) and defined
\begin{align}\label{def_H}
    H_i(q,z):=P \lc   L_i(q) > z, L_j(q)>c_j(q)\ \forall j\neq i \rc .
\end{align}

We note that due to Proposition \ref{propunifconvHRprocess} we have that for every $\delta>0$ there exists a $C_1(\delta)>0$ such that for all $z\geq \delta$, $\lc a_j\rc_{j\not=i}\in \R^{d-1}$, $\leqd$ and $q\in(0,1)$ 
\begin{align}
\big\vert P\lc  L_i(q)  > z, L_j(q)  >a_j\ \forall j\not=i \rc -q\Lambda\lc\{\bmy:  y_i>z,y_j>a_j\ \forall j\not=i \}\rc \big\vert  \leq C_1q^2. \label{eqn1}
\end{align}
as well as
$$ \lv P\lc L_i(q)>z \rc-q\Lambda(\{ \bmy: y_i>z\} ) \rv\leq C_1q^2. 
$$
Recalling~\eqref{eq:Rhomogeneous} and \eqref{eq:LamHRLamGen}, this yields the following bound which we will use repeatedly throughout the proof: for all $\delta > 0$ there exists a constant $C_1$ depending only on $\delta$ such that for all $\bmx \in (0,\infty)^d$ such that $\max_{i \in I}G_i^{-1}(x_i) \ge \delta$
\begin{align}
\Big| &P(L_i(q) > G_i^{-1}(x_i) \forall\ i \in I,L_i(q)>c_i(q) \ \forall i\not\in I) \nonumber \\
&- q \Lambda\lc \{ \bmy\in \R^d\mid y_i>  G_i^{-1}(x_i) \forall\ i \in I,y_i>c_i(q) \ \forall i\not\in I\}\rc\Big| \leq C_1 q^2.\label{eq:RelateLitoR}
\end{align}

Let us choose
$$
\delta=\min_{i\in  V}\Big\{   G_i^{-1}(K_1+1/2) ; \lc\Lambda\lc \{\bmy: y_i>\cdot,y_j>1\ \forall j\not=i \}\rc \rc^{-1}  (2K_1)\Big\},
$$
where $K_1$ was defined above.
Further, by Lemma \ref{lemtechproofs} ii), we can fix $q_1>0$ such that (\ref{eqn1}) holds and that for all $q<q_1$ we have $C_1q^2<K_1q$ and $\vert c_i(q)\vert\leq \min\{ 1;\delta\}$ for all $\leqd$. Recalling the definition of $H$ in~\eqref{def_H}, our choice of $\delta$ and $q_1$ yields for all $q<q_1$, noting that $x \mapsto \Lambda\lc  \{ \bmy: y_i>x,y_j>1\ \forall j\not=i \}\rc$ is continuous and strictly decreasing
\begin{align*}
    H_i(q,\delta)&= P\lc  L_i(q)  > \delta, L_j(q)  >c_j(q)\ \forall j\not= i \rc\geq  P\lc  L_i(q)  > \delta, L_j(q)  >1\ \forall j\not= i \rc \\
    &\geq  q\Lambda\lc \{\bmy: y_i>\delta,y_j>1\ \forall j\not=i \}\rc -C_1q^2 \geq q2K_1-K_1q=K_1q 
\end{align*}
where the second inequality utilizes~\eqref{eqn1}.
Thus, $H_i(q,z)<K_1q$ implies $z>\delta$ as $H_i(q,z)$ is decreasing in $z$. 
Further, note that there is a $q_2>0$ such {that for all $q<q_2$ $F_i^+ \lc L_i(1)\rc > 1-q$ implies that $L_i(1)>0$ since $\gamma_{\bm 1}>0$ . Therefore, choosing such a $q_2\leq q_1$, we have for all $0<q<q_2$ 
\begin{align*}
    &P\lc F_i^+ \lc L_i(1)\rc>1-qx_i\ \forall i\in I, L_i(1)>0\ \forall i\rc\\
    &=P\lc F_i^+ \lc q^{-1/\alpha_i} \lc L_i(q)-c_i(q)\rc\rc>1-qx_i\ \forall i\in I, q^{-1/\alpha_i} \lc L_i(q)-c_i(q)\rc>0\ \forall i\not\in I\rc \\
    &=P\lc 1-H_i\lc q, L_i(q) \rc/K_1  >1-qx_i\ \forall i\in I, ,  L_i(q)>c_i(q)\ \forall i\not\in I\rc \\
    &=P\lc H_i\lc q, L_i(q) \rc <qK_1x_i\ \forall i\in I,L_i(q)>c_i(q)\ \forall i\not\in I\rc, 
\end{align*}
where we used (\ref{eqnmargindistposlevy}) and Lemma \ref{lemtechproofs} ii).
Thus,
\begin{equation}\label{eq:Deltaxq1}
D_I(\bmx_I,q) = \Big|K_1^{-1}P(H_i(q,L_i(q)) < qK_1x_i ~\forall i \in I,L_i(q)>c_i(q)\ \forall i\not\in I ) - q R_I(\bmx_I)\Big|.    
\end{equation}
Define
\begin{align}
    h_i: [\delta,\infty)\times \R^{d-1}\to \R, (z,\mathbf{u})\mapsto \Lambda\lc \{ \bmy: y_i  > z,y_j  >u_j\ \forall j\not= i  \} \rc \label{defh_lipsch}. 
\end{align}

We can apply (\ref{eqn1}) to obtain 
\begin{equation}\label{eq:BoundHLambda}
\sup_{z > \delta}\big\vert H_i(q,z)-qh_i\lc z,\bmc_{\setminus i}(q)\rc \big\vert  \leq C_1q^2  \ \forall i\in I.
\end{equation} 
Since we have $H_i\lc q, L_i(q) \rc <qK_1x_i\leq qK_1$ implies $L_i(q)>\delta$, we have by~\eqref{eq:Rhomogeneous},~\eqref{eq:Deltaxq1} and~\eqref{eq:BoundHLambda}
\begin{align}\notag
&D_I(\bmx_I,q) = \big\vert K_1^{-1} P\lc H_i\lc q, L_i(q) \rc <qK_1x_i\ \forall i\in I,L_i(q)>c_i(q)\ \forall i\not\in I \rc-q R_I(\bmx_I) \big\vert 
\\ \notag
& = K_1^{-1}\big| P\lc H_i\lc q, L_i(q) \rc <qK_1x_i\ \forall i\in I, L_i(q)>c_i(q)\ \forall i\not\in I \rc-q R_I(K_1\bmx_I) \Big|
\\
&\overset{\star}{\leq} \tfrac{1}{K_1}\sup_{|\eps_i|\leq C_1 q, i\in I} \Big|  P\lc 
h_i(L_i(q),\bmc_{\setminus i}(q)) +\epsilon_i < K_1x_i \ \forall i\in I, \g L_{\setminus I}(q) > \bmc_{\setminus I}(q)  \rc - q R_I(K_1\bmx_I) \Big|  \label{eq:intermediate}
\end{align}
where $\star$ is validated as follows. 
Let $\g A, \g B$ be random vectors, $\bmx$ a deterministic vector and assume that we know $|A_i -B_i| \leq \eps_i$ a.s. for all $i\in I$ Then 
\[
\{B_i + \eps_i \leq x_i\ \forall i\in I\} \subseteq \{ A_i \leq x_i\ \forall i\in I\} \subseteq \{B_i - \eps_i \leq x_i\ \forall i\in I\}.
\]
Thus 
\[
P(B_i + \eps_i \leq x_i\ \forall i\in I) \leq P(A_i \leq x_i\ \forall i\in I) \leq P(B_i - \eps_i \leq x_i\ \forall i\in I). 
\]
Apply this with $A_i = H_i(q,L_i(q))$, $B_i = \Lambda(\{\bmy: y_i > L_i(q), y_j > c_j(q)\})$. Further, note that for a constant $c$ and $u \in [a,b]$, $|u-c| \leq \max (|c-a|, |c-b|)$. Set $c = qR_I(K_1\bmx_I)$, $a = P(B_i + \eps_i \leq x_i \forall i)$, $b = \{B_i - \eps_i \leq x_i \forall i\}$. Then $\star$ follows, and we in fact see that the supremum is attained for either $\eps_i = C_1 q\ \forall i\in I$ or $\eps_i = - C_1q\ \forall i\in I$.
}

Next, note that for any $z \geq \delta$ and any is a smooth function $s_\delta$ which has support $B_\infty^\complement(\delta/2)$ and $s_\delta(\bmx)=1$ for all $\bmx\in B^\complement_\infty(\delta)$
\[
h_i(x,\mathbf{u}) =\int_{\{ x_i  > z, \bmx_{\setminus i} > \mathbf{u}_{\setminus i}\} } \lambda(\bmx,\bag,\Theta) \rmd\bmx\\
=\int_{ \{ x_i  > z,\bmx_{\setminus i} > \mathbf{u}_{\setminus i} \}} s_\delta(\bmx) \lambda(\bmx,\bag,\Theta) \rmd\bmx.
\] 
Thus, $h_i$ is Lipschitz w.r.t.\ the $\Vert\cdot\Vert_1$ norm for some Lipschitz constant $A_i$ and has uniformly bounded derivatives. To see this, note that we can use the inclusion exclusion formula for switching from a distribution to a survival function to write $h_i$ as a weighted sum of terms of the form $W(\mathbf{u}):=\int_{\{ \bmx\leq \mathbf{u}\}} s_\delta(\bmx) \lambda(\bmx,\bag,\Theta) \rmd\bmx$, where $\mathbf{u}\in (-\infty,\infty]^d$. Thus, $W(\mathbf{u})=F_{ V\setminus I}(\mathbf{u}_{ V\setminus I})$, where $F_I$ is defined as in Lemma \ref{lemdiffintegrateddens} with $I:=\{i\in V\mid u_i=\infty\}$. Thus, as we have already shown in Corollary \ref{cordensregHR} that $s_\delta(\bmx) \lambda(\bmx,\bag,\Theta)$ satisfies the assumption of Lemma \ref{lemdiffintegrateddens}, each $W(\mathbf{u})$ is continuously differentiable with uniformly bounded first derivative, which implies their Lipschitz continuity. The Lipschitz continuity of $h_i$ now immediately follows.

For the following argument, set $C_2 := \max_i A_i$ so that every $h_i$ is Lipschitz with constant $\leq C_2$. Then $|h_i(L_i(q),\bmc_{\setminus i}(q)) - h_i(L_i(q),\bm 0)| \leq C_2 \|\bmc(q)\|_1$. Continuing from~\eqref{eq:intermediate} and using the same arguments as for the verification of $\star$ we find, setting $C(q) := C_1q+C_2\|\bmc(q)\|_1$.  
\begin{align*}
&D_I(\bmx_I,q)
\\
\leq& 
\tfrac{1}{K_1}\sup_{\bm\eps: \|\bm \eps\|_\infty \leq C(q)}
P\Big( 
h_i(L_i(q),\bm 0) + \eps_i < K_1x_i \ \forall i\in I, \g L_{\setminus I}(q) > \bmc_{\setminus I}(q)   \Big) 
- q R_I(K_1\bmx_I) \Big|
\\
=& 
\tfrac{1}{K_1}\sup_{\bm\eps: \|\bm \eps\|_\infty \leq C(q)}
P\Big( 
G_i(L_i(q)) < K_1x_i - \eps_i  \ \forall i\in I, \g L_{\setminus I}(q) > \bmc_{\setminus I}(q)   \Big)
- q R_I(K_1\bmx_I) \Big|
\\
\leq&  \tfrac{1}{K_1}\sup_{\bm\eps: \|\bm \eps\|_\infty \leq C(q)} \Big|  P\Big( 
G_i(L_i(q)) < K_1x_i - \eps_i  \ \forall i\in I, \g L_{\setminus I}(q) > \bmc_{\setminus I}(q)  \Big) - q R_I(K_1\bmx_I) \Big|
\end{align*}
We now distinguish two cases. If for some $i\in I$ we have $K_1x_i - \eps_i \leq 0$, then, noting that $P(G_i(L_i(q)) < 0) = 0, i \in  V$,
\begin{align*}
&\Big|  P\Big( 
G_i(L_i(q)) < K_1x_i - \eps_i  \ \forall i\in I,   \g L_{\setminus I}(q) > \bmc_{\setminus I}(q)  \Big) - q R_I(K_1\bmx_I) \Big| \\
=~& q R_I(K_1\bmx_I) \leq q\sum_{i\in I}K_1x_i \leq dK_1q C(q) 
\end{align*}
since $K_1x_i - \eps_i \leq 0$ implies $K_1x_i \leq \eps_i$, and since $R_I$ is Lipschitz continuous with $R_I(\bm 0)= 0$. It thus remains to consider the supremum with the additional constraint $\eps_i < K_1x_i,~i \in [|I|]$. Choose $q_2$ sufficiently small so that for $q \leq q_2$ we have $C(q) \leq 1/2$, which is possible by the the bound on $c_i(q)$ in Lemma~\ref{lemtechproofs} ii). Then by the choice of $\delta$ we have for any $\eps_i$ with $\eps_i < K_1 (x_I)_i, |\eps_i| \leq C(q), i \in [|I|] $ 
\begin{align*}
&\Big|P\Big( 
G_i(L_i(q)) < K_1x_i - \eps_i  \ \forall i\in I,  \g L_{\setminus I}(q) > \bmc_{\setminus I}(q) \Big) - qR_I(K_1\bmx_I - \bm \eps)\Big|
\\
&=
\Big|P\Big( 
L_i(q) > G_i^{-1}(K_1x_i - \eps_i)  \ \forall i\in I, \g L_{\setminus I}(q) > \bmc_{\setminus I}(q) \Big) \\
&\quad \quad \pm q\Lambda\Big(\big\{ \bmy \big|
y_i > G_i^{-1}(K_1x_i - \eps_i)  \ \forall i\in I, \bmy_{\setminus I}> \bmc_{\setminus I}(q) \big\} \Big)
- qR_I(K_1\bmx_I - \bm \eps)\Big| 
\\
&\leq C_1q^2 + q\Big| \Lambda\Big( 
\big\{ \bmy \big|
y_i > G_i^{-1}(K_1x_i - \eps_i)  \ \forall i\in I, \bmy_{\setminus I}> \bmc_{\setminus I}(q)\big\}  \Big) - R_I(K_1\bmx_I - \bm \eps)\Big|
\end{align*}
by~\eqref{eq:RelateLitoR}, which can be applied since $G_i^{-1}(\cdot)$ are strictly decreasing functions and $x_i < 1 $, so $G_i^{-1}(K_1x_i - \eps_i) \geq G_i^{-1}(K_1 + 1/2) = \delta$ by the bound on $\eps_i$ and $C(q)$ and the choice of $\delta$. 
Using (\ref{eq:LamHRLamGen}), and choosing an arbitrary $\underline i\in I$ we obtain
\begin{align*}
&    \Big| \Lambda\Big(\big\{ \bmy \big|
y_i > G_i^{-1}(K_1x_i - \eps_i)  \ \forall i\in I, \bmy_{\setminus I}> \bmc_{\setminus I}(q) \big\} \Big) - R_I(K_1\bmx_I - \bm \eps)\Big|\\
&=    \Big| \Lambda\Big(\big\{ \bmy \big|
y_i > G_i^{-1}(K_1x_i - \eps_i)  \ \forall i\in I, \bmy_{\setminus I}> \bmc_{\setminus I}(q) \big\} \Big)
\\
&\quad\quad - \Lambda\Big(\big\{ 
\bmy \big| y_i > G_i^{-1}(K_1x_i - \eps_i)  \ \forall i\in I \}\Big)\Big|\\
&=\Big| h_{\underline i}\lc G_{\underline i}^{-1}(K_1x_{\underline i} - \eps_{\underline i}), \lc G_i^{-1}(K_1x_i - \eps_i)\id_{\{i\in I\}} +c_i(q)\id_{\{i\not\in I\}} \rc_{\underline i\neq i \in V} \rc \\
&\quad \quad - h_{\underline i}\lc G_{\underline i}^{-1}(K_1x_{\underline i} - \eps_{\underline i}), \lc G_i^{-1}(K_1x_i - \eps_i)\id_{\{i\in I\}}  \rc_{\underline i\neq i \in V} \rc\Big| \\
&\leq A_{\underline i}\sum_{i\not\in I}\vert c_i(q)\vert \leq C(q) 
\end{align*}
due to the Lipschitz continuity of $h_{\underline i}$. Combining this with Lipchitz continuity of $R_I$ [see~\eqref{eq:RisLipshitz}], $K_1\leq 1$ and the earlier arguments we find that 
\[
D_I(\bmx_I,q)  \leq (2 \vee d)qC(q)/K_1.
\]
Combined with the bound on $c_i(q)$ in Lemma~\ref{lemtechproofs} ii) this completes the proof.
\end{proof}

\begin{proof}[Proof of Theorem \ref{thm:secondorder}]
This result is a direct consequence of Lemma~\ref{lem:HRsecorder} after noting that 
\[
R_I(\bmx_I) = \Lambda_{\rm HR}^{(I)}\big(\{\bmy \in [0,\infty)^{|I|}] \mid \bmy \geq \bmx_I \}\big)
\]
using that the $I$ marginal measure of an HR exponent measure on $[0,\infty)^d$ with variogram matrix $\Gamma$ is in turn an HR exponent measure with parameter matrix $\Gamma_{I,I}$; see \citet[][Example 7]{engelkehitz2020}.
\end{proof}

Next, we prove Theorem \ref{thmconsitencygammaestimator}.

\begin{proof}[Proof of Theorem \ref{thmconsitencygammaestimator}]

Let $i,j,m\in  V$, $i\neq j$, $\tilde\bmo\in\{-1,1\}^3$. Since the dimension $d$ is fixed, and so each $\widehat{\Gamma}_{i,j}$ is a linear combination of a fixed, finite number of $\widehat{\Gamma}^{(m,\tilde\bmo)}_{i,j}$, it suffices to establish $\widehat{\Gamma}^{(m,\tilde\bmo)}_{i,j} \Pkonv \Gamma_{i,j}$ for a fixed $ \bmo$ and $m$. To simplify notation, we will assume that $i < j \leq  m$ and set $I = (i,j,m)$. 

Recalling the notation $\g D_s = \g L(s\Delta) - \g L((s-1)\Delta$, conditional on the collection $\g S = \big(\bmo(\g D_{I,t})\big)_{t \in [n]}$, the collection $\{\g D_{I,t} \mid \g D_{I,t} \in \tilde \bmo\}$ is an i.i.d. sample of size $n^{I,\tilde\bmo} $ from the law $P(\g L_I(\Delta) \in \cdot\mid \g L_I(\Delta) \in \tilde \bmo)$. By Lemma~\ref{lem:HRsecorder}, this law satisfies Assumption S1 from~\cite{engelkelalancettevolgushev2022}.
Their Assumption S2 was already verified for multivariate HR Pareto distributions in \citet[][Lemma S4]{engelkelalancettevolgushev2022}.  Combined with \citet[][Proposition S3, Theorem 1]{engelkelalancettevolgushev2022} this yields for any $0 < \zeta < 1\wedge \min_{i \in  V} \alpha_i^{-1}$, provided $\sqrt{\lfloor n^{I,\tilde\bmo} q_{n} \rfloor}/(\log n^{I,\tilde\bmo})^4 > \log n$, 
\begin{equation}\label{eq:rategammahat}
P\Big( \Big| \widehat{\Gamma}^{(m,\tilde\bmo)}_{i,j} -\Gamma_{i,j} \Big| > \bar C a_{n,n^{I,\tilde\bmo},\log n}\Big| \g S = \g s \Big) \leq M d^3 e^{-c(\log n )^2},
\end{equation} 
for some constants $\bar C, M, c > 0$ where
\[
a_{n,n^{I,\tilde\bmo},\lambda} := \big(\tfrac{\lfloor n^{I,\tilde\bmo} q_{n} \rfloor}{n^{I,\tilde\bmo}}\big)^\zeta \log^2(\tfrac{\lfloor n^{I,\tilde\bmo} q_n \rfloor}{n^{I,\tilde\bmo}}) + \tfrac{1+\lambda}{\lfloor n^{I,\tilde\bmo} q_n \rfloor^{1/2}}.
\] 
Fix $\eps > 0$ and let
\[
\mathcal{S}_{n,\eps} = \Big\{\g s \in [\{-1,1\}^3]^n \vert \bar C a_{n,n^{I,\tilde\bmo},\log n} < \eps, \tfrac{\sqrt{\lfloor n^{I,\tilde\bmo} q_{n^{I,\tilde\bmo}} \rfloor}}{(\log n^{I,\tilde\bmo})^4} > \log n \Big\}.
\]
Note that $n^{I,\tilde\bmo} \sim Bin(n,p_{\tilde\bmo})$ for $p_{\tilde\bmo}:=P\lc L_I(\Delta)\in\tilde\bmo\rc$. By Lemma \ref{lem:posorthantprob} $p_{\tilde\bmo}>0$ and so $n^{I,\tilde\bmo}/n \Pkonv p_{\tilde\bmo}>0$. Thus, given our assumptions on $q_n$, $P(\g S \in \mathcal{S}_{n,\eps}) \to 1$. Thus
\begin{align*}
P\Big( \Big| \widehat{\Gamma}^{(m,\tilde\bmo)}_{i,j} -\Gamma_{i,j} \Big| > \eps\Big) = \sum_{\g s} P\Big( \Big| \widehat{\Gamma}^{(m,\tilde\bmo)}_{i,j} -\Gamma_{i,j} \Big| > \eps, \g S = \g s\Big) \leq P(\g S \in \mathcal{S}_{n,\eps}^\complement) \to 0
\end{align*}
since $P\Big( \Big| \widehat{\Gamma}^{(m,\tilde\bmo)}_{i,j} -\Gamma_{i,j} \Big| > \eps, \g S = \g s\Big) = 0$ for $\g s \in \mathcal{S}_{n,\eps}$. This completes the proof. 
\end{proof}

\subsection{Proof of Corollary~\ref{corsparsistentesttheta}}

\begin{proof}[Proof of Corollary \ref{corsparsistentesttheta}]

Recall the following notation from \cite{engelkelalancettevolgushev2022}: denote by $s$ the maximal number of non-zero entries in any given row of $\Theta$, define $\theta_{\min} := \min_{i, \ell: \Theta_{i\ell} \neq 0} |\Theta_{i\ell}|/\Theta_{\ell\ell}$, and set
\begin{align*}
\mu &:= \min_{m,\ell \in V, m \neq \ell} \lambda_{\min} ( \Sigma^{(m)}_{S_{m,\ell}, S_{m,\ell}}),
\\
\kappa &:= \max_{m,\ell \in V, m \neq \ell} \normop{ \Sigma^{(m)}_{S_{m,\ell}^c, S_{m,\ell}}}_\infty,
\\
\vartheta & := \max_{m,\ell \in V, m \neq \ell} \normop{ ( \Sigma^{(m)}_{S_{m,\ell}, S_{m,\ell}})^{-1}}_\infty.
\end{align*}

A close look at the proof of Theorem 2 in \cite{engelkelalancettevolgushev2022} reveals that for any input matrix $\widehat\Gamma$ that satisfies 
\begin{equation}\label{eq:condconseglearn}
\big\| \widehat\Gamma - \Gamma \big\|_\infty < \frac{2}{3} \min\bigg\{
	\frac{\mu}{2s},
	\frac{\eta}{4\vartheta (1 + \kappa \vartheta) s},
	\frac{\theta_{\min} - \vartheta \rho_n}{2 \vartheta (1 + \kappa \vartheta)},
	\frac{\rho_n \eta}{8(1 + \kappa \vartheta)^2}
	\bigg\},
\end{equation}
for a conditionally negative definite matrix $\Gamma$ and corresponding graph $G$ constructed from the sparsity pattern of $\Theta = (-P\Gamma P/2)^{+}$, the estimated edge set from \texttt{EGlearn} coincides with that of $G$.  
Here, we note that we apply the \texttt{EGlearn} algorithm with neighborhood selection where all penalty parameters $\rho$ are equal. It thus remain to show that under the conditions of the present theorem for $\hat \Gamma$ from~\eqref{defestimatorvariogrammatrix}, the bound in~\eqref{eq:condconseglearn} holds with probability going to one.

Under the present assumptions, the quantities $\eta, \mu, \vartheta, \kappa, s, \theta_{\rm min}$ do not depend on $n$ and thus it suffices to show that $\big\| \widehat\Gamma - \Gamma \big\|_\infty \Pkonv 0$, $\theta_{\rm min} > \vartheta \limsup_{n \to \infty}\rho_n$ and $P(\rho_n > 8 (1 + \kappa \vartheta)^2\big\| \widehat\Gamma - \Gamma \big\|_\infty\eta^{-1}) \to 1$. Consistency of $\widehat \Gamma$ follows from Theorem~\ref{thmconsitencygammaestimator}. Moreover, by~\eqref{eq:rategammahat} in the proof of Theorem~\ref{thmconsitencygammaestimator} we obtain
for any $0 < \zeta < 1\wedge \min_{i \in  V} \alpha_i^{-1}$
\begin{equation*}
\Big\| \widehat{\Gamma} -\Gamma \Big\|_\infty = O_P(q_n^{\zeta}(\log q_n)^2 + (nq_n)^{-1/2}).
\end{equation*}
The statements $P(\rho_n > 8 (1 + \kappa \vartheta)^2\big\| \widehat\Gamma - \Gamma \big\|_\infty\eta^{-1}) \to 1$ and $\theta_{\rm min} > \vartheta \limsup_{n \to \infty}\rho_n$ now follow from our assumptions on $\rho_n$. 
\end{proof}

\subsection{Proof of Lemma \ref{lemIsingcorrviaPLM}}

\begin{proof}[Proof of Lemma \ref{lemIsingcorrviaPLM}]
 First note that we have 
\begin{align*}
    &\frac{{\chi}_{i,j}^{(1,1)}+{\chi}_{i,j}^{(-1,-1)}- {\chi}_{i,j}^{(1,-1)}-{\chi}_{i,j}^{(-1,1)}}{{\chi}_{i,j}}\\
    &=\frac{\Lambda^\star \lc \{ \bmx: o(x_i)=o(x_j), \vert x_i\vert  >1, \vert x_j\vert >1\} \rc- \Lambda^\star \lc \{\bmx: o(x_i)\neq o(x_j), \vert x_i\vert >1, \vert x_j\vert >1 \} \rc}{\Lambda^\star \lc \{ \bmx: \vert x_i\vert  >1, \vert x_j\vert >1 \} \rc} \\
    &=\frac{\sum_{\bmo\in\Oo: o_i=o_j}\gamma_\bmo(\Psi)  \Lambda^\star \lc\{ \bmx:  x_i >1, x_j >1 \} \rc- \sum_{\bmo\in\Oo: o_i\neq o_j}\gamma_\bmo(\Psi)\Lambda^\star \lc \{ \bmx:  x_i >1, x_j >1 \} \rc}{\sum_{\bmo\in\Oo}\gamma_\bmo(\Psi)\Lambda^\star \lc \{ \bmx:   x_i  >1,  x_j >1 \} \rc } \\
    &=\frac{\sum_{\bmo\in\Oo: o_i=o_j}\gamma_\bmo(\Psi)  - \sum_{\bmo\in\Oo: o_i\neq o_j}\gamma_\bmo(\Psi)}{2 } 
\end{align*} 
where we have used that $\sum_{\bmo \in \Oo} \gamma_\bmo(\Psi) = 2$ by \eqref{weight_norm}.
To conclude, observe that
$$
\Cov_\Psi(B_i,B_j)
=\e\lk B_iB_j\rk=P\lc B_i
=B_j\rc-P\lc B_j\neq B_j\rc
=\frac{1}{2}\Big(\sum_{\underset{o_i=o_j}{\bmo\in\Oo}}\gamma_\bmo(\Psi)-\sum_{\underset{o_i\neq o_j}{\bmo\in\Oo}}\gamma_\bmo(\Psi)\Big),
$$ 
by the identity~\eqref{ising_equiv}: $\gamma(\Psi) = 2 P(\g B = \bmo)$ and since $\e\lk B_i\rk=0$ for all $\leqd$. 
\end{proof}

\subsection{Proof of Theorem~\ref{thmchiestfixedgraph}}

\begin{proof}[Proof of Theorem~\ref{thmchiestfixedgraph}]

The Ising model on the graph with edges $E$ and parameter vector $\Psi$ forms a canonical, full and regular (in the sense of Definitions 2.33 and 2.35 of \cite{Bedbur2021}) exponential family with natural parameter space $\R^{|E|}$. By \citet[Theorem 2.56]{Bedbur2021}, the map 
\[
\tau: \R^{|E|} \to \tau(\R^{|E|}); \quad \Psi \mapsto (\e_\Psi[B_iB_j])_{(i,j) \in E} 
\]
is a diffeomorphism, i.e. it is continuously differentiable and invertible with continuously differentiable inverse. Since $\tau^{-1}$ is continuous, $\tau(\R^{|E|})$ is an open subset of $\R^{|E|}$. Below, we will prove that
\begin{equation}\label{eq:hatchiconsistent}
\widehat{\chi}^{(o_i,o_j)}_{i,j}\Pkonv \chi^{(o_i,o_j)}_{i,j}, \quad \forall~ i,j \in  V, o_i, o_j \in \{-1,1\}.  
\end{equation}
Since $\chi_{i,j} = \Lambda^\star (\{ \bmx:  x_i >1, x_j >1 \}) > 0$ for all $i,j$, it follows that 
\[
\hat a_{ij} := \frac{\widehat{\chi}_{i,j}^{(1,1)}+\widehat{\chi}_{i,j}^{(-1,-1)}- \widehat{\chi}_{i,j}^{(1,-1)}-\widehat{\chi}_{i,j}^{(-1,1)}}{\widehat{\chi}_{i,j}} \Pkonv \e_\Psi[B_iB_j] \quad \forall (i,j) \in E.
\]
Since $\tau(\R^{|E|})$ is open, it follows that $\hat A := (\hat a_{ij})_{(i,j) \in E} \in \tau(\R^{|E|})$ with probability going to one, and so $\hat \Psi = \tau^{-1}(\hat A)$ with probability going to one. Continuity of $\tau^{-1}$ now implies that $\hat \Psi \Pkonv \Psi$.

It remains to show~\eqref{eq:hatchiconsistent}. Since the structure of the estimators $\widehat{\chi}^{(o_i,o_j)}_{i,j}$ is similar for all combinations of $o_i, o_j$ we will exemplarily discuss $\widehat{\chi}^{-+}_{i,j}$. We have
\begin{align*}
2\widehat{\chi}^{-+}_{i,j}
&= \frac{2}{k} \sum_{t=1}^{n} \mathbf{1} \left\{ 
\widehat{F}_{\Delta_i}(\Delta_{ti}) < \frac{k}{2n},\ 
\widehat{F}_{\Delta_j}(\Delta_{tj}) > 1 - \frac{k}{2n}
\right\} 
\\
&= \frac{2}{k} \sum_{t=1}^{n} \mathbf{1} \left\{ 
\widehat{F}_{-\Delta_i}(-\Delta_{ti}) - \frac{1}{n} > 1 - \frac{k}{2n},\ 
\widehat{F}_{\Delta_j}(\Delta_{tj}) > 1 - \frac{k}{2n}
\right\}
\\
& = \frac{1}{\lfloor k/2 \rfloor} \sum_{t=1}^{n} \mathbf{1} \left\{ 
\widehat{F}_{-\Delta_i}(-\Delta_{ti}) > 1 - \frac{\lfloor k/2 \rfloor}{n},\ 
\widehat{F}_{\Delta_j}(\Delta_{tj}) > 1 - \frac{\lfloor k/2 \rfloor}{n}
\right\} + O(k^{{-1}}).
\end{align*}
Here, $\widehat{F}_{-\Delta_i}$ denotes the empirical cdf of $\{-\Delta_1,\dots,-\Delta_n\}$ and in the middle equation we have used that for $u \in \{X_1,\dots,X_n\}$ we have $\hat F_X(u) = 1 - \hat F_{-X}(-u) + 1/n$. In the last line we have used the fact that $\tfrac1k - \tfrac{1}{k-1} = - \tfrac{1}{k(k-1)}$ and that $\widehat{F}_{-\Delta_i}, \widehat{F}_{\Delta_j}$ are piecewise constant with jumps of size $1/n$ since the $\Delta_i$ have continuous distributions. 

The quantity in the last line is exactly the tail correlation estimator $\hat \chi_{ij}$ from extreme value theory which uses $\lfloor k/2 \rfloor$ extreme observations and is based on the bivariate sample $(-\Delta_{1i}, \Delta_{1j}),\dots,(-\Delta_{ni}, \Delta_{nj})$. Define 
\[
2\chi_{i,j}^{-+}(q) := q^{-1}P\lc 1-F^{-}_{i,\Delta}(-L_i(\Delta)) \leq q/x_i, 1-F^{+}_{j,\Delta}( L_j(\Delta)) \leq q/x_j \rc
\]
where $F_{i,t}^{o_i}(x)=P\lc o_iL_i(t)\leq x  \rc$
Noting that the proposed estimator is the estimator of the $L$-function of the random vector $(o_iL_i(\Delta),o_jL_j(\Delta))$ it is sufficient to show that it is multivariate regularly varying. By \citet[][Theorem 5.1]{kallsentankov2006} we have for all $(u_1,u_2)$ in a neighbourhood of $(1,1)$
\begin{align*}
\Lambda^\star(\{\bmx: x_i > u_1, x_j > u_2\}) &= \lim_{t\downarrow 0} t^{-1} P(1-F_{i,t}^-(-L_i(t)) \leq t/u_1 , 1-F_{j,t}^+(L_j(t)) \leq t/u_2)
\\
&= \lim_{t\downarrow 0} t^{-1} P(1-F_{i,\Delta}^-(-L_i(\Delta)) \leq t/u_1 , 1-F_{j,\Delta}^+(L_j(\Delta)) \leq t/u_2)
\end{align*}
where $F_{i,t}^{o_i}(x)=P\lc o_iL_i(t)\leq x  \rc$ and in the second equality we used that the copula of $(-L_i(t),L_j(t))$ does not depend on $t$ by Corollary~\ref{cortimescalingHRprocess}. Thus $\chi_{i,j}^{-+}(q) \to \chi_{ij}^{-+}$ as $q \downarrow 0$. 

Observing that the sample $(-\Delta_{1i}, \Delta_{1j}),\dots,(-\Delta_{ni}, \Delta_{nj})$ is i.i.d., 
\citet[][Proposition 2 in the online supplement]{engelkevolgushev2022} implies the existence of a universal constant $K$ such that for any $s > 0$ we have
 \[
 P\Big(2\Big|\widehat\chi_{ij}^{-+}-\chi_{ij}^{-+}\big(\lfloor k/2 \rfloor/n\big)\Big| > s\Big) \leq 5\exp\Big(-\frac{3k}{10}\Big\{\frac{s^2}{K^2}\wedge 1\Big\}\Big).
 \]
Pick $s = s_n$ such that $s_n \to 0, s_n^2 k \to \infty$ to conclude that $\widehat\chi_{ij}^{-+}-\chi_{ij}^{-+}\big(\lfloor k/2 \rfloor/n\big) = o_\p(1)$ and combine this with $\chi_{ij}^{-+}\big(\lfloor k/2 \rfloor/n\big) \to \chi_{ij}^{-+}$ to conclude the proof of~\eqref{eq:hatchiconsistent}.

Next, we prove the remaining parts of the Theorem. Note that $P(\hat{G} = G) \to 1$ by Corollary \ref{corsparsistentesttheta}. Thus $P(\widehat{\Psi}(\widehat E) = \widehat{\Psi}(E)) \to 1$. From the first part of the proof we already know that $\widehat{\Psi}(E) \Pkonv \Psi$. Since the map $\Psi \mapsto \gamma(\Psi)$ is continuous, $\hat \gamma = \gamma(\widehat \Psi) \Pkonv \gamma$ follows.

\end{proof}

\subsection{Proof of claim in Example~\ref{exdensfac}}\label{sec:proofexdensfac}

\begin{proof}[Proof of claim in Example~\ref{exdensfac}]  

By equation~\eqref{eq:gammaofac1} in the proof of Theorem~\ref{thmglobalmarkovprop} we have that $\Theta_{1,3} = 0$ together with $L_1\perp L_3\mid L_2$ implies that $\gamma_\bmo$ factorizes as a function of $\bmo$, i.e. we must have $\gamma_\bmo = f_1(\bmo_{V\setminus\{i\}})f_2(\bmo_{V\setminus\{j\}})$ for strictly positive functions $f_1, f_2$. Note that by the normalization constraint in~\eqref{weight_norm}, $g: \bmo \mapsto \gamma_\bmo$ equals twice a positive probability mass function on the product space $\{-1,1\}^3$ and let $\bm B$ denote the corresponding random vector. Factorization of $g$ implies conditional independence $B_1 \perp B_3 \mid B_2$, which further implies that we must have
\[
P\lc B_1 = 1 \mid B_2=1, B_3=1\rc=m_{1,3} = (1-m_{1,3}) = P\lc B_1 = 1\mid B_2 = 1, B_3 = -1\rc.
\]
This can only hold if $m_{1,3} = 1/2$.

\end{proof}

\clearpage

\section{Additional plots}
\label{appaddplots}
\subsection{Tree graphs}
We report additional results on the estimation of tree graphs for smaller sample sizes. Since tree estimation is
an easier problem than estimation of general graph, we consider as quality measure the proportion of correctly recovered trees for the minimum spanning tree algorithms based on the weights $\hat{\Gamma}$ and $\hat{\chi}$, depending on the threshold $q_n\in(0,1)$. We focus here on the symmetric regime and repeat the experiments $50$ times for dimensions $d=5,10,20$ and sample sizes $n=500,750,1000,2000$.
The results are shown in Figure \ref{fig:trees_sim_stud_asym}.
The minimum spanning tree algorithm based on $\Gamma$ consistently outperforms its competitor based on $\chi$ from \cite{engelke2024levygraphicalmodels} across all dimensions and samples sizes. Another advantage seems to be that the minimum spanning tree based on $\Gamma$ is more stable across different thresholds $q_n$, and it works well even for lower thresholds.

\begin{figure}[!htbp]
    \centering

    \begin{subfigure}{0.49\linewidth}
        \centering
        \includegraphics[width=\linewidth]{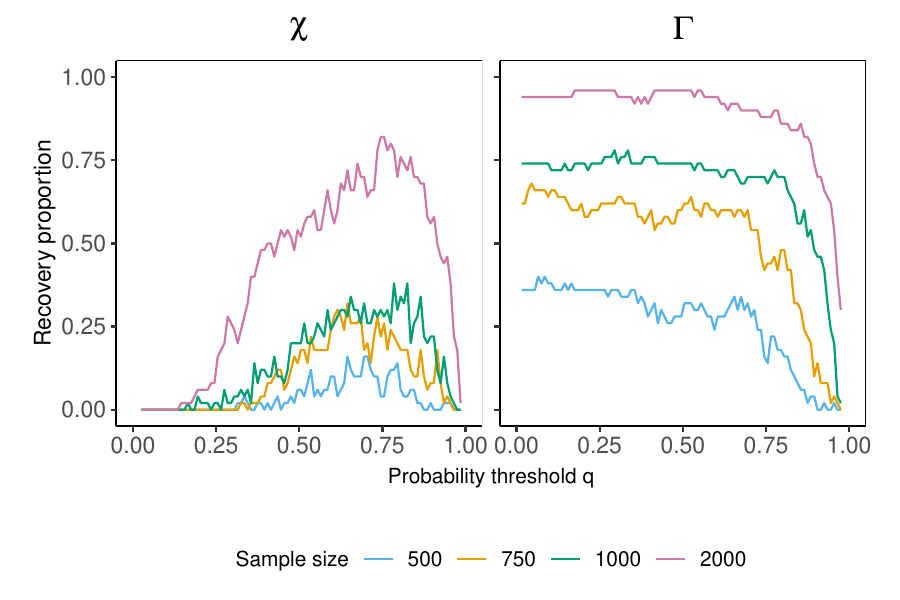}

    \end{subfigure}
    \hfill
        \begin{subfigure}{0.49\linewidth}
        \centering
        \includegraphics[width=\linewidth]{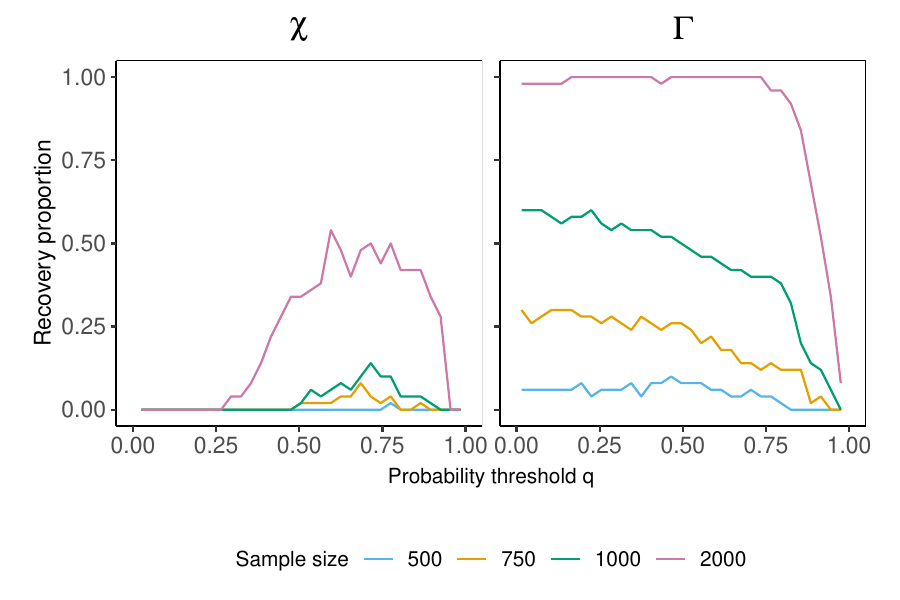}

    \end{subfigure}
    
    \caption{Proportion of correctly recovered trees for the minimums spanning tree algorithms in the asymmetric regime based on $\Gamma$ and $\chi$ for dimensions $d=10$ (left) and $20$ (right) and sample sizes $n=500,750,1000,2000$.}
    \label{fig:trees_sim_stud_asym}
\end{figure}

Additionally, we report the $F_1$-score of the our experiments for tree graphs from Section \ref{secsimulation} in the symmetric regime in Figure \ref{fig:tree_f_1_score_sym} below, mentioning that the conclusions are identical the to conclusions in Section \ref{secsimulation}. 
  \begin{figure}[!htbp]
    \centering
    
        \centering
        \includegraphics[scale=0.65]{Plots_TeX/tree_F1_score_sym_s_size_2000.pdf}

        \centering
        \includegraphics[scale=0.65]{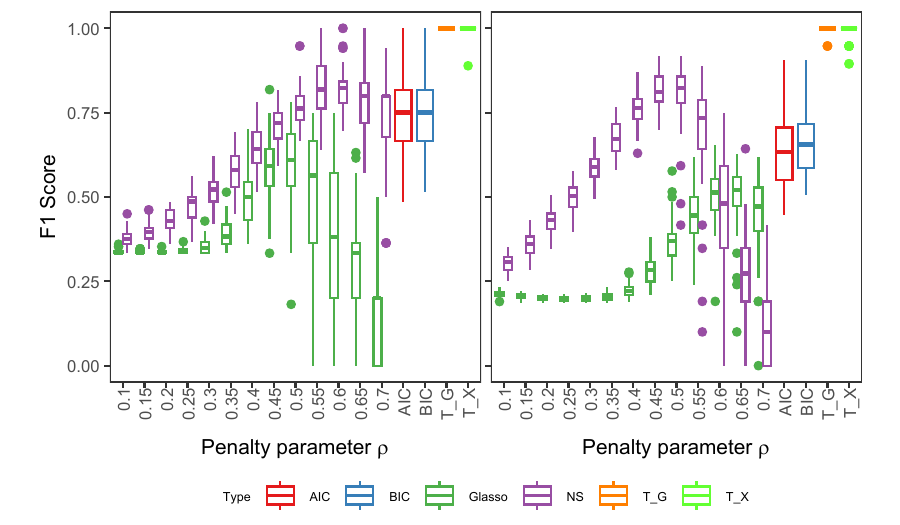}
    \caption{Boxplots of $F_1$ scores for a tree graph in the asymmetric regime for EG-Learn based on neighborhood selection (NS) graphical lasso (Glasso), the minimum spanning trees based on $\Gamma$ and $\chi$ and the two information criteria based penalty parameter selection methods AIC and BIC for various penalty parameters for samples size $n=2000$ (top) and $n=10000$ (bottom) and dimensions $d=10$ (left) and $d=20$ (right).  }
    \label{fig:tree_f_1_score_sym}
\end{figure}

\clearpage

\subsection{General graphs}

We report the $F_1$-score of our experiments for general graphs from Section \ref{secsimulation} in the symmetric regime in Figure \ref{fig:f_1_score_dim_symm}. In general, the performance is slightly worse than in the asymmetric regime, which is discussed in Section \ref{secsimulation}. 

\begin{figure}[!htbp]
    \centering
    
        \centering
        \includegraphics[scale=0.75]{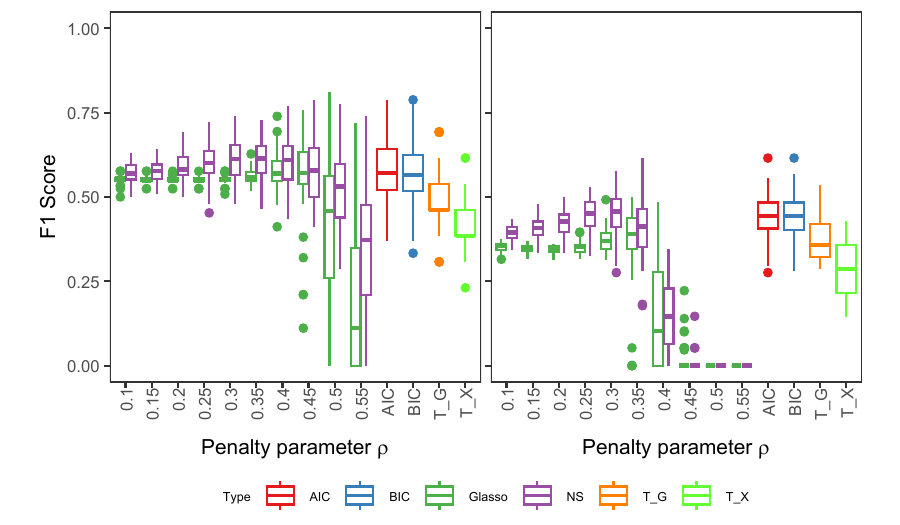}

        \centering
        \includegraphics[scale=0.75]{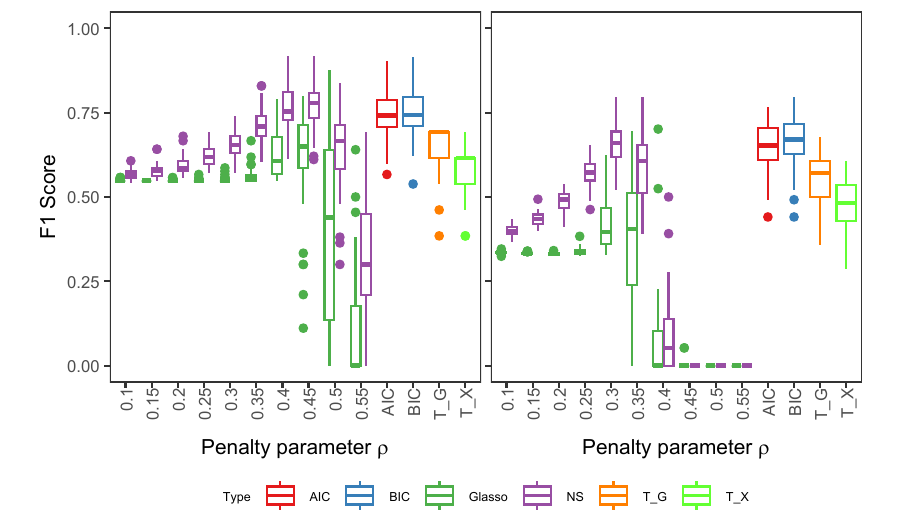}
    \caption{Boxplots of $F_1$ scores for a tree graph in the symmetric regime for EGLearn based on neighborhood selection (NS) graphical lasso (Glasso), the minimum spanning trees based on $\Gamma$ and $\chi$ and the two information criteria based penalty parameter selection methods AIC and BIC for various penalty parameters for samples sizes $n=2000$ (top) and $n=10000$ (bottom) and dimensions $d=10$ (left) and $d=20$ (right).  }
    \label{fig:f_1_score_dim_symm}

\end{figure}

\subsection{Traceplot for Ising weights}

Figure~\ref{fig:traceplot} shows an exemplary traceplot of
the gradient ascent procedure described in Section~\ref{sec:est_ising}.

\begin{figure}
    \centering
    \includegraphics[width=0.6\linewidth]{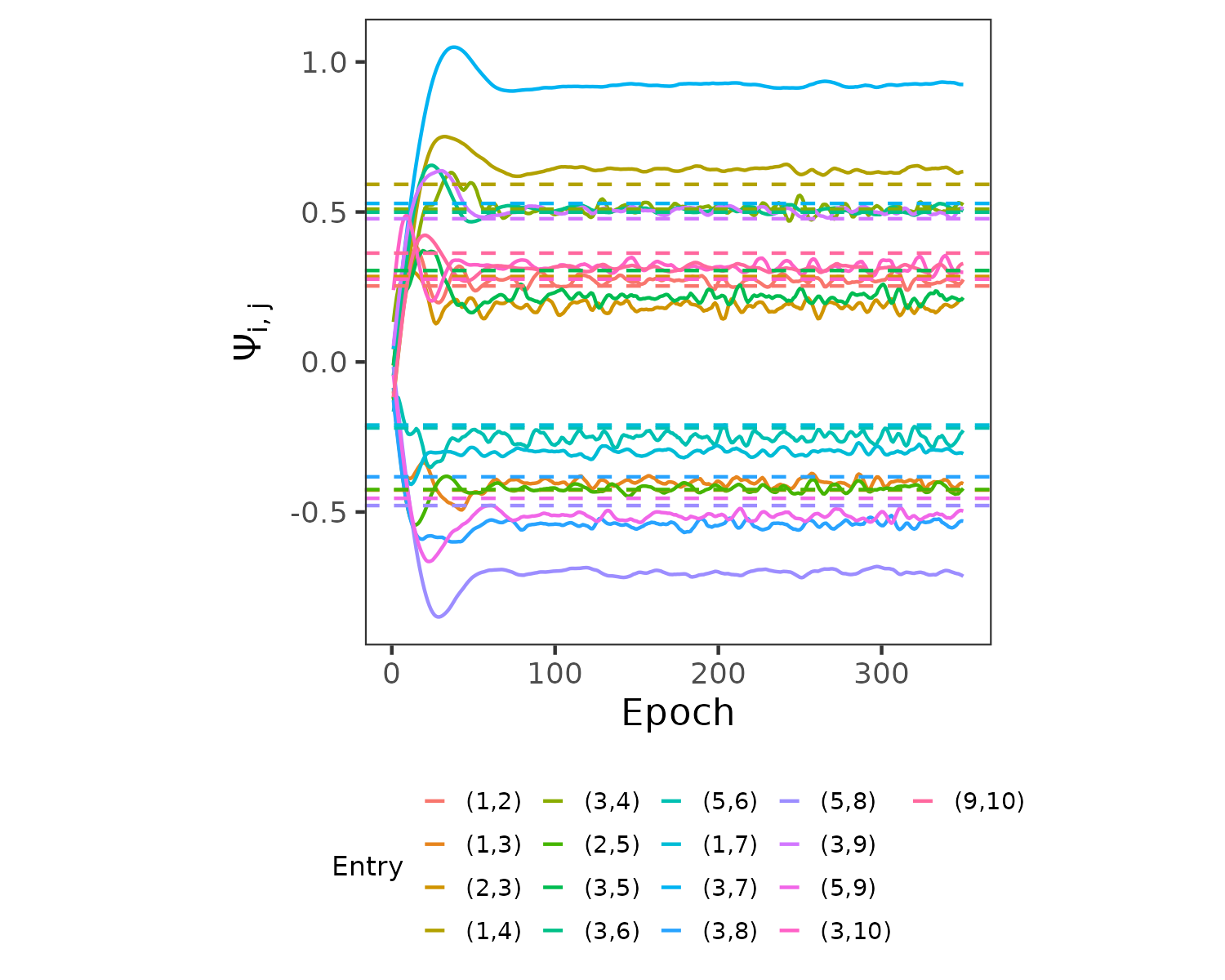}
    \caption{Exemplary traceplot of gradient ascent procedure for the estimation of $\Psi$ in dimension $d=10$.}
    \label{fig:traceplot}
\end{figure}

\end{document}